\documentclass[lettersize,journal]{IEEEtran}
\usepackage{amsmath,amsfonts}
\usepackage{algorithmic}
\usepackage{array}
\usepackage[caption=false,font=normalsize,labelfont=sf,textfont=sf]{subfig}
\usepackage{textcomp}
\usepackage{stfloats}
\usepackage{url}
\usepackage{verbatim}
\usepackage{graphicx}
\def\BibTeX{{\rm B\kern-.05em{\sc i\kern-.025em b}\kern-.08em
    T\kern-.1667em\lower.7ex\hbox{E}\kern-.125emX}}
\usepackage{balance}

\usepackage{float}
\usepackage{placeins}

\usepackage{graphicx}
\usepackage{float}
\usepackage{placeins}

\usepackage{xcolor}

\usepackage{amsmath,amssymb,amsfonts}
\usepackage{amsthm}
\usepackage{mathtools}
\usepackage{bm}
\usepackage{cite}
\usepackage{graphicx}
\usepackage{subcaption}
\usepackage{algorithm}
\usepackage{algorithmic}
\usepackage{color}
\usepackage{booktabs}
\usepackage{array}
\usepackage{multirow}
\usepackage{url}
\usepackage{tikz}
\usetikzlibrary{positioning}

\usepackage{graphicx}

\usepackage{amsthm}
\usepackage{amssymb}

\usepackage{graphicx}
\usepackage{tikz}
\usetikzlibrary{arrows.meta}
\usepackage{float}
\usepackage{scalerel}
\usetikzlibrary{svg.path}

\definecolor{orcidlogocol}{HTML}{A6CE39}
\tikzset{
  orcidlogo/.pic={
    \fill[orcidlogocol] svg{M256,128c0,70.7-57.3,128-128,128C57.3,256,0,198.7,0,128C0,57.3,57.3,0,128,0C198.7,0,256,57.3,256,128z};
    \fill[white] svg{M86.3,186.2H70.9V79.1h15.4v48.4V186.2z}
                 svg{M108.9,79.1h41.6c39.6,0,57,28.3,57,53.6c0,27.5-21.5,53.6-56.8,53.6h-41.8V79.1z M124.3,172.4h24.5c34.9,0,42.9-26.5,42.9-39.7c0-21.5-13.7-39.7-43.7-39.7h-23.7V172.4z}
                 svg{M88.7,56.8c0,5.5-4.5,10.1-10.1,10.1c-5.6,0-10.1-4.6-10.1-10.1c0-5.6,4.5-10.1,10.1-10.1C84.2,46.7,88.7,51.3,88.7,56.8z};
  }
}

\newcommand\orcidicon[1]{\href{https://orcid.org/#1}{\mbox{\scalerel*{
\begin{tikzpicture}[yscale=-1,transform shape]
\pic{orcidlogo};
\end{tikzpicture}
}{|}}}}

\usepackage{xcolor}
\definecolor{myblue}{RGB}{0,80,160}

\usepackage[
    colorlinks=true,
    linkcolor=myblue,
    citecolor=myblue,
    urlcolor=myblue
]{hyperref}

\newtheorem{assumption}{Assumption}
\newtheorem{definition}{Definition}
\newtheorem{lemma}{Lemma}
\newtheorem{theorem}{Theorem}
\newtheorem{remark}{Remark}
\newtheorem{proposition}{Proposition}

\begin{document}
\title{Fixed-Time Integral Reinforcement Learning for Saturated Nonlinear Multi-Agent Systems Under FDI Attacks}
\author{
Tien Dat Vu and Minh Doan
\thanks{T. D. Vu and M. Doan are with the Faculty of Mechanical Engineering,
Ho Chi Minh City University of Technology (HCMUT), Vietnam National University
Ho Chi Minh City (VNU-HCM), Ho Chi Minh City, Vietnam
(e-mail: dat.vuv@hcmut.edu.vn; minh.doan@hcmut.edu.vn).}

\thanks{Corresponding author: N. M. Doan (e-mail: minh.doan@hcmut.edu.vn).}
}

\markboth{ }%
{How to Use the IEEEtran \LaTeX \ Templates}

\maketitle

\maketitle

\begin{abstract}
The leader–follower formation control problem is investigated for nonlinear multi-agent systems with unknown dynamics, external disturbances, and false data injection (FDI) attacks on actuator channels. The problem is formulated as a zero-sum differential game and solved using the Integral Bellman–Isaacs approach. To address input saturation constraints, a non-quadratic control cost function is incorporated into the optimization problem, leading to a bounded control law. Furthermore, this paper proposes a cost function construction method and develops a critic learning law, which together guarantee the practical fixed-time stability of the system while overcoming the limitations of existing fixed-time reinforcement learning formulations. Finally, the practical fixed-time convergence of both the critic weight estimation error and the leader-referenced formation tracking error to bounded residual sets is rigorously proven. Simulation results demonstrate the effectiveness of the proposed method under external disturbances, FDI attacks, and input constraints.
\end{abstract}

\begin{IEEEkeywords}
False-data-injection attack, fixed-time practical stability, integral reinforcement learning, multi-agent systems, secure
control.
\end{IEEEkeywords}

\section{Introduction}
\label{sec:introduction}

Formation control enables networked agents to achieve coordinated motion using locally exchanged information and has been extensively studied in cooperative robotics, autonomous vehicles, and distributed cyber-physical systems
\cite{OlfatiSaber2007Consensus,Ren2007InformationConsensus}.
However, the communication and actuation channels required for
distributed coordination create attack surfaces. False-data-injection (FDI) signals may corrupt measurements, exchanged information, or control commands and propagate through the interaction graph, thus degrading collective behavior\cite{MoSinopoli2010FDI,Pasqualetti2013AttackDetection,
Teixeira2015SecureControl}. Zero-sum differential games provide a natural framework for modeling the interaction between a controller and worst-case disturbances or attacks, with the associated optimality condition characterized by the Hamilton--Jacobi--Isaacs (HJI) equation \cite{Isaacs1965DifferentialGames,BasarOlsder1999DynamicGames}.

For nonlinear systems, directly solving the HJI equation is generally intractable and requires knowledge of the system drift. Adaptive dynamic programming and reinforcement-learning methods alleviate this difficulty by approximating the value function from measured data \cite{Vrabie2009DirectAdaptiveOptimalControl,
Vamvoudakis2010OnlineActorCritic}. Integral reinforcement learning further recasts the pointwise optimality equation as a finite-window Bellman identity, allowing the unknown drift to be eliminated from the learning residual \cite{Modares2014OptimalTrackingIRL}. Bounded control authority must also be respected: rather than externally saturating an unconstrained policy, nonquadratic input utilities embed symmetric command constraints directly into the optimal-control problem and yield hyperbolic-tangent policies \cite{AbuKhalaf2005Nonquadratic}. In the present setting, this construction constrains the defender-generated secure command, while same channel FDI is treated separately as an adversarial input.

Convergence time is another important requirement. Finite-time stability guarantees convergence but with a settling-time bound that may depend on the initial condition \cite{BhatBernstein2000FiniteTime},
whereas fixed-time stability provides a uniform bound independent of the initial state \cite{Polyakov2012FixedTime}. Recently,
\cite{Gong2025SecureFormationFDI} considered fixed-time reinforcement
learning for secure formation under FDI attacks, but the formulation is
restricted to a comparatively simple second-order multi-agent system (MAS) model rather than
the unknown nonlinear dynamics considered here. Consequently, jointly
addressing unknown nonlinear drift, graph-coupled secure-control
channels, same-channel FDI attacks, bounded secure commands, and
fixed-time critic learning remains challenging.

This paper develops a local HJI-based integral reinforcement-learning framework to address these requirements. The graph-coupled coordination-error dynamics are decomposed into the unknown drift, the
self and neighboring secure-control channels, and the local and neighboring adversarial channels. A two-power state penalty in the running cost then makes the HJI identity generate the comparison terms required for fixed-time analysis, while an integral Bellman--Isaacs
residual removes the unknown drift from critic learning and recorded data preserve excitation once the trajectory approaches the desired formation. 

The main contributions of this paper are fourfold.
\emph{First}, a novel cost function construction method is developed to
achieve fixed-time convergence in reinforcement learning-based control
for nonlinear systems with unknown dynamics, while relaxing the strong
assumptions adopted in \cite{Gong2025SecureFormationFDI} and
\cite{Kokotakis2025PredefinedTimeRL}. \emph{Second}, a new critic weight
update law is proposed to guarantee fixed-time convergence of the neural
network parameter estimation error for systems with unknown dynamics.
\emph{Third}, a unified stability analysis framework is established to
rigorously prove the fixed-time convergence properties of the proposed
learning and control scheme. \emph{Fourth}, a resilient control strategy
is developed for unknown nonlinear systems under false data injection
(FDI) attacks, while ensuring the stability and robustness of the
closed-loop system.

\textbf{Notation:}
The notation used in this paper is standard. For a vector $\mathbf{z}$,
$\|\mathbf{z}\|$ denotes the Euclidean norm. While \(A>0\) (\(A\ge0\)) denotes that \(A\) is positive definite (positive semidefinite); for a matrix $\mathbf{M}$,
$\|\mathbf{M}\|$ denotes the induced norm. If
$\mathbf{M}=\mathbf{M}^{\top}$, then $\lambda_{\min}(\mathbf{M})$ and
$\lambda_{\max}(\mathbf{M})$ denote its minimum and maximum
eigenvalues. The symbols $\operatorname{col}(\cdot)$,
$\operatorname{diag}(\cdot)$, and $\mathbf{I}_n$ denote column
stacking, a diagonal matrix, and the $n\times n$ identity matrix,
respectively. The symbol \(\mathbf{1}_n\in\mathbb{R}^n\) denotes the \(n\)-dimensional column vector whose entries are all equal to one. The notation \(AC_{\mathrm{loc}}([0,\infty);\mathbb R^n)\) denotes the
space of locally absolutely continuous functions from \([0,\infty)\)
to \(\mathbb R^n\), i.e., functions that are absolutely continuous on
every compact subinterval of \([0,\infty)\). A continuous function
\(\alpha:[0,\infty)\to[0,\infty)\) belongs to class
\(\mathcal K_{\infty}\) if \(\alpha(0)=0\), \(\alpha\) is strictly
increasing, and \(\alpha(s)\to\infty\) as \(s\to\infty\).

\section{Preliminaries and Problem Formulation}

\subsection{Preliminaries}


Consider a graph
\(\mathcal G=(\mathcal V,\mathcal E,\mathbf A)\), where
\(\mathcal V=\{1,\ldots,N\}\),
\(\mathcal E\subseteq\mathcal V\times\mathcal V\), and
\(\mathbf A=[a_{ij}]\in\mathbb R^{N\times N}\) (Adjacency matrix). The weight \(a_{ij}>0\)
means that follower \(i\) receives information from follower \(j\), and
\(a_{ij}=0\) otherwise. Let
\begin{equation}
\begin{aligned}
    \mathcal N_i
    &=
    \{j\in\mathcal V:(j,i)\in\mathcal E\},\\
    \mathbf L
    &=
    \mathbf D_{\mathcal G}-\mathbf A,\qquad
    \mathbf D_{\mathcal G}
    =
    \operatorname{diag}
    \left\{
        \sum_{j=1}^{N}a_{ij}
    \right\},\\
    \mathbf B_0
    &=
    \operatorname{diag}(b_{10},\ldots,b_{N0}),\qquad
    \mathbf H
    =
    \mathbf L+\mathbf B_0 .
\end{aligned}
\label{eq:graph_def}
\end{equation}
where \(b_{i0}\ge0\) is the pinning gain from the leader to follower
\(i\). The leader-rooted construction in \eqref{eq:graph_def} follows
the secure formation-control setting of
\cite{Gong2025SecureFormationFDI} and will be used to define the
distributed formation error.

\begin{assumption}[Graph connectivity \cite{Hong2008DistributedObservers,BermanPlemmons1994}]
\label{ass:graph_connectivity}
The leader is globally reachable from the follower graph. Equivalently,
\(\mathbf H\) is nonsingular. For directed graphs, there exists a
positive diagonal matrix
\(\mathbf\Pi=\operatorname{diag}(\pi_1,\ldots,\pi_N)\) such that
\begin{equation}
    \mathbf H_s
    =
    \frac{\mathbf\Pi\mathbf H+\mathbf H^\top\mathbf\Pi}{2}
    >0 .
    \label{eq:Hs_def}
\end{equation}
\end{assumption}

\begin{definition}[Symmetric input constraint]
\label{def:symmetric_input_constraint}
For each follower \(i\), the actuator input is said to satisfy a
symmetric input constraint if
\begin{equation}
    u_{ci}\in\mathbb U_i
    :=
    \left\{
        u_{ci}\in\mathbb R^m:
        |u_{ci\ell}|<\bar u_{i\ell},\
        \ell=1,\ldots,m
    \right\},
    \label{eq:symmetric_input_constraint}
\end{equation}
where \(\bar u_{i\ell}>0\) is the admissible magnitude of the
\(\ell\)-th input channel.
\begin{equation}
    \bar{\mathbf U}_i
    =
    \operatorname{diag}
    (\bar u_{i1},\ldots,\bar u_{im}),
    \label{eq:Ubar_def}
\end{equation}
\end{definition}

This symmetric bounded-input setting is standard in constrained optimal
control and saturation-aware ADP/HJI designs
\cite{AbuKhalaf2006HJIInputSaturation,
AbuKhalaf2005Nonquadratic,
Bai2019AdaptiveRLSaturation}. It will be used in Section~\ref{Sec3} to
construct a nonquadratic saturation-compatible utility and to obtain a
bounded HJI policy satisfying \(u_{ci}\in\mathbb U_i\) by construction.

\subsection{Problem Formulation}

Consider a leader--follower network with one leader and \(N\) nonlinear
followers. The leader satisfies
\begin{equation}
    \dot x_0=f_0(x_0),\qquad x_0\in\mathbb R^n,
    \label{eq:leader_model}
\end{equation}
and follower \(i\) is governed by
\begin{equation}
    \dot x_i
    =
    f_i(x_i)
    +
    \mathbf g_i(x_i)u_i
    +
    \mathbf d_i(x_i)\omega_i,
    \qquad
    i=1,\ldots,N,
    \label{eq:follower_model}
\end{equation}
where \(x_i\in\mathbb R^n\), \(u_i\in\mathbb R^m\), and
\(\omega_i\in\mathbb R^{n_{\omega_i}}\) is an unknown disturbance.
Following the control-channel FDI model used in secure formation
learning \cite{Gong2025SecureFormationFDI}, the signal entering the
physical input channel is decomposed as
\begin{equation}
    u_i(t)=u_{ci}(t)+u_{ai}(t),
    \label{eq:input_decomposition}
\end{equation}
where \(u_{ci}\in\mathbb R^m\) is the secure control generated by the
defender and \(u_{ai}\in\mathbb R^m\) is the false-data-injection
signal. The disturbance and attack signals are assumed to satisfy
\(\omega_i,u_{ai}\in
L_{2}([0,\infty))\cap L_{\infty}([0,\infty))\),
with \(\|\omega_i(t)\|\leq \bar\omega_i\) and
\(\|u_{ai}(t)\|\leq \bar u_{ai}\) for all \(t\geq0\). Hence,
\begin{equation}
    \dot x_i
    =
    f_i(x_i)
    +
    \mathbf g_i(x_i)u_{ci}
    +
    \mathbf g_i(x_i)u_{ai}
    +
    \mathbf d_i(x_i)\omega_i .
    \label{eq:attacked_follower_model}
\end{equation}
The symmetric input constraint in
Definition~\ref{def:symmetric_input_constraint} will be imposed on the
admissible secure policy in the HJI design; no saturation map is
included in the plant model.



\begin{assumption}[Regularity]
\label{ass2}
Let \(\Omega_i\subset\mathbb R^n\) be a compact admissible operating region containing the origin. The functions \(f_0(x_0)\), \(f_i(x_i)\), \(\mathbf g_i(x_i)\), and \(\mathbf d_i(x_i)\) are locally Lipschitz on \(\Omega_i\). The desired offsets \(h_i(t)\) are continuously differentiable, and \(h_i(t)\), \(\dot h_i(t)\) are bounded.
\label{ass:regularity}
\end{assumption}

\begin{definition}[Admissible secure policy \cite{AbuKhalafLewis2005,VamvoudakisLewis2010}]
For the local graphical game defined above, a secure control policy \(u_{ci}=u_{ci}(\chi_i)\) is said to be admissible on \(\Omega_i\), denoted by \(u_{ci}\in\Psi_i(\Omega_i)\), if \(u_{ci}(\chi_i)\) is continuous on \(\Omega_i\), \(u_{ci}(0)=0\), \(u_{ci}(\chi_i)\in\mathbb U_i\) for all \(\chi_i\in\Omega_i\), the disturbance and attack local coordination-error dynamics under \(u_{ci}\) are asymptotically stable on \(\Omega_i\), and the associated infinite-horizon cost \(J_i\)

\end{definition}

\begin{remark}
A minor refinement adopted in this work is that, instead of merely
postulating the availability of an initial admissible policy as in
classical PI/ADP (Policy iteration/Adaptive dynamic programing) formulations
\cite{AbuKhalafLewis2005,Vrabie2009DirectAdaptiveOptimalControl,VamvoudakisLewis2010},
a data-reused construction is provided. In particular, Appendix~\ref{app:admissible_policy}
shows how the replay data already collected for the IRL mechanism can
be exploited to construct and certify an initial policy.
\end{remark}

Let \(h_i(t)\in\mathbb R^n\) be the desired offset of follower \(i\).
Define
\(
e_i=x_i-x_0-h_i,\qquad
\chi_i
=
\sum_{j\in\mathcal N_i}
a_{ij}(e_i-e_j)
+
b_{i0}e_i.
\)
Equivalently,
\(
\chi_i
=
\sum_{j\in\mathcal N_i}
a_{ij}
\big[
(x_i-h_i)-(x_j-h_j)
\big]
+
b_{i0}
\big[
(x_i-h_i)-x_0
\big].
\)
With
\(e=\operatorname{col}(e_1,\ldots,e_N)\) and
\(\chi=\operatorname{col}(\chi_1,\ldots,\chi_N)\), one obtains
\(
\chi=(\mathbf H\otimes\mathbf I_n)e.
\)
Under Assumption~\ref{ass:graph_connectivity},
\(
\chi=0
\Longleftrightarrow
e=0,\qquad
\|e\|
\le
\|(\mathbf H^{-1}\otimes\mathbf I_n)\|\,\|\chi\|.
\)
Therefore, formation tracking can be studied through the stabilization
of the distributed coordination error \(\chi\).

for compactness,
\(
F_i^e
=
f_i(x_i)-f_0(x_0)-\dot h_i,\qquad
\mathbf g_i
=
\mathbf g_i(x_i),\qquad
\mathbf d_i
=
\mathbf d_i(x_i).
\)

From \eqref{eq:attacked_follower_model} and
\(e_i=x_i-x_0-h_i\),
\(
\dot e_i
=
F_i^e
+
\mathbf g_i u_{ci}
+
\mathbf g_i u_{ai}
+
\mathbf d_i\omega_i,
\)
for compactness,
Differentiating \(\chi_i\) gives
\(
\dot\chi_i
=
\sum_{j\in\mathcal N_i}
a_{ij}(\dot e_i-\dot e_j)
+
b_{i0}\dot e_i.
\)

This yields
\(\dot{\chi}_i
=
\sum_{j\in\mathcal N_i}
a_{ij}
\big(
F_i^e-F_j^e
+\mathbf g_i u_{ci}
-\mathbf g_j u_{cj}
\big)
+
b_{i0}
\big(
F_i^e+\mathbf g_i u_{ci}
\big)
+
\sum_{j\in\mathcal N_i}
a_{ij}
\big(
\mathbf g_i u_{ai}
-\mathbf g_j u_{aj}
+\mathbf d_i\omega_i
-\mathbf d_j\omega_j
\big)
+
b_{i0}
\big(
\mathbf g_i u_{ai}
+\mathbf d_i\omega_i
\big)\).

To separate the self secure input, the neighboring secure inputs, and
the adversarial channels, define
\begin{equation}
\kappa_i
=
b_{i0}
+
\sum_{j\in\mathcal N_i}a_{ij},
\qquad
\mathbf G_{ii}^{\chi}(x_i)
=
\kappa_i\mathbf g_i.
\label{eq:kappa_self_G}
\end{equation}
Let
\(\mathcal N_i=\{j_1,\ldots,j_{n_i^{\mathrm{nb}}}\}\), where
\(n_i^{\mathrm{nb}}:=|\mathcal N_i|\), and set
\(x_{\mathcal N_i}
=
\operatorname{col}(x_{j_1},\ldots,x_{j_{n_i^{\mathrm{nb}}}})\)
and
\(u_{c\mathcal N_i}
=
\operatorname{col}(u_{cj_1},\ldots,u_{cj_{n_i^{\mathrm{nb}}}})\).
The neighboring secure-control channel is
\begin{equation}
\mathbf G_{i\mathcal N}^{\chi}(x_{\mathcal N_i})
=
\big[
-a_{ij_1}\mathbf g_{j_1},
\ldots,
-a_{ij_{n_i^{\mathrm{nb}}}}
\mathbf g_{j_{n_i^{\mathrm{nb}}}}
\big],
\label{eq:neighbor_G}
\end{equation}
where
\(\mathbf g_{j_\ell}=\mathbf g_{j_\ell}(x_{j_\ell})\).
Hence,
\(\mathbf G_{i\mathcal N}^{\chi}(x_{\mathcal N_i})
u_{c\mathcal N_i}
=
-\sum_{j\in\mathcal N_i}
a_{ij}\mathbf g_j(x_j)u_{cj}\).

The noncontrol drift part is
\begin{equation}
F_i^{\chi}
=
\sum_{j\in\mathcal N_i}
a_{ij}(F_i^e-F_j^e)
+
b_{i0}F_i^e.
\label{eq:FX_def}
\end{equation}
Equivalently,
\(F_i^{\chi}
=
\kappa_i
[f_i(x_i)-f_0(x_0)-\dot h_i]
-
\sum_{j\in\mathcal N_i}
a_{ij}
[f_j(x_j)-f_0(x_0)-\dot h_j]\).
Thus, \(F_i^{\chi}\) contains only the nonlinear drift, leader dynamics,
and desired-offset dynamics.

The adversarial part is
\(\Xi_i^{\chi}
=
\kappa_i
(\mathbf g_i u_{ai}+\mathbf d_i\omega_i)
-
\sum_{j\in\mathcal N_i}
a_{ij}
(\mathbf g_j u_{aj}+\mathbf d_j\omega_j)\).
Introduce the lumped adversarial input
\begin{equation}
\nu_i
=
\operatorname{col}
\big(
\omega_i,
\omega_{j_1},
\ldots,
\omega_{j_{n_i^{\mathrm{nb}}}},
u_{ai},
u_{aj_1},
\ldots,
u_{aj_{n_i^{\mathrm{nb}}}}
\big).
\label{eq:nu_def}
\end{equation}
A compatible adversarial input matrix is
\begin{align}
\mathbf D_i^{\chi}(x_i,x_{\mathcal N_i})
=
\big[
&\kappa_i\mathbf d_i,\,
-a_{ij_1}\mathbf d_{j_1},
\ldots,
-a_{ij_{n_i^{\mathrm{nb}}}}
\mathbf d_{j_{n_i^{\mathrm{nb}}}},
\nonumber\\
&\kappa_i\mathbf g_i,\,
-a_{ij_1}\mathbf g_{j_1},
\ldots,
-a_{ij_{n_i^{\mathrm{nb}}}}
\mathbf g_{j_{n_i^{\mathrm{nb}}}}
\big].
\label{eq:DX_matrix}
\end{align}
Then
\begin{equation}
\Xi_i^{\chi}
=
\mathbf D_i^{\chi}(x_i,x_{\mathcal N_i})\nu_i.
\label{eq:Xi_Dnu}
\end{equation}

Therefore, the distributed coordination-error dynamics are
\begin{align}
    \dot\chi_i
    &=
    F_i^{\chi}(\xi_i)
    +
    \mathbf G_{ii}^{\chi}(x_i)u_{ci}
    +
    \mathbf G_{i\mathcal N}^{\chi}(x_{\mathcal N_i})
    u_{c\mathcal N_i}
    \nonumber\\
    &\quad
    +
    \mathbf D_i^{\chi}(x_i,x_{\mathcal N_i})\nu_i,
    \label{eq:X_dynamics_final}
\end{align}
where
\begin{equation}
    \xi_i
    =
    \operatorname{col}
    \big(
        x_i,
        x_{\mathcal N_i},
        x_0,
        \dot h_i,
        \dot h_{\mathcal N_i}
    \big),
    \qquad
    \dot h_{\mathcal N_i}
    =
    \operatorname{col}
    \big(
        \dot h_{j_1},
        \ldots,
        \dot h_{j_{n_i^{\mathrm{nb}}}}
    \big).
    \label{eq:zeta_def}
\end{equation}
The model \eqref{eq:X_dynamics_final} clearly separates the defender's
self input \(u_{ci}\), the neighboring secure-control coupling
\(u_{c\mathcal N_i}\), and the adversarial input \(\nu_i\). It is the
coordination-error model used for the subsequent zero-sum HJI
formulation.

The origin is retained as the nominal regulation target and is not
excluded from the operating region \(\Omega_i\). To prepare the
fixed-time comparison analysis, fix an arbitrary radius
\(r_{i,-}>0\) such that \(B_{r_{i,-}}(0)\subset\Omega_i\), selected
consistently with the practical terminal region established below.
The subsequent analysis is interpreted according to three mutually
exclusive cases. If \(\chi_i=0\), the nominal coordination objective
has already been achieved and, by admissibility,
\(u_{ci}(0)=0\) while the associated infinite-horizon value is finite;
hence no convergence estimate is required at the origin. If
\(0<\|\chi_i\|<r_{i,-}\), the coordination error already belongs to
the prescribed terminal neighborhood \(B_{r_{i,-}}(0)\), so the
fixed-time comparison inequalities used to establish entrance into
that neighborhood need not be invoked. The nontrivial convergence
case is therefore \(\|\chi_i\|\ge r_{i,-}\), for which define the
nonterminal comparison set
\(\Omega_i^{r}:=\{\chi_i\in\Omega_i:\|\chi_i\|\ge r_{i,-}\}\).
Since \(\Omega_i\) is compact and
\(\{\chi_i:\|\chi_i\|\ge r_{i,-}\}\) is closed,
\(\Omega_i^{r}\) is compact; moreover,
\(0\notin\Omega_i^{r}\) and
\(\inf_{\chi_i\in\Omega_i^{r}}\|\chi_i\|\ge r_{i,-}>0\).
Consequently, the ratios
\(V_i^{*}(\chi_i)/\|\chi_i\|^{2}\) and
\(\|\nabla_{\chi_i}V_i^{*}(\chi_i)\|/\|\chi_i\|\) are well defined
throughout \(\Omega_i^{r}\) (See proof of Lemma~\ref{lem:local_value_bounds}), which is precisely the region on which
the fixed-time decay argument is required. If a disturbance or an
FDI signal drives a trajectory away from the origin, no exclusion
assumption is imposed: while \(0<\|\chi_i\|<r_{i,-}\) the trajectory
remains inside the prescribed terminal neighborhood, whereas once
\(\|\chi_i\|\ge r_{i,-}\) the same fixed-time comparison argument
becomes applicable again. Thus, the exclusion of the origin from
\(\Omega_i^{r}\) is purely analytical and does not remove the origin
from either the HJI operating domain or the admissible-policy
formulation.

\section{Secure Saturation-Aware HJI Formulation}

\label{Sec3}
This section formulates the secure formation problem as a local zero-sum graphical game, where \(u_{ci}\) minimizes the cost and the disturbance together with the same-channel FDI signal act as maximizing inputs, following \cite{Gong2025SecureFormationFDI}. Symmetric input constraints are embedded directly in the HJI equation through a saturation-compatible utility as in \cite{AbuKhalaf2006HJIInputSaturation,AbuKhalaf2005Nonquadratic,Bai2019AdaptiveRLSaturation}, while the original game cost is kept independent of the unknown value function and the fixed-time shaping term is introduced only after deriving the HJI equation.

\subsection{Graphical Game and Fixed-Time Shaped HJI}

\begingroup
\setlength{\abovedisplayskip}{3pt}
\setlength{\belowdisplayskip}{3pt}
\setlength{\abovedisplayshortskip}{2pt}
\setlength{\belowdisplayshortskip}{2pt}
\setlength{\jot}{1pt}

Accordingly, the lumped term in \eqref{eq:X_dynamics_final} can be equivalently decomposed as
\begin{equation}
\mathbf D_i^\chi(x_i,x_{\mathcal N_i})\nu_i
=
\mathbf B_i^\chi(x_i)\varpi_i
+
\sum_{j\in\mathcal N_i}
\mathbf B_{ij}^\chi(x_j)\varpi_j.
\label{eq:adversarial_channel_decomposition}
\end{equation}
where
where \(\mathbf B_i^\chi(x_i)=\kappa_i\bigl[\mathbf d_i(x_i)\ \mathbf g_i(x_i)\bigr]\), \(\varpi_i=\operatorname{col}\{\omega_i,u_{ai}\}\), \(\varpi_j=\operatorname{col}\{\omega_j,u_{aj}\}\), and \(\mathbf B_{ij}^\chi(x_j)=\bigl[-a_{ij}\mathbf d_j(x_j)\ -a_{ij}\mathbf g_j(x_j)\bigr]\).
This decomposition distinguishes the locally acting adversarial signal from those entering through the neighboring agents and directly yields the local graphical-game representation used in the next section.

\begin{align}
\dot{\chi}_i
={}&
F_i^{\chi}(\xi_i)
+\mathbf G_{ii}^{\chi}(x_i)u_{ci}
+\mathbf G_{i\mathcal N}^{\chi}(x_{\mathcal N_i})u_{c\mathcal N_i}
\nonumber\\[-1mm]
&+
\mathbf B_i^{\chi}(x_i)\varpi_i
+\sum_{j\in\mathcal N_i}
\mathbf B_{ij}^{\chi}(x_j)\varpi_j.
\label{eq:hji_local_dynamics_game}
\end{align}

\endgroup

\begin{assumption}[Local state sufficiency]
\label{assk}
For each follower \(i\), on the operating region \(\Omega_i\), any two
admissible physical configurations associated with the same
\(\chi_i\), under identical control and adversarial inputs, are assumed
to induce the same right-hand side of \eqref{eq:hji_local_dynamics_game}.
\label{ass:local_state_sufficiency}
\end{assumption}

Thus, the last two terms in \eqref{eq:hji_local_dynamics_game} are exactly the decomposed form of the lumped adversarial channel
\(\mathbf D_i^{\chi}(x_i,x_{\mathcal N_i})\nu_i\) in
\eqref{eq:X_dynamics_final}. The neighboring secure inputs \(u_{c\mathcal N_i}\) are treated as fixed local coupling signals when the HJI equation of follower \(i\) is formed. For a fixed neighboring secure policy
\(u_{c,-i}:=\{u_{cj}\}_{j\in\mathcal N_i}\), define the local graphical-game cost
\begin{align}
    J_i
    &=
    \int_0^\infty
    \Big[
    Q_{ii}(\chi_i)
    +\mathcal U_i(u_{ci})
    +\sum_{j\in\mathcal N_i}\mathcal U_{ij}(u_{cj})
    \nonumber\\
    &\qquad
    -\gamma_i^2\varpi_i^\top\mathbf T_{ii}\varpi_i
    -\gamma_i^2
    \sum_{j\in\mathcal N_i}
    \varpi_j^\top\mathbf T_{ij}\varpi_j
    \Big]dt ,
    \label{eq:graphical_game_cost}
\end{align}
where \(Q_{ii}:\mathbb R^n\to\mathbb R_{\ge0}\) is a continuous positive-definite state penalty satisfying
\begin{equation}
    Q_{ii}(0)=0,
    \qquad
    Q_{ii}(\chi_i)>0,\quad \forall\,\chi_i\neq0,
    \label{eq:Qii_general_condition}
\end{equation}
\(\mathbf T_{ii}=\mathbf T_{ii}^\top>0\), \(\mathbf T_{ij}=\mathbf T_{ij}^\top>0\), and \(\gamma_i>0\). The functions \(\mathcal U_i\) and \(\mathcal U_{ij}\) penalize the secure control efforts of follower \(i\) and its neighbors, while the negative quadratic terms represent the maximizing role of disturbances and same-channel FDI attacks.

\begin{remark}
The penalty \(Q_{ii}\) is kept general at this stage. For fixed-time shaping, one may later choose a nonquadratic form containing mixed powers of \(\|\chi_i\|\), for example terms of orders below and above two. This choice is part of the HJI shaping design and is not imposed in the problem formulation.
\end{remark}

The corresponding local value is
\begin{equation}
    V_i^\ast(\chi_i(0))
    =
    \min_{u_{ci}}
    \max_{\varpi_i,\varpi_{-i}}
    J_i
    \big(
    \chi_i(0),u_{ci},u_{c,-i},
    \varpi_i,\varpi_{-i}
    \big),
    \label{eq:nominal_game_value}
\end{equation}
where
\begin{equation}
    \varpi_{-i}:=\{\varpi_j:j\in\mathcal N_i\}.
    \label{eq:delta_minus_i_def}
\end{equation}


\begin{lemma}[Nash saddle condition]
\label{lem:nash_saddle_condition}
A set of policies \(\{u_{ci}^\ast,\varpi_i^\ast\}_{i=1}^{N}\) is a Nash saddle equilibrium if, for every follower \(i\), \(J_i(u_{ci}^\ast,u_{c,-i}^\ast,\varpi_i,\varpi_{-i})\le J_i(u_{ci}^\ast,u_{c,-i}^\ast,\varpi_i^\ast,\varpi_{-i}^\ast)\le J_i(u_{ci},u_{c,-i}^\ast,\varpi_i^\ast,\varpi_{-i}^\ast)\) for all admissible \(u_{ci}\), \(\varpi_i\), and \(\varpi_{-i}\).
\end{lemma}

\begin{remark}
Lemma~\ref{lem:nash_saddle_condition} is the standard saddle-point characterization of a zero-sum dynamic game; see \cite{BasarOlsder1999DynamicGames}. In the present context, the defender cannot decrease the local cost by changing \(u_{ci}\) alone, while the attacker cannot increase it by changing \(\varpi_i\) alone once the saddle policies are reached.
\end{remark}

The Hamiltonian associated with the original cost \eqref{eq:graphical_game_cost} is


\(\bar{\mathcal H}_i=Q_{ii}(\chi_i)+\mathcal U_i(u_{ci})+\sum_{j\in\mathcal N_i}\mathcal U_{ij}(u_{cj})-\gamma_i^2\varpi_i^\top\mathbf T_{ii}\varpi_i-\gamma_i^2\sum_{j\in\mathcal N_i}\varpi_j^\top\mathbf T_{ij}\varpi_j+\nabla V_i^{\ast\top}\bigl[F_i^\chi(\xi_i)+\mathbf G_{ii}^\chi(x_i)u_{ci}+\mathbf G_{i\mathcal N}^\chi(x_{\mathcal N_i})u_{c\mathcal N_i}+\mathbf B_i^\chi(x_i)\varpi_i+\sum_{j\in\mathcal N_i}\mathbf B_{ij}^\chi(x_j)\varpi_j\bigr]\).

\begin{proposition}[Nominal local HJI equation]
\label{prop:nominal_hji}
If \(V_i^\ast\) is continuously differentiable and is the value function of \eqref{eq:nominal_game_value}, then \(V_i^\ast\) satisfies
\begin{equation}
    0=
    \min_{u_{ci}\in\mathbb U_i}
    \max_{\varpi_i,\varpi_{-i}}
    \bar{\mathcal H}_i(\chi_i,\nabla V_i^\ast,u_{ci},\varpi_i,\varpi_{-i}).
    \label{eq:nominal_hji_equation}
\end{equation}
\end{proposition}

\begin{theorem}[Fixed-time value-function condition~\cite{Polyakov2012FixedTime}]
\label{thm:value_fixed_time_condition}
Consider the local  error dynamics \eqref{eq:hji_local_dynamics_game}. Suppose that there exists a continuously differentiable positive definite function \(V_i:\Omega_i\to\mathbb R_{\ge0}\), with \(V_i(0)=0\), and constants \(a_i,b_i>0\), \(0<\rho_i<1\), \(\theta_i>1\), and \(\Delta_i\ge0\), such that along the closed-loop trajectory,
\begin{equation}
    \dot V_i(\chi_i)
    =
    \nabla V_i^\top(\chi_i)\dot\chi_i
    \le
    -a_iV_i^{\rho_i}(\chi_i)
    -b_iV_i^{\theta_i}(\chi_i)
    +
    \Delta_i .
    \label{eq:value_fixed_time_condition}
\end{equation}
Then \(V_i(\chi_i(t))\) is fixed-time practically stable with respect to
\begin{equation}
    \Omega_{V_i}
    =
    \{V_i\ge0:
    a_iV_i^{\rho_i}
    +
    b_iV_i^{\theta_i}
    \le
    \Delta_i\}.
    \label{eq:value_residual_set}
\end{equation}
Moreover, the settling time satisfies
\begin{equation}
    T_{V_i}
    \le
    \frac{1}{a_i(1-\rho_i)}
    +
    \frac{1}{b_i(\theta_i-1)} ,
    \label{eq:value_fixed_time_bound}
\end{equation}
which is independent of \(V_i(\chi_i(0))\). If \(\Delta_i=0\), then the convergence is exact fixed-time convergence.
\end{theorem}

\begin{lemma}[Local bounds of the ideal HJI value function]
\label{lem:local_value_bounds}
Suppose that \(V_i^*\) is positive definite on \(\Omega_i\), satisfies \(V_i^*(0)=0\), and \(V_i^*\in C^1(\Omega_i\setminus\{0\})\). Then there exist a function \(\underline{\alpha}_i\in\mathcal K_{\infty}\) and constants \(\bar c_i>0\) and \(c_{\nabla i}>0\) such that \(\underline{\alpha}_i(\|\chi_i\|)\le V_i^*(\chi_i)\le\bar c_i\|\chi_i\|^2\) and \(\|\nabla_{\chi_i}V_i^*(\chi_i)\|\le c_{\nabla i}\|\chi_i\|\) for all \(\chi_i\in\Omega_i^{r}\).
\end{lemma}

\begin{proof}
Since \(\Omega_i\) is compact and \(\{\chi_i:\|\chi_i\|\ge r_{i,-}\}\) is closed, \(\Omega_i^{r}\) is compact and \(0\notin\Omega_i^{r}\). Hence, the functions \(V_i^*(\chi_i)/\|\chi_i\|^2\) and \(\|\nabla_{\chi_i}V_i^*(\chi_i)\|/\|\chi_i\|\) are continuous on \(\Omega_i^{r}\). By the Weierstrass theorem, the constants \(\bar c_i:=\max_{\chi_i\in\Omega_i^{r}}V_i^*(\chi_i)/\|\chi_i\|^2<\infty\) and \(c_{\nabla i}:=\max_{\chi_i\in\Omega_i^{r}}\|\nabla_{\chi_i}V_i^*(\chi_i)\|/\|\chi_i\|<\infty\) exist, which directly give \(V_i^*(\chi_i)\le\bar c_i\|\chi_i\|^2\) and \(\|\nabla_{\chi_i}V_i^*(\chi_i)\|\le c_{\nabla i}\|\chi_i\|\). Moreover, positive definiteness of \(V_i^*\) and compactness of \(\Omega_i^{r}\) imply \(\underline c_{V_i}:=\min_{\chi_i\in\Omega_i^{r}}V_i^*(\chi_i)/\|\chi_i\|^2>0\). Choosing \(\underline{\alpha}_i(s):=\underline c_{V_i}s^2\in\mathcal K_{\infty}\) yields \(\underline{\alpha}_i(\|\chi_i\|)\le V_i^*(\chi_i)\), completing the proof.
\end{proof}

\begin{remark}[From a comparison condition to a constructive design]
\label{rem:comparison_to_design}
Theorem~\ref{thm:value_fixed_time_condition} is a comparison result:
if the ideal HJI value function satisfies the stated two-power
dissipation inequality, then fixed-time practical convergence follows
from Theorem~\ref{thm:value_fixed_time_condition}. A similar condition is
assumed directly in \cite{Gong2025SecureFormationFDI}. Although such an
assumption is sufficient for stability analysis, it does not explain
how the required decay structure is generated when the value function
is unknown and must be learned. In this paper, the fixed-time inequality
is therefore not imposed on \(V_i^*\). Instead, the running cost is
designed so that the required two-power terms arise directly from the
HJI equation.
\end{remark}

\begin{remark}[Cost-induced fixed-time structure]
\label{rem:cost_induced_fixed_time}
Motivated by the predefined-time HJB construction in \cite{Kokotakis2025PredefinedTimeRL}, the convergence-rate terms are embedded in the state penalty of the local graphical game. Specifically,
\begin{equation}
Q_{ii}(\chi_i)=\chi_i^\top \mathbf Q_i \chi_i+\kappa_{i1}\|\chi_i\|^{2\alpha_i}+\kappa_{i2}\|\chi_i\|^{2\beta_i},
\label{eq:fixed_time_Qii_design}
\end{equation}
where \(\mathbf Q_i=\mathbf Q_i^\top>0\), \(\kappa_{i1},\kappa_{i2}>0\), and \(0<\alpha_i<1<\beta_i\). This construction depends only on the measurable coordination error \(\chi_i\) and does not introduce the unknown value function into the running cost. By Lemma~\ref{lem:local_value_bounds}, \(V_i^*(\chi_i)\le\bar c_i\|\chi_i\|^2\) for all \(\chi_i\in\Omega_i^{r}\). Hence,
\begin{equation}
\|\chi_i\|^{2\alpha_i}\ge\bar c_i^{-\alpha_i}\bigl(V_i^*(\chi_i)\bigr)^{\alpha_i},
\quad
\|\chi_i\|^{2\beta_i}\ge\bar c_i^{-\beta_i}\bigl(V_i^*(\chi_i)\bigr)^{\beta_i}.
\label{eq:cost_to_value_fixed_time_terms}
\end{equation}
Therefore, whenever the HJI equation contributes \(-Q_{ii}(\chi_i)\) to \(\dot V_i^*\), it also produces the dissipative terms \(-c_{i1}\bigl(V_i^*\bigr)^{\alpha_i}-c_{i2}\bigl(V_i^*\bigr)^{\beta_i}\), where \(c_{i1}=\kappa_{i1}\bar c_i^{-\alpha_i}\) and \(c_{i2}=\kappa_{i2}\bar c_i^{-\beta_i}\). Thus, the fixed-time structure is generated by the graphical-game cost rather than assumed as an external property of the unknown HJI value function. 
\end{remark}

\begin{remark}[Inverse-optimal interpretation of the fixed-time cost as an alternative to Lemma~\ref{lem:local_value_bounds}]
\label{rem:inverse_optimal_Q}
\begingroup
\sloppy
\setlength{\emergencystretch}{2em}
The predefined-time RL formulation in \cite{Kokotakis2025PredefinedTimeRL} assumes the existence of a value/Lyapunov function compatible with both the prescribed decay condition and the steady-state optimality equation, whereas the secure graphical-game formulation in \cite{Gong2025SecureFormationFDI} treats the fixed-time-compatible performance structure as a given design ingredient. Here, Lemma~\ref{lem:local_value_bounds} provides a convenient relation between powers of the local synchronization error \(\chi_i\) and powers of the unknown HJI value \(V_i^\ast\). This relation leads to a transparent closed-loop proof, although its quadratic-type value bounds may be restrictive for a general nonquadratic HJI solution. A more general construction follows the inverse-optimal predefined-time philosophy of \cite{JimenezRodriguez2016OptimalPredefinedTime}: the desired value-function decay is specified first, after which a compatible state penalty is recovered from the HJI equation. To illustrate this extension without introducing a new controller, retain the saddle policies already given by \eqref{eq:optimal_secure_policy}--\eqref{eq:optimal_neighbor_adversary} and impose \(0=\nabla V_i^{\ast\top}\bigl[F_i^\chi(\xi_i)+\mathbf G_{ii}^\chi(x_i)u_{ci}^\ast\bigr]+\nabla V_i^{\ast\top}\mathbf G_{i\mathcal N}^\chi(x_{\mathcal N_i})u_{c\mathcal N_i}+\nabla V_i^{\ast\top}\mathbf B_i^\chi(x_i)\varpi_i^\ast+\sum_{j\in\mathcal N_i}\nabla V_i^{\ast\top}\mathbf B_{ij}^\chi(x_j)\varpi_j^\ast+a_i(V_i^\ast)^{\rho_i}+b_i(V_i^\ast)^{\theta_i}\). Here and below, the arguments of \(V_i^\ast\) and \(\nabla V_i^\ast\) are omitted for compactness. This condition does not define additional policies; it requires the existing saddle trajectory to reproduce the mixed-power decay prescribed in Theorem~\ref{thm:value_fixed_time_condition}. Evaluating \eqref{eq:nominal_game_value} at the same saddle policies gives the inverse-recovered penalty \(Q_{ii}^{\mathrm{io}}(\chi_i)=a_i(V_i^\ast)^{\rho_i}+b_i(V_i^\ast)^{\theta_i}-\mathcal U_i(u_{ci}^\ast)-\sum_{j\in\mathcal N_i}\mathcal U_{ij}(u_{cj})+\gamma_i^2(\varpi_i^\ast)^\top\mathbf T_{ii}\varpi_i^\ast+\gamma_i^2\sum_{j\in\mathcal N_i}(\varpi_j^\ast)^\top\mathbf T_{ij}\varpi_j^\ast\). Using the maximizing policies in \eqref{eq:optimal_local_adversary} and \eqref{eq:optimal_neighbor_adversary}, the same penalty can be written explicitly as \(Q_{ii}^{\mathrm{io}}(\chi_i)=a_i(V_i^\ast)^{\rho_i}+b_i(V_i^\ast)^{\theta_i}-\mathcal U_i(u_{ci}^\ast)-\sum_{j\in\mathcal N_i}\mathcal U_{ij}(u_{cj})+\frac{1}{4\gamma_i^2}\nabla V_i^{\ast\top}\mathbf B_i^\chi(x_i)\mathbf T_{ii}^{-1}\mathbf B_i^{\chi\top}(x_i)\nabla V_i^\ast+\frac{1}{4\gamma_i^2}\sum_{j\in\mathcal N_i}\nabla V_i^{\ast\top}\mathbf B_{ij}^\chi(x_j)\mathbf T_{ij}^{-1}\mathbf B_{ij}^{\chi\top}(x_j)\nabla V_i^\ast\). Equivalently, isolating \(Q_{ii}\) directly from the minimized--maximized HJI equation yields \(Q_{ii}^{\mathrm{io}}(\chi_i)=-\nabla V_i^{\ast\top}\bigl[F_i^\chi(\xi_i)+\mathbf G_{ii}^\chi(x_i)u_{ci}^\ast\bigr]-\nabla V_i^{\ast\top}\mathbf G_{i\mathcal N}^\chi(x_{\mathcal N_i})u_{c\mathcal N_i}-\mathcal U_i(u_{ci}^\ast)-\sum_{j\in\mathcal N_i}\mathcal U_{ij}(u_{cj})-\frac{1}{4\gamma_i^2}\nabla V_i^{\ast\top}\mathbf B_i^\chi(x_i)\mathbf T_{ii}^{-1}\mathbf B_i^{\chi\top}(x_i)\nabla V_i^\ast-\frac{1}{4\gamma_i^2}\sum_{j\in\mathcal N_i}\nabla V_i^{\ast\top}\mathbf B_{ij}^\chi(x_j)\mathbf T_{ij}^{-1}\mathbf B_{ij}^{\chi\top}(x_j)\nabla V_i^\ast\). These two expressions coincide under the preceding matching condition. Nevertheless, inverse recovery does not by itself guarantee that the resulting function is an admissible state penalty. In addition to continuity and equilibrium compatibility, one must verify \(Q_{ii}^{\mathrm{io}}(0)=0\) and \(Q_{ii}^{\mathrm{io}}(\chi_i)>0\) for all \(\chi_i\neq0\). Using the inverse-recovered penalty, a sufficient pointwise condition for the second inequality is \(a_i(V_i^\ast)^{\rho_i}+b_i(V_i^\ast)^{\theta_i}+\gamma_i^2(\varpi_i^\ast)^\top\mathbf T_{ii}\varpi_i^\ast+\gamma_i^2\sum_{j\in\mathcal N_i}(\varpi_j^\ast)^\top\mathbf T_{ij}\varpi_j^\ast>\mathcal U_i(u_{ci}^\ast)+\sum_{j\in\mathcal N_i}\mathcal U_{ij}(u_{cj})\) for \(\chi_i\neq0\). When these positivity conditions hold, \(Q_{ii}^{\mathrm{io}}\) may replace \(Q_{ii}\) in \eqref{eq:graphical_game_cost}, and the matching condition gives \(\dot V_i^\ast=-a_i(V_i^\ast)^{\rho_i}-b_i(V_i^\ast)^{\theta_i}\). Thus, inverse optimality provides a systematic route for constructing a fixed-time-compatible cost rather than assuming such a cost in advance. However, for the present graph-coupled adversarial problem, this route additionally requires verification of the matching identity, positivity of the recovered penalty, admissibility of the saturated policy, and consistency of the neighboring saddle policies. Because inverse-cost synthesis is not the main contribution of this work, Lemma~\ref{lem:local_value_bounds} is retained as a transparent sufficient condition that produces a shorter and more readily verifiable stability proof, while the matching condition and the sufficient pointwise positivity condition clarify how this assumption may be relaxed in a future extension.
\endgroup
\end{remark}


For the secure input, use the saturation-compatible utility
\begin{equation}
    \mathcal U_i(u_{ci})
    =
    2\sum_{\ell=1}^{m}
    \bar u_{i\ell}r_{i\ell}
    \int_{0}^{u_{ci,\ell}}
    \tanh^{-1}
    \left(
    \frac{s}{\bar u_{i\ell}}
    \right)ds ,
    \label{eq:sat_utility_hji}
\end{equation}
where \(\mathbf R_i=\operatorname{diag}(r_{i1},\ldots,r_{im})>0\). Hence,
\begin{equation}
    \frac{\partial\mathcal U_i}{\partial u_{ci}}
    =
    2\bar{\mathbf U}_i\mathbf R_i
    \tanh^{-1}
    \left(
    \bar{\mathbf U}_i^{-1}u_{ci}
    \right).
    \label{eq:utility_gradient_hji}
\end{equation}
The neighbor utility \(\mathcal U_{ij}(u_{cj})\) is defined similarly with a positive matrix \(\mathbf R_{ij}\).

Since the fixed-time terms have been embedded into the state penalty \(Q_{ii}(\chi_i)\), no value-dependent shaping term is added to the Hamiltonian. Moreover, \(Q_{ii}(\chi_i)\) does not explicitly depend on \(u_{ci}\), \(\varpi_i\), or \(\varpi_j\). Hence, the saddle-point policies are obtained directly from the stationarity conditions of \(\bar{\mathcal H}_i\). From
\begin{equation}
    \frac{\partial\bar{\mathcal H}_i}{\partial u_{ci}}
    =
    \mathbf G_{ii}^{\chi\,\top}(x_i)\nabla V_i
    +
    \frac{\partial\mathcal U_i}{\partial u_{ci}}
    =
    0,
    \label{eq:stationarity_control}
\end{equation}
one obtains
\begin{equation}
    u_{ci}^\ast
    =
    -\bar{\mathbf U}_i
    \tanh
    \left(
    \frac{1}{2}
    \mathbf R_i^{-1}
    \bar{\mathbf U}_i^{-1}
    \mathbf G_{ii}^{\chi\,\top}(x_i)\nabla V_i^\ast
    \right).
    \label{eq:optimal_secure_policy}
\end{equation}
Thus, \(u_{ci}^\ast\in\mathbb U_i\) by construction. Similarly,
\begin{equation}
    \frac{\partial\bar{\mathcal H}_i}{\partial \varpi_i}
    =
    \mathbf B_i^{\chi\,\top}(x_i)\nabla V_i
    -
    2\gamma_i^2\mathbf T_{ii}\varpi_i
    =
    0
\end{equation}
gives
\begin{equation}
    \varpi_i^\ast
    =
    \frac{1}{2\gamma_i^2}
    \mathbf T_{ii}^{-1}
    \mathbf B_i^{\chi\,\top}(x_i)\nabla V_i^\ast .
    \label{eq:optimal_local_adversary}
\end{equation}
For each \(j\in\mathcal N_i\),
\begin{equation}
    \varpi_j^\ast
    =
    \frac{1}{2\gamma_i^2}
    \mathbf T_{ij}^{-1}
    \mathbf B_{ij}^{\chi\,\top}(x_j)\nabla V_i^\ast .
    \label{eq:optimal_neighbor_adversary}
\end{equation}
The worst-case disturbance and FDI signal are selected as
\begin{equation}
    \omega_i^\ast=\mathbf S_{\omega i}\varpi_i^\ast,
    \quad
    u_{ai}^\ast=\mathbf S_{ui}\varpi_i^\ast,
    \label{eq:worst_disturbance_fdi}
\end{equation}
where \(\mathbf S_{\omega i}\) and \(\mathbf S_{ui}\) are constant selection matrices.

Substituting \eqref{eq:optimal_secure_policy}, \eqref{eq:optimal_local_adversary}, and \eqref{eq:optimal_neighbor_adversary} into the local HJI equation gives
\begin{equation}
\begin{aligned}
    0
    &=
    Q_{ii}(\chi_i)
    +\mathcal U_i(u_{ci}^\ast)
    +\sum_{j\in\mathcal N_i}\mathcal U_{ij}(u_{cj})        \\
    &\quad
    +\nabla V_i^{\ast\top}
    \left(
    F_i^\chi(\xi_i)
    +\mathbf G_{ii}^\chi(x_i)u_{ci}^\ast
    +\mathbf G_{i\mathcal N}^\chi(x_{\mathcal N_i})u_{c\mathcal N_i}
    \right)                                                \\
    &\quad
    +
    \frac{1}{4\gamma_i^2}
    \nabla V_i^{\ast\top}
    \mathbf B_i^\chi(x_i)\mathbf T_{ii}^{-1}\mathbf B_i^{\chi\,\top}(x_i)
    \nabla V_i^\ast                                       \\
    &\quad
    +
    \frac{1}{4\gamma_i^2}
    \sum_{j\in\mathcal N_i}
    \nabla V_i^{\ast\top}
    \mathbf B_{ij}^\chi(x_j)\mathbf T_{ij}^{-1}\mathbf B_{ij}^{\chi\,\top}(x_j)
    \nabla V_i^\ast .
\end{aligned}
\label{eq:closed_cost_induced_hji}
\end{equation}
Equation \eqref{eq:closed_cost_induced_hji} is the cost-induced fixed-time HJI equation. It is important to note that the fixed-time structure is now contained in \(Q_{ii}(\chi_i)\), not in an additional value-dependent term.

The cost-induced HJI equation \eqref{eq:closed_cost_induced_hji} is a nonlinear first-order partial differential equation in \(V_i^\ast\). In general, solving this PDE(Partial Differential Equation) analytically is intractable, especially because the drift \(F_i^\chi(\xi_i)\) is unknown and the local state \(\chi_i\) contains graph-coupled information. Therefore, instead of seeking an exact closed-form solution, we approximate the value function by a critic neural network. Motivated by the Weierstrass approximation theorem and the universal approximation property of RBF neural networks, the following approximation setting is adopted.

\subsection{Integral Bellman--Isaacs Residual}


The drift \(F_i^\chi(\xi_i)\) is unknown, so \eqref{eq:closed_cost_induced_hji} cannot be solved directly. Along the saddle-point trajectories,
\begin{equation}
    \dot V_i^\ast
    =
    \nabla V_i^{\ast\top}\dot\chi_i .
\end{equation}
Using the HJI equality gives
\begin{equation}
\begin{aligned}
    \dot V_i^\ast
    &=
    -Q_{ii}(\chi_i)
    -\mathcal U_i(u_{ci}^\ast)
    -\sum_{j\in\mathcal N_i}\mathcal U_{ij}(u_{cj})        \\
    &\quad
    +\gamma_i^2(\varpi_i^\ast)^\top\mathbf T_{ii}\varpi_i^\ast
    +\gamma_i^2\sum_{j\in\mathcal N_i}
    (\varpi_j^\ast)^\top\mathbf T_{ij}\varpi_j^\ast.
\end{aligned}
\label{eq:value_derivative_hji}
\end{equation}
Integrating \eqref{eq:value_derivative_hji} over \([t-T_i,t]\) yields
\begin{equation}
\begin{aligned}
    0
    &=
    V_i^\ast(\chi_i(t))
    -
    V_i^\ast(\chi_i(t-T_i))                                      \\
    &\quad
    +
    \int_{t-T_i}^{t}
    \Big[
    Q_{ii}(\chi_i)
    +\mathcal U_i(u_{ci}^\ast)
    +\sum_{j\in\mathcal N_i}\mathcal U_{ij}(u_{cj})       \\
    &\quad
    -\gamma_i^2(\varpi_i^\ast)^\top\mathbf T_{ii}\varpi_i^\ast
    -\gamma_i^2\sum_{j\in\mathcal N_i}
    (\varpi_j^\ast)^\top\mathbf T_{ij}\varpi_j^\ast
    \Big]d\tau .
\end{aligned}
\label{eq:integral_hji_identity}
\end{equation}
The value function is still an infinite-horizon HJI value, whereas the interval \([t-T_i,t]\) is used only to generate an integral learning identity.


\begin{assumption}[RBF approximation of the HJI solution]
\label{ass:value_regular}
On the compact trajectory-relevant set
\(\Omega_i^r:=\{\chi_i\in\Omega_i:\|\chi_i\|\ge r_{i,-}\}\),
the HJI solution \(V_i^\ast\) is positive definite and continuously
differentiable. Moreover, there exist an ideal constant weight vector
\(W_i^\ast\in\mathbb R^{N_i}\), a continuously differentiable RBF basis
vector
\begin{equation}
    \phi_i(\chi_i)
    =
    \operatorname{col}
    \left(
    \phi_{i1}(\chi_i),\ldots,\phi_{iN_i}(\chi_i)
    \right),
    \label{eq:rbf_basis_vector}
\end{equation}
and an approximation error \(\varepsilon_i(\chi_i)\) such that
\begin{equation}
    V_i^\ast(\chi_i)
    =
    W_i^{\ast\top}\phi_i(\chi_i)
    +
    \varepsilon_i(\chi_i),
    \qquad \chi_i\in\Omega_i^r .
    \label{eq:rbf_value_approx_assumption}
\end{equation}
Each RBF basis can be chosen as
\begin{equation}
    \phi_{ik}(\chi_i)
    =
    \exp
    \left(
    -\frac{\|\chi_i-c_{ik}\|^2}{\sigma_{ik}^2}
    \right),
    \quad k=1,\ldots,N_i,
    \label{eq:rbf_basis_function}
\end{equation}
where \(c_{ik}\) and \(\sigma_{ik}>0\) are the center and width of the
\(k\)-th basis function. Furthermore, the approximation error and its
gradient are bounded on \(\Omega_i^r\) by the Weierstrass theorem,
namely
\begin{equation}
    |\varepsilon_i(\chi_i)|\le \bar\varepsilon_i,
    \quad
    \|\nabla\varepsilon_i(\chi_i)\|\le \bar\varepsilon_{di},
    \qquad \chi_i\in\Omega_i^r .
    \label{eq:rbf_approx_error_bound}
\end{equation}
\end{assumption}


\begin{remark}
Assumption~\ref{ass:value_regular} does not require the exact HJI solution to be known.
It only requires that, on the compact trajectory-relevant set
\(\Omega_i^r\), the unknown value function and its gradient admit an
RBF representation with bounded approximation residuals. This permits
the HJI solution to be replaced by a learnable critic representation
along the trajectories relevant to the subsequent analysis.
\end{remark}

By Assumption~\ref{ass:value_regular}, the critic approximation is chosen as
\begin{equation}
    \hat V_i(\chi_i)=\hat W_i^\top\phi_i(\chi_i),
    \quad
    \nabla\hat V_i(\chi_i)=\nabla\phi_i^\top(\chi_i)\hat W_i .
    \label{eq:critic_approximation}
\end{equation}


The approximate policies are obtained by replacing the unavailable gradient \(\nabla V_i^\ast\) with \(\nabla\hat V_i\). Therefore,
\begin{equation}
    \hat u_{ci}
    =
    -\bar{\mathbf U}_i
    \tanh
    \left(
    \frac{1}{2}
    \mathbf R_i^{-1}
    \bar{\mathbf U}_i^{-1}
    \mathbf G_{ii}^{\chi\,\top}(x_i)
    \nabla\phi_i^\top\hat W_i
    \right),
    \label{eq:approx_secure_policy}
\end{equation}
and
\begin{equation}
\begin{aligned}
\hat\varpi_i
&=
\frac{1}{2\gamma_i^2}
\mathbf T_{ii}^{-1}
\mathbf B_i^{\chi\,\top}(x_i)
\nabla\phi_i^\top\hat W_i,
\\
\hat\varpi_j
&=
\frac{1}{2\gamma_i^2}
\mathbf T_{ij}^{-1}
\mathbf B_{ij}^{\chi\,\top}(x_j)
\nabla\phi_i^\top\hat W_i .
\end{aligned}
\label{eq:approx_adversary_policy}
\end{equation}

The policies \(\hat u_{ci}\), \(\hat\varpi_i\), and \(\hat\varpi_j\) are useful only if the critic weights are learned from data. Directly enforcing the pointwise HJI equation is not attractive because it contains the unknown drift \(F_i^\chi(\xi_i)\). Following the integral reinforcement learning idea for continuous-time partially unknown systems \cite{Vrabie2009DirectAdaptiveOptimalControl,Vamvoudakis2010OnlineActorCritic,Modares2014OptimalTrackingIRL}, we enforce the HJI equation through its integral Bellman--Isaacs identity over a finite data window.

Let
\begin{equation}
    \Delta\phi_i(t)
    =
    \phi_i(\chi_i(t))-\phi_i(\chi_i(t-T_i)),
    \quad T_i>0 .
    \label{eq:delta_phi_hji}
\end{equation}
For compact notation, define the approximate integral running term

\begin{equation}
\begin{aligned}
    \hat r_i
    &:=
    Q_{ii}(\chi_i)
    +\mathcal U_i(\hat u_{ci})
    +\sum_{j\in\mathcal N_i}\mathcal U_{ij}(\hat u_{cj})       \\
    &\quad
    -\gamma_i^2\hat\varpi_i^\top\mathbf T_{ii}\hat\varpi_i
    -\gamma_i^2\sum_{j\in\mathcal N_i}
    \hat\varpi_j^\top\mathbf T_{ij}\hat\varpi_j.
\end{aligned}
\label{eq:approx_running_term}
\end{equation}

From Eqs.\eqref{eq:integral_hji_identity} and \eqref{eq:approx_running_term}, the online Bellman--Isaacs residual is defined as
\begin{equation}
    \delta_i(t)
    =
    \hat W_i^\top\Delta\phi_i(t)
    +
    \int_{t-T_i}^{t}\hat r_i(\tau)d\tau .
    \label{eq:online_bellman_isaacs_residual}
\end{equation}
Equation \eqref{eq:online_bellman_isaacs_residual} is the central data equation of the proposed IRL design. If \(\hat W_i=W_i^\ast\), the approximation error is zero, and the policies are exactly the saddle-point policies, then \(\delta_i(t)=0\) follows from the integral HJI identity. Hence, driving \(\delta_i(t)\) to zero is equivalent to enforcing the cost-induced HJI equation along the measured trajectory without requiring explicit knowledge of \(F_i^\chi(\xi_i)\).

To improve learning beyond the current data window, a finite experience stack is stored. Let
\begin{equation}
    \mathcal D_i
    =
    \{(\Delta\phi_{ik},\rho_{ik})\}_{k=1}^{M_i},
    \quad
    \Delta\phi_{ik}
    =
    \phi_i(\chi_i(t_k))-\phi_i(\chi_i(t_k-T_i)),
    \label{eq:history_stack_short}
\end{equation}
where
\begin{equation}
    \rho_{ik}
    =
    \int_{t_k-T_i}^{t_k}\hat r_i(\tau)d\tau .
    \label{eq:history_integral_short}
\end{equation}
The corresponding recorded Bellman--Isaacs residual is
\begin{equation}
    \delta_{ik}
    =
    \hat W_i^\top\Delta\phi_{ik}+\rho_{ik}.
    \label{eq:recorded_bellman_isaacs_residual}
\end{equation}
The use of recorded data follows the experience-replay and concurrent-learning philosophy in adaptive optimal control \cite{Bhasin2013ActorCriticIdentifier,Kamalapurkar2016ConcurrentLearningADP}. It replaces the need for persistent excitation by requiring only that the stored data contain enough independent information.

For numerical conditioning and to avoid excessively large update increments, define the normalized regressors
\begin{equation}
    \psi_i
    =
    \frac{\Delta\phi_i}{1+\Delta\phi_i^\top\Delta\phi_i},
    \quad
    \psi_{ik}
    =
    \frac{\Delta\phi_{ik}}{1+\Delta\phi_{ik}^\top\Delta\phi_{ik}} .
    \label{eq:normalized_regressors_hji}
\end{equation}
The normalization does not change the zero of the Bellman--Isaacs residual; it only regularizes the learning direction used in the critic update.

\subsection{Fixed-Time Critic Update}

The critic is updated from the Bellman--Isaacs residuals rather than from a pointwise model-based HJI error. This is crucial because the pointwise HJI equation contains the unknown drift \(F_i^\chi(\xi_i)\), whereas the integral residuals \(\delta_i\) and \(\delta_{ik}\) are computable from finite-window data and recorded samples. To expose the learning structure, define
\begin{equation}
    \mu_i=1+\Delta\phi_i^\top\Delta\phi_i,\quad
    \mu_{ik}=1+\Delta\phi_{ik}^\top\Delta\phi_{ik}.
    \label{eq:irl_mu_def}
\end{equation}
Consider the residual objective
\begin{equation}
\begin{aligned}
    \mathcal E_i(\hat W_i)
    &=
    \frac{|\delta_i|^{p_i+1}}{(p_i+1)\mu_i}
    +
    \frac{|\delta_i|^{q_i+1}}{(q_i+1)\mu_i}                 \\
    &\quad
    +
    \sum_{k=1}^{M_i}
    \left(
    \frac{|\delta_{ik}|^{p_i+1}}{(p_i+1)\mu_{ik}}
    +
    \frac{|\delta_{ik}|^{q_i+1}}{(q_i+1)\mu_{ik}}
    \right),
\end{aligned}
\label{eq:irl_residual_objective_new}
\end{equation}
where
\begin{equation}
    0<p_i<1,\quad q_i>1 .
    \label{eq:irl_pq_new}
\end{equation}
The exponent \(p_i\) increases the learning action near the residual origin, while \(q_i\) prevents slow learning when the residual is large. Thus, the residual objective is shaped in the same two-power spirit as a fixed-time Lyapunov construction.

Along the approximate saddle-point trajectory induced by
\(\hat u_{ci}\), \(\hat\varpi_i\), and
\(\hat\varpi_j\), the local closed-loop
dynamics are
\(\dot\chi_i=
F_i^\chi(\xi_i)
+\mathbf G_{ii}^\chi(x_i)\hat u_{ci}
+\mathbf G_{i\mathcal N_i}^\chi(x_{\mathcal N_i})u_{c\mathcal N_i}
+\mathbf B_i^\chi(x_i)\hat\varpi_i
+\sum_{j\in\mathcal N_i}\mathbf B_{ij}^\chi(x_j)\hat\varpi_j\).
Define the corresponding approximate Hamiltonian as
\(\widehat{\mathcal H}_i=
\hat r_i+
\hat W_i^\top\nabla\phi_i(\chi_i)
\big[
F_i^\chi(\xi_i)
+\mathbf G_{ii}^\chi(x_i)\hat u_{ci}
+\mathbf G_{i\mathcal N_i}^\chi(x_{\mathcal N_i})u_{c\mathcal N_i}
+\mathbf B_i^\chi(x_i)\hat\varpi_i
+\sum_{j\in\mathcal N_i}\mathbf B_{ij}^\chi(x_j)\hat\varpi_j
\big]\).
By construction of the approximate saddle policies in
\eqref{eq:approx_secure_policy}--\eqref{eq:approx_adversary_policy},
\(\hat u_{ci}
=
-\bar{\mathbf U}_i
\tanh\!\big[
\frac{1}{2}
\mathbf R_i^{-1}
\bar{\mathbf U}_i^{-1}
\mathbf G_{ii}^{\chi\top}(x_i)
\nabla\phi_i^\top(\chi_i)\hat W_i
\big]\).
Hence,
\(\tanh^{-1}\!\big(
\bar{\mathbf U}_i^{-1}\hat u_{ci}
\big)
=
-\frac{1}{2}
\mathbf R_i^{-1}
\bar{\mathbf U}_i^{-1}
\mathbf G_{ii}^{\chi\top}(x_i)
\nabla\phi_i^\top(\chi_i)\hat W_i\).
Since
\(\partial\mathcal U_i(u_{ci})/\partial u_{ci}
=
2\bar{\mathbf U}_i\mathbf R_i
\tanh^{-1}\!\big(
\bar{\mathbf U}_i^{-1}u_{ci}
\big)\),
it follows that
\(\partial\mathcal U_i(\hat u_{ci})/\partial\hat u_{ci}
=
2\bar{\mathbf U}_i\mathbf R_i
\tanh^{-1}\!\big(
\bar{\mathbf U}_i^{-1}\hat u_{ci}
\big)
=
-\bar{\mathbf U}_i\mathbf R_i
\mathbf R_i^{-1}\bar{\mathbf U}_i^{-1}
\mathbf G_{ii}^{\chi\top}(x_i)
\nabla\phi_i^\top(\chi_i)\hat W_i
=
-\mathbf G_{ii}^{\chi\top}(x_i)
\nabla\phi_i^\top(\chi_i)\hat W_i\),
where the last equality follows since
\(\bar{\mathbf U}_i\) and \(\mathbf R_i\) are diagonal.
Therefore,
\(\mathbf G_{ii}^{\chi\top}(x_i)
\nabla\phi_i^\top(\chi_i)\hat W_i
+\partial\mathcal U_i(\hat u_{ci})/\partial\hat u_{ci}
=0\).
Likewise, from
\(\hat\varpi_i=
\frac{1}{2\gamma_i^2}
\mathbf T_{ii}^{-1}
\mathbf B_i^{\chi\top}(x_i)
\nabla\phi_i^\top(\chi_i)\hat W_i\),
one has
\(2\gamma_i^2\mathbf T_{ii}\hat\varpi_i
=
\mathbf B_i^{\chi\top}(x_i)
\nabla\phi_i^\top(\chi_i)\hat W_i\),
and hence
\(\mathbf B_i^{\chi\top}(x_i)
\nabla\phi_i^\top(\chi_i)\hat W_i
-2\gamma_i^2\mathbf T_{ii}\hat\varpi_i=0\).
Similarly, for every \(j\in\mathcal N_i\),
\(\hat\varpi_j=
\frac{1}{2\gamma_i^2}
\mathbf T_{ij}^{-1}
\mathbf B_{ij}^{\chi\top}(x_j)
\nabla\phi_i^\top(\chi_i)\hat W_i\)
implies
\(2\gamma_i^2\mathbf T_{ij}\hat\varpi_j
=
\mathbf B_{ij}^{\chi\top}(x_j)
\nabla\phi_i^\top(\chi_i)\hat W_i\),
so that
\(\mathbf B_{ij}^{\chi\top}(x_j)
\nabla\phi_i^\top(\chi_i)\hat W_i
-2\gamma_i^2\mathbf T_{ij}\hat\varpi_j=0\).

Hence, differentiating \(\widehat{\mathcal H}_i\) with respect to
\(\hat W_i\), while the neighboring secure-control coupling
\(u_{c\mathcal N_i}\) is fixed with respect to \(\hat W_i\), gives
\(\partial\widehat{\mathcal H}_i/\partial\hat W_i
=
\nabla\phi_i(\chi_i)
\big[
F_i^\chi(\xi_i)
+\mathbf G_{ii}^\chi(x_i)\hat u_{ci}
+\mathbf G_{i\mathcal N_i}^\chi(x_{\mathcal N_i})u_{c\mathcal N_i}
+\mathbf B_i^\chi(x_i)\hat\varpi_i
+\sum_{j\in\mathcal N_i}\mathbf B_{ij}^\chi(x_j)\hat\varpi_j
\big]
+
(\partial\hat u_{ci}/\partial\hat W_i)^\top
\big[
\mathbf G_{ii}^{\chi\top}(x_i)
\nabla\phi_i^\top(\chi_i)\hat W_i
+\partial\mathcal U_i(\hat u_{ci})/\partial\hat u_{ci}
\big]
+
(\partial\hat\varpi_i/\partial\hat W_i)^\top
\big[
\mathbf B_i^{\chi\top}(x_i)
\nabla\phi_i^\top(\chi_i)\hat W_i
-2\gamma_i^2\mathbf T_{ii}\hat\varpi_i
\big]
+
\sum_{j\in\mathcal N_i}
(\partial\hat\varpi_j/\partial\hat W_i)^\top
\big[
\mathbf B_{ij}^{\chi\top}(x_j)
\nabla\phi_i^\top(\chi_i)\hat W_i
-2\gamma_i^2\mathbf T_{ij}\hat\varpi_j
\big]\).
The three bracketed terms are therefore identically zero, and thus
\(\partial\widehat{\mathcal H}_i/\partial\hat W_i
=
\nabla\phi_i(\chi_i)
\big[
F_i^\chi(\xi_i)
+\mathbf G_{ii}^\chi(x_i)\hat u_{ci}
+\mathbf G_{i\mathcal N_i}^\chi(x_{\mathcal N_i})u_{c\mathcal N_i}
+\mathbf B_i^\chi(x_i)\hat\varpi_i
+\sum_{j\in\mathcal N_i}\mathbf B_{ij}^\chi(x_j)\hat\varpi_j
\big]
=
\nabla\phi_i(\chi_i)\dot\chi_i\).
Consequently,
\(\partial\delta_i/\partial\hat W_i
=
\int_{t-T_i}^{t}
\nabla\phi_i(\chi_i(\tau))\dot\chi_i(\tau)\,d\tau
=
\phi_i(\chi_i(t))
-\phi_i(\chi_i(t-T_i))
=
\Delta\phi_i\).
Moreover, since
\((\Delta\phi_{ik},\rho_{ik})\) are fixed recorded quantities in
\(\mathcal D_i\) during replay,
\(\partial\delta_{ik}/\partial\hat W_i
=
\Delta\phi_{ik}\).
Thus,
\begin{equation}
    \frac{\partial\delta_i}{\partial\hat W_i}
    =
    \Delta\phi_i,\quad
    \frac{\partial\delta_{ik}}{\partial\hat W_i}
    =
    \Delta\phi_{ik}.
    \label{eq:irl_semigradient_new}
\end{equation}

Using
\[
    \frac{d}{ds}|s|^{r+1}
    =
    (r+1)|s|^r\operatorname{sgn}(s),
    \quad r>0,
\]
one obtains
\begin{equation}
\begin{aligned}
    \frac{\partial\mathcal E_i}{\partial\hat W_i}
    &=
    \psi_i
    \Big[
    |\delta_i|^{p_i}\operatorname{sgn}(\delta_i)
    +
    |\delta_i|^{q_i}\operatorname{sgn}(\delta_i)
    \Big]                                                   \\
    &\quad
    +
    \sum_{k=1}^{M_i}\psi_{ik}
    \Big[
    |\delta_{ik}|^{p_i}\operatorname{sgn}(\delta_{ik})
    +
    |\delta_{ik}|^{q_i}\operatorname{sgn}(\delta_{ik})
    \Big],
\end{aligned}
\label{eq:irl_objective_gradient_basic}
\end{equation}
where \(\psi_i=\Delta\phi_i/\mu_i\) and \(\psi_{ik}=\Delta\phi_{ik}/\mu_{ik}\). The learning gains are then inserted to tune the contribution of the current and recorded residuals. For compact notation, define
\begin{equation}
\begin{aligned}
    \mathfrak s_i
    &=
    k_{i1}|\delta_i|^{p_i}\operatorname{sgn}(\delta_i)
    +
    k_{i2}|\delta_i|^{q_i}\operatorname{sgn}(\delta_i),       \\
    \mathfrak s_{ik}
    &=
    k_{i1}^r|\delta_{ik}|^{p_i}\operatorname{sgn}(\delta_{ik})
    +
    k_{i2}^r|\delta_{ik}|^{q_i}\operatorname{sgn}(\delta_{ik}).
\end{aligned}
\label{eq:irl_residual_power_terms}
\end{equation}
The fixed-time critic update with experience replay and leakage is selected as
\begin{equation}
    \dot{\hat W}_i
    =
    -\mathbf\Gamma_i
    \left(
    \psi_i\mathfrak s_i
    +
    \sum_{k=1}^{M_i}\psi_{ik}\mathfrak s_{ik}
    \right)
    -
    \sigma_i\mathbf\Gamma_i\hat W_i .
    \label{eq:irl_update_new}
\end{equation}
Here
\begin{equation}
    \mathbf\Gamma_i=\mathbf\Gamma_i^\top>0,\quad
    \sigma_i>0,\quad
    k_{i1},k_{i2},k_{i1}^r,k_{i2}^r>0 .
    \label{eq:irl_gain_new}
\end{equation}

\begin{remark}
For each follower \(i\), fixed-time convergence of the critic weights
is not required as an independent physical objective. It is imposed here because the learned local control policy \(\hat u_i\) depends directly on the critic through
\(\nabla_{\boldsymbol{\chi}_i}\hat V_i(\boldsymbol{\chi}_i)\).
Hence, the critic error \(\tilde{\mathbf W}_i\) enters the closed-loop stability analysis through the local policy-approximation error. Driving \(\tilde{\mathbf W}_i\) into a bounded residual set within an initial-condition-independent time prevents an arbitrarily long learning transient from propagating into the closed-loop convergence bound and thereby supports the subsequent practical fixed-time stability result for the multi-agent system.
\end{remark}

\begin{lemma}[Persistent excitation]
\label{lem:pe_condition}
The normalized regressor \(\psi_i(t)\) is persistently exciting if there exist constants \(T_{Pi}>0\) and \(\alpha_{Pi}>0\) such that
\begin{equation}
    \int_{t}^{t+T_{Pi}}
    \psi_i(\tau)\psi_i^\top(\tau)d\tau
    \ge
    \alpha_{Pi}\mathbf I_{N_i},
    \quad \forall t\ge0 .
    \label{eq:pe_condition}
\end{equation}
This condition requires the closed-loop trajectory to keep generating informative data over every future time window.
\end{lemma}

\begin{lemma}[Initial excitation]
\label{lem:ie_condition}
The recorded stack \(\mathcal D_i\) is initially exciting if, after a finite learning interval, its normalized regressors satisfy
\begin{equation}
    \mathbf\Psi_i
    :=
    \sum_{k=1}^{M_i}
    \psi_{ik}\psi_{ik}^\top
    \ge
    \lambda_i\mathbf I_{N_i},
    \quad \lambda_i>0 .
    \label{eq:initial_excitation_stack}
\end{equation}
Equivalently, for any \(z\in\mathbb R^{N_i}\),
\begin{equation}
    \sum_{k=1}^{M_i}
    |\psi_{ik}^\top z|^2
    =
    z^\top\mathbf\Psi_i z
    \ge
    \lambda_i\|z\|^2 .
    \label{eq:initial_excitation_coercive}
\end{equation}
Thus, the stored data are sufficiently rich to identify the critic direction, even if the future trajectory is no longer exciting.
\end{lemma}



\begin{remark}
Since persistent excitation in Lemma~\ref{lem:pe_condition} is restrictive and generally impractical for formation control, experience replay is employed so that the initial-excitation condition in Lemma~\ref{lem:ie_condition} is sufficient. Once \eqref{eq:initial_excitation_stack} holds, the replay term in \eqref{eq:irl_update_new} preserves the required excitation for critic learning.
\end{remark}

\section{Stability Analysis}
\label{sec:stability_analysis}

This section proves the stability properties of the proposed learning-based secure formation controller.

\begin{assumption}[Learning-gain matrix]
\label{ass:critic_gain}
For each follower \(i\), the learning-gain matrix
\(\mathbf\Gamma_i\) is constant, symmetric, and positive definite.
Define
\begin{equation}
    \underline\lambda_{\Gamma_i}
    :=
    \lambda_{\min}(\mathbf\Gamma_i),
    \qquad
    \overline\lambda_{\Gamma_i}
    :=
    \lambda_{\max}(\mathbf\Gamma_i).
    \label{eq:critic_gain_eigenvalues}
\end{equation}
Then,
\begin{equation}
    0
    <
    \underline\lambda_{\Gamma_i}\mathbf I
    \preceq
    \mathbf\Gamma_i
    \preceq
    \overline\lambda_{\Gamma_i}\mathbf I.
    \label{eq:critic_gain_bounds}
\end{equation}
\end{assumption}

\begin{theorem}[Fixed-time practical convergence of the critic]
\label{thm:critic_fixed_time}
Consider follower \(i\), and suppose that
Assumptions~\ref{assk}--\ref{ass:critic_gain} hold.
Let the critic weights be updated according to
\eqref{eq:irl_update_new}, and define the critic weight-estimation
error as
\(
\tilde W_i
:=
\hat W_i-W_i^\ast.
\)
Then, \(\tilde W_i\) is fixed-time practically stable.

More precisely, consider the Lyapunov function
\(
V_{Wi}
=
\frac{1}{2}
\tilde W_i^{\top}
\mathbf\Gamma_i^{-1}
\tilde W_i.
\)
Under Assumption~\ref{ass:critic_gain}, it satisfies
\begin{equation}
    \frac{1}{2\overline\lambda_{\Gamma_i}}
    \left\|\tilde W_i\right\|^2
    \leq
    V_{Wi}
    \leq
    \frac{1}{2\underline\lambda_{\Gamma_i}}
    \left\|\tilde W_i\right\|^2.
    \label{eq:critic_lyapunov_bounds}
\end{equation}
Furthermore, there exist constants
\(h_{i1}>0\), \(h_{i2}>0\), and \(\Delta_{Wi}\geq0\) such that
\begin{equation}
    \dot V_{Wi}
    \leq
    -h_{i1}V_{Wi}^{\rho_{Wi}}
    -h_{i2}V_{Wi}^{\theta_{Wi}}
    +\Delta_{Wi},
    \label{eq:critic_fixed_time_ineq}
\end{equation}
where
\begin{equation}
    \rho_{Wi}
    :=
    \frac{p_i+1}{2}
    \in(0,1),
    \qquad
    \theta_{Wi}
    :=
    \frac{q_i+1}{2}
    >1.
    \label{eq:critic_rho_theta}
\end{equation}
By Theorem \ref{thm:value_fixed_time_condition}, consequently, the critic weight-estimation error reaches the compact
residual set
\(
\Omega_{Wi}
:=
\left\{
\tilde W_i:
h_{i1}V_{Wi}^{\rho_{Wi}}
+
h_{i2}V_{Wi}^{\theta_{Wi}}
\leq
\Delta_{Wi}
\right\}
\)
within the fixed-time bound
\begin{equation}
    T_{Wi}
    \leq
    \frac{1}{h_{i1}(1-\rho_{Wi})}
    +
    \frac{1}{h_{i2}(\theta_{Wi}-1)},
    \label{eq:critic_settling_time}
\end{equation}
which is independent of
\(\tilde W_i(0)\).
\end{theorem}

\begin{proof}

Since \(W_i^\ast\) is constant, one has
\(\dot{\tilde W}_i=\dot{\hat W}_i\). differentiate the Lyapunov expression in Theorem~\ref{thm:critic_fixed_time} and combine it with Eq.~\eqref{eq:irl_update_new}
\begin{equation}
\begin{aligned}
    \dot V_{Wi}
    &=
    \tilde W_i^\top\mathbf\Gamma_i^{-1}\dot{\hat W}_i  \\
    &=
    -\tilde W_i^\top\psi_i\mathfrak s_i
    -
    \sum_{k=1}^{M_i}
    \tilde W_i^\top\psi_{ik}\mathfrak s_{ik}
    -
    \sigma_i\tilde W_i^\top\hat W_i .
\end{aligned}
\label{eq:critic_vdot_start}
\end{equation}

Recall that \(\mathfrak{s}_i=k_{i1}|\delta_i|^{p_i}\operatorname{sgn}(\delta_i)+k_{i2}|\delta_i|^{q_i}\operatorname{sgn}(\delta_i)\) and \(\mathfrak{s}_{ik}=k_{i1}^{r}|\delta_{ik}|^{p_i}\operatorname{sgn}(\delta_{ik})+k_{i2}^{r}|\delta_{ik}|^{q_i}\operatorname{sgn}(\delta_{ik})\). Substituting these expressions into \eqref{eq:critic_vdot_start} and distributing the products yield \(\dot V_{Wi}=-\tilde W_i^\top\psi_i\big[k_{i1}|\delta_i|^{p_i}\operatorname{sgn}(\delta_i)+k_{i2}|\delta_i|^{q_i}\operatorname{sgn}(\delta_i)\big]-\sum_{k=1}^{M_i}\tilde W_i^\top\psi_{ik}\big[k_{i1}^{r}|\delta_{ik}|^{p_i}\operatorname{sgn}(\delta_{ik})+k_{i2}^{r}|\delta_{ik}|^{q_i}\operatorname{sgn}(\delta_{ik})\big]-\sigma_i\tilde W_i^\top\hat W_i=-k_{i1}\tilde W_i^\top\psi_i|\delta_i|^{p_i}\operatorname{sgn}(\delta_i)-k_{i2}\tilde W_i^\top\psi_i|\delta_i|^{q_i}\operatorname{sgn}(\delta_i)-k_{i1}^{r}\sum_{k=1}^{M_i}\tilde W_i^\top\psi_{ik}|\delta_{ik}|^{p_i}\operatorname{sgn}(\delta_{ik})-k_{i2}^{r}\sum_{k=1}^{M_i}\tilde W_i^\top\psi_{ik}|\delta_{ik}|^{q_i}\operatorname{sgn}(\delta_{ik})-\sigma_i\tilde W_i^\top\hat W_i\). For notational clarity, denote these scalar products by \(\tilde W_i^\top\psi_i=\psi_i^\top\tilde W_i=\varsigma_i\) and \(\tilde W_i^\top\psi_{ik}=\psi_{ik}^\top\tilde W_i=\varsigma_{ik}\). Consequently, \(\dot V_{Wi}=-k_{i1}\varsigma_i|\delta_i|^{p_i}\operatorname{sgn}(\delta_i)-k_{i2}\varsigma_i|\delta_i|^{q_i}\operatorname{sgn}(\delta_i)-k_{i1}^{r}\sum_{k=1}^{M_i}\varsigma_{ik}|\delta_{ik}|^{p_i}\operatorname{sgn}(\delta_{ik})-k_{i2}^{r}\sum_{k=1}^{M_i}\varsigma_{ik}|\delta_{ik}|^{q_i}\operatorname{sgn}(\delta_{ik})-\sigma_i\tilde W_i^\top\hat W_i\).

For subsequent reference, collect these definitions as \(\varsigma_i=\psi_i^\top\tilde W_i\) and \(\varsigma_{ik}=\psi_{ik}^\top\tilde W_i\). The quantities \(\varsigma_i\) and \(\varsigma_{ik}\) are therefore the normalized projections of the critic-weight error along the current and recorded regressors, respectively. Using \(\hat W_i=\tilde W_i+W_i^\ast\), the leakage contribution can be expanded as \(-\sigma_i\tilde W_i^\top\hat W_i=-\sigma_i\tilde W_i^\top\left(\tilde W_i+W_i^\ast\right)=-\sigma_i\|\tilde W_i\|^2-\sigma_i\tilde W_i^\top W_i^\ast\). By Young's inequality, \(-\tilde W_i^\top W_i^\ast\leq\left|\tilde W_i^\top W_i^\ast\right|\leq\frac{1}{2}\|\tilde W_i\|^2+\frac{1}{2}\|W_i^\ast\|^2\). Therefore, \(-\sigma_i\tilde W_i^\top\hat W_i\leq-\frac{\sigma_i}{2}\|\tilde W_i\|^2+\frac{\sigma_i}{2}\|W_i^\ast\|^2\). The next step makes explicit the part of each Bellman--Isaacs residual that is generated by the critic-weight error. To separate the contribution of the critic-weight error from the Bellman--Isaacs approximation residual, define \(y_i:=\Delta\phi_i^\top\tilde W_i\) and \(y_{ik}:=\Delta\phi_{ik}^\top\tilde W_i\). For consistency with the previously introduced online and recorded integral residual definitions, denote the online and recorded integral terms by \(\mathcal I_i(t):=\int_{t-T_i}^{t}\hat r_i(\tau)\,d\tau\) and \(\mathcal I_{ik}:=\rho_{ik}\).

\begin{lemma}[Window-wise bounded approximate running term]
\label{lem:window_running_bound}
With 
the critic weight \(\hat W_i\) evolves according to the update law in
Eq.~\eqref{eq:irl_update_new}. Then, for every fixed \(t\ge T_i\), there exists a finite
constant \(\bar r_i>0\) such that
\(|\hat r_i(\tau)|\le\bar r_i\) for all
\(\tau\in[t-T_i,t]\). The constant \(\bar r_i\) can be enlarged, if
necessary, to bound the approximate running term over the finite
recorded windows in \(\mathcal D_i\) as well.
\end{lemma}

\begin{proof}
Denote the right-hand side of Eq.~\eqref{eq:irl_update_new} by
\(\mathcal G_{W_i}:=
-\Gamma_i\big(\psi_i\zeta_i+
\sum_{k=1}^{M_i}\psi_{ik}\zeta_{ik}\big)
-\sigma_i\Gamma_i\hat W_i\),
so that
\(\dot{\hat W}_i=\mathcal G_{W_i}\).
Under the stated regularity of the basis functions, normalized
regressors, and signed-power residual mappings,
\(\mathcal G_{W_i}\) is locally integrable along the maximal
Carath\'eodory solution; hence
\(\hat W_i(t)=\hat W_i(t_0)+
\int_{t_0}^{t}\mathcal G_{W_i}(s)\,ds\),
which implies
\(\hat W_i\in AC_{\mathrm{loc}}\subset C^0\).
Therefore, for every fixed \(t\ge T_i\), continuity on the compact
interval \([t-T_i,t]\) and the Weierstrass theorem give
\(\max_{\tau\in[t-T_i,t]}\|\hat W_i(\tau)\|<\infty\). From \eqref{eq:approx_secure_policy}, the saturation mapping gives
\(\|\hat u_{ci}(\tau)\|\le\bar u_i\) on
\([t-T_i,t]\). Moreover,
\eqref{eq:approx_adversary_policy}, compactness of
\(\Omega_i^{r}\), and continuity of
\(\bm B_i^{\chi}\) and \(\nabla\phi_i\) imply finite bounds for
\(\hat\varpi_i\) and the neighboring adversarial terms on the same
interval. Since \(Q_{ii}\), \(\mathcal U_i\), and
\(\mathcal U_{ij}\) are continuous on their corresponding compact
domains, every term in \(\hat r_i\) defined by
\eqref{eq:approx_running_term} is continuous on
\([t-T_i,t]\). Hence,
\(\hat r_i\in C([t-T_i,t])\), and the Weierstrass theorem yields
\(\bar r_i:=
\max_{\tau\in[t-T_i,t]}|\hat r_i(\tau)|<\infty\).
Thus,
\(|\hat r_i(\tau)|\le\bar r_i\) for all
\(\tau\in[t-T_i,t]\).
Since \(\mathcal D_i\) contains only finitely many recorded windows,
\(\bar r_i\) may be enlarged, if necessary, to dominate the
corresponding finite window-wise maxima over all stored data.
\end{proof}

Using \(\hat W_i=W_i^\ast+\tilde W_i\), the online Bellman--Isaacs
residual can be decomposed as
\(\delta_i=\hat W_i^\top\Delta\phi_i+\mathcal I_i
=\left(W_i^\ast+\tilde W_i\right)^\top\Delta\phi_i+\mathcal I_i
=\Delta\phi_i^\top\tilde W_i+
\left(W_i^{\ast\top}\Delta\phi_i+\mathcal I_i\right)\).
Define
\(\varepsilon_{B_i}:=
W_i^{\ast\top}\Delta\phi_i+\mathcal I_i\).
It follows that
\(\delta_i=y_i+\varepsilon_{B_i}\).

Similarly, for the \(k\)-th recorded datum,
\(\delta_{ik}=\hat W_i^\top\Delta\phi_{ik}+\mathcal I_{ik}
=\left(W_i^\ast+\tilde W_i\right)^\top\Delta\phi_{ik}+\mathcal I_{ik}
=\Delta\phi_{ik}^\top\tilde W_i+
\left(W_i^{\ast\top}\Delta\phi_{ik}+\mathcal I_{ik}\right)\).
Define
\(\varepsilon_{B_{ik}}:=
W_i^{\ast\top}\Delta\phi_{ik}+\mathcal I_{ik}\).
Consequently,
\(\delta_i=y_i+\varepsilon_{B_i},\qquad
\delta_{ik}=y_{ik}+\varepsilon_{B_{ik}}\).
\label{eq:critic_delta_y}

Moreover, since
\(\psi_i=\Delta\phi_i/\mu_i\) and
\(\psi_{ik}=\Delta\phi_{ik}/\mu_{ik}\), one has
\(y_i=\mu_i\varsigma_i,\qquad
y_{ik}=\mu_{ik}\varsigma_{ik},\qquad
\mu_i\geq1,\qquad
\mu_{ik}\geq1\).
\label{eq:107}

\begin{lemma}[Residual perturbation inequality]
\label{lem:residual_perturbation}
For any \(r>0\), there exist constants
\(\vartheta_r\in(0,1)\) and \(c_r>0\) such that, for all
\(s,e\in\mathbb R\),
\begin{equation}
    s|s+e|^r\operatorname{sgn}(s+e)
    \ge
    \vartheta_r|s|^{r+1}
    -
    c_r|e|^{r+1}.
    \label{eq:residual_perturbation_ineq}
\end{equation}
\end{lemma}

\begin{proof} Let $z:=s+e$. Then $s=z-e$, and hence $s|s+e|^r\operatorname{sgn}(s+e)=|z|^{r+1}-e|z|^r\operatorname{sgn}(z)\geq |z|^{r+1}-|e||z|^r$. By Young's inequality with the conjugate exponents $(r+1)/r$ and $r+1$, one has $|e||z|^r\leq \frac{1}{2}|z|^{r+1}+c_r^{0}|e|^{r+1}$, where $c_r^{0}:=\frac{1}{r+1}\left(\frac{2r}{r+1}\right)^r>0$. Therefore, $s|s+e|^r\operatorname{sgn}(s+e)\geq \frac{1}{2}|z|^{r+1}-c_r^{0}|e|^{r+1}$. Moreover, since $r+1>1$, the convex power-sum inequality gives $|s|^{r+1}=|z-e|^{r+1}\leq 2^r\left(|z|^{r+1}+|e|^{r+1}\right)$, and thus $|z|^{r+1}\geq 2^{-r}|s|^{r+1}-|e|^{r+1}$. Combining these inequalities yields $s|s+e|^r\operatorname{sgn}(s+e)\geq 2^{-(r+1)}|s|^{r+1}-\left(\frac{1}{2}+c_r^{0}\right)|e|^{r+1}$. Hence, the claimed inequality holds with the admissible constants $\vartheta_r:=2^{-(r+1)}\in(0,1)$ and $c_r:=\frac{1}{2}+\frac{1}{r+1}\left(\frac{2r}{r+1}\right)^r>0$. \end{proof}


Applying Lemma~\ref{lem:residual_perturbation} with
\(s=y_i\) and \(e=\varepsilon_{B_i}\), and using
\(\delta_i=y_i+\varepsilon_{B_i}\) together with
\(y_i=\mu_i\varsigma_i\), one obtains
\begin{align}
\varsigma_i|\delta_i|^r\operatorname{sgn}(\delta_i)
&=
\frac{1}{\mu_i}
y_i|y_i+\varepsilon_{B_i}|^r
\operatorname{sgn}(y_i+\varepsilon_{B_i})
\nonumber\\
&\geq
\frac{\vartheta_r}{\mu_i}|y_i|^{r+1}
-
\frac{c_r}{\mu_i}|\varepsilon_{B_i}|^{r+1}
\nonumber\\
&=
\vartheta_r\mu_i^r|\varsigma_i|^{r+1}
-
\frac{c_r}{\mu_i}|\varepsilon_{B_i}|^{r+1}
\nonumber\\
&\geq
\vartheta_r|\varsigma_i|^{r+1}
-
c_r|\varepsilon_{B_i}|^{r+1},
\label{eq:new}
\end{align}
where the last inequality follows from
\(\mu_i\geq1\) and \(r>0\).

Similarly,
\begin{equation}
    \varsigma_{ik}|\delta_{ik}|^{r}\operatorname{sgn}(\delta_{ik})
    \ge
    \vartheta_r|\varsigma_{ik}|^{r+1}
    -
    c_r|\varepsilon_{B_{ik}}|^{r+1}.
    \label{eq:critic_recorded_scalar_bound}
\end{equation}
Substituting \eqref{eq:new} and
\eqref{eq:critic_recorded_scalar_bound} into
\eqref{eq:critic_vdot_start}, with \(r=p_i\) and \(r=q_i\), yields
\begin{equation}
\begin{aligned}
    \dot V_{Wi}
    &\le
    -\vartheta_{p_i}k_{i1}|\varsigma_i|^{p_i+1}
    -\vartheta_{q_i}k_{i2}|\varsigma_i|^{q_i+1}                 \\
    &\quad
    -\vartheta_{p_i}k_{i1}^r
    \sum_{k=1}^{M_i}|\varsigma_{ik}|^{p_i+1}
    -\vartheta_{q_i}k_{i2}^r
    \sum_{k=1}^{M_i}|\varsigma_{ik}|^{q_i+1}                    \\
    &\quad
    -\frac{\sigma_i}{2}\|\tilde W_i\|^2
    +
    \Delta_{B_i},
\end{aligned}
\label{eq:critic_vdot_with_xi}
\end{equation}
where
\begin{equation}
\begin{aligned}
    \Delta_{B_i}
    &=
    c_{p_i}
    \left(k_{i1}+M_i k_{i1}^r\right)
    \left(\bar\varepsilon_{B_i}\right)^{p_i+1}              \\
    &\quad
    +
    c_{q_i}
    \left(k_{i2}+M_i k_{i2}^r\right)
    \left(\bar\varepsilon_{B_i}\right)^{q_i+1}
    +
    \frac{\sigma_i}{2}\|W_i^\ast\|^2,
\end{aligned}
\label{eq:critic_delta_ei}
\end{equation}

\[
\bar\varepsilon_{B_i}
=
2\|W_i^\ast\|\bar\phi_i+T_i\bar r_i.
\]

Thus, the current and recorded residual terms admit the same perturbation structure, which permits their contributions to be collected in a single Lyapunov estimate.

\noindent\textit{Interpretation of \eqref{eq:critic_delta_ei}.}
By Lemma~\ref{lem:residual_perturbation}, one has
\begin{equation}
\small
\begin{aligned}
    -k_{i1}\varsigma_i|\delta_i|^{p_i}\operatorname{sgn}(\delta_i)
    &\le
    -\vartheta_{p_i}k_{i1}|\varsigma_i|^{p_i+1}
    +
    c_{p_i}k_{i1}|\varepsilon_{B_i}|^{p_i+1},                    \\
    -k_{i2}\varsigma_i|\delta_i|^{q_i}\operatorname{sgn}(\delta_i)
    &\le
    -\vartheta_{q_i}k_{i2}|\varsigma_i|^{q_i+1}
    +
    c_{q_i}k_{i2}|\varepsilon_{B_i}|^{q_i+1},                    \\
    -k_{i1}^{r}\varsigma_{ik}|\delta_{ik}|^{p_i}
    \operatorname{sgn}(\delta_{ik})
    &\le
    -\vartheta_{p_i}k_{i1}^{r}|\varsigma_{ik}|^{p_i+1}
    +
    c_{p_i}k_{i1}^{r}|\varepsilon_{B_{ik}}|^{p_i+1},             \\
    -k_{i2}^{r}\varsigma_{ik}|\delta_{ik}|^{q_i}
    \operatorname{sgn}(\delta_{ik})
    &\le
    -\vartheta_{q_i}k_{i2}^{r}|\varsigma_{ik}|^{q_i+1}
    +
    c_{q_i}k_{i2}^{r}|\varepsilon_{B_{ik}}|^{q_i+1}.
\end{aligned}
\label{eq:critic_interpret_each_term}
\end{equation}

Taking the summation over the recorded data gives
\(
-\sum_{k=1}^{M_i} k_{i1}^{r}\varsigma_{ik}|\delta_{ik}|^{p_i}\operatorname{sgn}(\delta_{ik})
\le
-\vartheta_{p_i}k_{i1}^{r}\sum_{k=1}^{M_i}|\varsigma_{ik}|^{p_i+1}
+
c_{p_i}k_{i1}^{r}\sum_{k=1}^{M_i}|\varepsilon_{B_{ik}}|^{p_i+1},
\quad
-\sum_{k=1}^{M_i} k_{i2}^{r}\varsigma_{ik}|\delta_{ik}|^{q_i}\operatorname{sgn}(\delta_{ik})
\le
-\vartheta_{q_i}k_{i2}^{r}\sum_{k=1}^{M_i}|\varsigma_{ik}|^{q_i+1}
+
c_{q_i}k_{i2}^{r}\sum_{k=1}^{M_i}|\varepsilon_{B_{ik}}|^{q_i+1}.
\)

To close the perturbation estimate, the approximation residuals must be uniformly bounded on the compact learning region.

\begin{lemma}[Bounded Bellman--Isaacs approximation residual]
Suppose that the learning trajectory remains in the compact set
\(\Omega_i\), the ideal weight \(W_i^\ast\) is finite, and the basis
functions are continuous on \(\Omega_i^{r}\). Then there exists a
constant \(\bar\varepsilon_{B_i}>0\) such that
\(|\varepsilon_{B_i}(t)|\leq\bar\varepsilon_{B_i}\) and
\(|\varepsilon_{B_{ik}}|\leq\bar\varepsilon_{B_i}\),
\(k=1,\ldots,M_i\).
\label{lm:lemma7}
\end{lemma}

\begin{proof}
Since \(\Omega_i^{r}\) is compact and \(\phi_i\) is continuous, there
exists \(\bar\phi_i>0\) such that
\(\|\phi_i(\chi_i)\|\leq\bar\phi_i\) for all
\(\chi_i\in\Omega_i^{r}\). Hence,
\(\|\Delta\phi_i\|\leq2\bar\phi_i\). Moreover, by
Lemma~\ref{lem:window_running_bound},
\(|\hat r_i(\tau)|\leq\bar r_i\) on the considered finite window, and
therefore
\(|\mathcal I_i|
\leq
\int_{t-T_i}^{t}|\hat r_i(\tau)|\,d\tau
\leq
T_i\bar r_i\).
Therefore,
\(|\varepsilon_{B_i}|
\leq
\|W_i^\ast\|\,\|\Delta\phi_i\|+|\mathcal I_i|
\leq
2\|W_i^\ast\|\bar\phi_i+T_i\bar r_i\).
The same argument applies to each recorded datum. Thus, one may take
\(\bar\varepsilon_{B_i}
=
2\|W_i^\ast\|\bar\phi_i+T_i\bar r_i\).
\end{proof}




By Lemma~\ref{lm:lemma7},
\(
|\varepsilon_{B_i}|^{p_i+1}
\le
(\bar\varepsilon_{B_i})^{p_i+1},
\qquad
|\varepsilon_{B_i}|^{q_i+1}
\le
(\bar\varepsilon_{B_i})^{q_i+1},
\qquad
\sum_{k=1}^{M_i}|\varepsilon_{B_{ik}}|^{p_i+1}
\le
M_i(\bar\varepsilon_{B_i})^{p_i+1},
\qquad
\sum_{k=1}^{M_i}|\varepsilon_{B_{ik}}|^{q_i+1}
\le
M_i(\bar\varepsilon_{B_i})^{q_i+1}.
\)

Therefore,
\(
\Delta_{B_i}
=
c_{p_i}
\left(
k_{i1}+M_i k_{i1}^{r}
\right)
(\bar\varepsilon_{B_i})^{p_i+1}
+
c_{q_i}
\left(
k_{i2}+M_i k_{i2}^{r}
\right)
(\bar\varepsilon_{B_i})^{q_i+1}
+
\frac{\sigma_i}{2}\|W_i^\ast\|^2.
\)

The constant \(\Delta_{B_i}\) therefore collects only bounded approximation, integration, and leakage contributions and is independent of the initial critic-weight error.

After collecting all bounded residual contributions, the initial-excitation condition supplies the coercivity needed to recover powers of the full critic-weight error. The online terms in \eqref{eq:critic_vdot_with_xi} are nonpositive and may be dropped. Since
\(\mathbf\Psi_i=\sum_{k=1}^{M_i}\psi_{ik}\psi_{ik}^\top
\ge\lambda_i\mathbf I\) by Lemma~\ref{lem:ie_condition}, one has
\(
\sum_{k=1}^{M_i}|\varsigma_{ik}|^2
=
\sum_{k=1}^{M_i}|\psi_{ik}^\top\tilde W_i|^2
=
\tilde W_i^\top
\left(
\sum_{k=1}^{M_i}\psi_{ik}\psi_{ik}^\top
\right)
\tilde W_i
=
\tilde W_i^\top\mathbf\Psi_i\tilde W_i
\ge
\lambda_i\|\tilde W_i\|^2.
\)

\begin{lemma}[Finite-dimensional \(\ell_p\)-norm relation
\cite{HardyLittlewoodPolya1952Inequalities,HornJohnson2013MatrixAnalysis}]
\label{lem:lp_norm_relation}
For any \(z\in\mathbb R^M\) and \(1\le a\le b\), one has
\begin{equation}
    \|z\|_a\ge \|z\|_b .
    \label{eq:lp_norm_monotonicity}
\end{equation}
\end{lemma}

Since \(p_i+1\in(1,2)\), Lemma~\ref{lem:lp_norm_relation} gives
\begin{equation}
    \sum_{k=1}^{M_i}|\varsigma_{ik}|^{p_i+1}
    \ge
    \left(
    \sum_{k=1}^{M_i}|\varsigma_{ik}|^2
    \right)^{\frac{p_i+1}{2}}
    \ge
    \lambda_i^{\frac{p_i+1}{2}}
    \|\tilde W_i\|^{p_i+1}.
    \label{eq:critic_low_power_ie}
\end{equation}
Since \(q_i+1>2\),
\begin{equation}
    \sum_{k=1}^{M_i}|\varsigma_{ik}|^{q_i+1}
    \ge
    M_i^{\frac{1-q_i}{2}}
    \left(
    \sum_{k=1}^{M_i}|\varsigma_{ik}|^2
    \right)^{\frac{q_i+1}{2}},
\end{equation}
and hence
\begin{equation}
    \sum_{k=1}^{M_i}|\varsigma_{ik}|^{q_i+1}
    \ge
    M_i^{\frac{1-q_i}{2}}
    \lambda_i^{\frac{q_i+1}{2}}
    \|\tilde W_i\|^{q_i+1}.
    \label{eq:critic_high_power_ie}
\end{equation}

Using \eqref{eq:critic_low_power_ie} and
\eqref{eq:critic_high_power_ie} in
\eqref{eq:critic_vdot_with_xi} gives
\(\dot V_{Wi}
\le
-\bar h_{i1}\|\tilde W_i\|^{p_i+1}
-\bar h_{i2}\|\tilde W_i\|^{q_i+1}
+\Delta_{B_i}\),
where
\(\bar h_{i1}
=
\vartheta_{p_i}k_{i1}^r
\lambda_i^{\frac{p_i+1}{2}}\)
and
\(\bar h_{i2}
=
\vartheta_{q_i}k_{i2}^r
M_i^{\frac{1-q_i}{2}}
\lambda_i^{\frac{q_i+1}{2}}\).
Let
\(\underline\lambda_{\Gamma_i}
=
\lambda_{\min}(\mathbf\Gamma_i)\)
and
\(\overline\lambda_{\Gamma_i}
=
\lambda_{\max}(\mathbf\Gamma_i)\).
we have
\(\frac{1}{2\overline\lambda_{\Gamma_i}}\|\tilde W_i\|^2
\le
V_{Wi}
\le
\frac{1}{2\underline\lambda_{\Gamma_i}}\|\tilde W_i\|^2\),
which implies
\(\|\tilde W_i\|^{p_i+1}
\ge
(2\underline\lambda_{\Gamma_i})^{\rho_{Wi}}
V_{Wi}^{\rho_{Wi}}\)
and
\(\|\tilde W_i\|^{q_i+1}
\ge
(2\underline\lambda_{\Gamma_i})^{\theta_{Wi}}
V_{Wi}^{\theta_{Wi}}\).
Therefore,
\(\dot V_{Wi}
\le
-h_{i1}V_{Wi}^{\rho_{Wi}}
-h_{i2}V_{Wi}^{\theta_{Wi}}
+\Delta_{Wi}\),
where
\(h_{i1}
=
\bar h_{i1}(2\underline\lambda_{\Gamma_i})^{\rho_{Wi}}\),
\(h_{i2}
=
\bar h_{i2}(2\underline\lambda_{\Gamma_i})^{\theta_{Wi}}\),
and
\(\Delta_{Wi}=\Delta_{B_i}\).

This proves \eqref{eq:critic_fixed_time_ineq}. By
Theorem~\ref{thm:value_fixed_time_condition}, \(V_{Wi}\) reaches
\(\Omega_{Wi}\) in a time satisfying
\eqref{eq:critic_settling_time}. Since the bound is independent of
\(V_{Wi}(0)\), the critic weight error is fixed-time practically stable.
\end{proof}

The preceding result concerns the learning subsystem alone. The next theorem combines it with the cost-induced HJI decay to establish the behavior of the learned formation dynamics.

\begin{theorem}[Fixed-time practical stability of the local learned formation dynamics]
\label{thm:closed_loop_fixed_time}
Consider follower \(i\) under the local dynamics \eqref{eq:hji_local_dynamics_game}, the saturation-compatible learned policy \eqref{eq:approx_secure_policy}, and the critic update \eqref{eq:irl_update_new}. Suppose that Assumptions~\ref{assk}--\ref{ass:critic_gain} hold, all closed-loop signals remain in the compact operating set \(\Omega_i\), the disturbance and FDI channels are bounded, and the ideal HJI value function satisfies Lemma~\ref{lem:local_value_bounds}. Let the critic powers be selected consistently with the cost-induced fixed-time powers as \(p_i=\alpha_i\), \(q_i=2\beta_i-1\), \(0<\alpha_i<1\), and \(\beta_i>1\), and define \(\rho_i=(\alpha_i+1)/2\in(0,1)\) and \(\theta_i=\beta_i>1\). Under this matching, \eqref{eq:critic_rho_theta} gives \(\rho_{Wi}=\rho_i\) and \(\theta_{Wi}=\theta_i\), so the plant and critic contributions share the same pair of fixed-time powers. Then, for any \(\lambda_{Wi}>0\), the mixed integral--pointwise Lyapunov functional
\begin{equation}
\mathcal J_i(t)=\int_{t-T_i}^{t}V_i^\ast(\chi_i(\tau))\,d\tau+\lambda_{Wi}V_{Wi}(t)
\label{eq:thm3_local_lyapunov}
\end{equation}
satisfies
\begin{equation}
\dot{\mathcal J}_i\le -c_{i1}\mathcal J_i^{\rho_i}-c_{i2}\mathcal J_i^{\theta_i}+\Delta_{ci},
\label{eq:thm3_main_ineq}
\end{equation}
for some constants \(c_{i1}>0\), \(c_{i2}>0\), and \(\Delta_{ci}\ge0\). Consequently, \(\mathcal J_i(t)\) reaches the compact residual set \(\Omega_{ci}=\{\mathcal J_i\ge0:c_{i1}\mathcal J_i^{\rho_i}+c_{i2}\mathcal J_i^{\theta_i}\le\Delta_{ci}\}\) within the fixed time \(T_{ci}\le[c_{i1}(1-\rho_i)]^{-1}+[c_{i2}(\theta_i-1)]^{-1}\). Moreover, for all \(t\ge T_{\rm cl}:=\max_{1\le i\le N}T_{ci}\), the leader-rooted formation error is practically bounded.
\end{theorem}

\begin{proof}
The Lyapunov functional for follower \(i\) as
\begin{equation}
    \mathcal J_i(t)
    =
    \bar V_i(t)
    +
    \lambda_{Wi}V_{Wi}(t),
    \label{eq:proof3_Ji_select}
\end{equation}
where
\begin{equation}
    \bar V_i(t)
    =
    \int_{t-T_i}^{t}
    V_i^\ast(\chi_i(\tau))d\tau.
    \label{eq:proof3_barV_VW}
\end{equation}
Here \(\bar V_i\) is an integral value on the IRL window, while
\(V_{Wi}\) is a pointwise critic Lyapunov function. Differentiating
\eqref{eq:proof3_Ji_select} gives
\begin{equation}
    \dot{\mathcal J}_i
    =
    \dot{\bar V}_i
    +
    \lambda_{Wi}\dot V_{Wi}.
    \label{eq:proof3_Jidot_split}
\end{equation}
For the first term,
\begin{equation}
\begin{aligned}
    \dot{\bar V}_i
    &=
    V_i^\ast(\chi_i(t))
    -
    V_i^\ast(\chi_i(t-T_i))                                  \\
    &=
    \int_{t-T_i}^{t}
    \dot V_i^\ast(\chi_i(\tau))d\tau .
\end{aligned}
\label{eq:proof3_barV_dot}
\end{equation}

For compactness in the remainder of the proof, define
\begin{equation}
    \mathcal F_i^{\chi}
    :=
    F_i^{\chi}(\xi_i)
    +
    \mathbf G_{i\mathcal N}^{\chi}(x_{\mathcal N_i})
    u_{c\mathcal N_i}.
    \label{eq:proof3_compact_drift}
\end{equation}

Thus, \(\mathcal F_i^{\chi}\) contains the nonlinear drift and the
fixed neighboring secure-control coupling appearing in
\eqref{eq:hji_local_dynamics_game}. Along the learned closed-loop trajectory,
\begingroup
\setlength{\jot}{1.5pt}
\begin{equation}
\small
\begin{aligned}
\dot V_i^\ast
={}&
\nabla V_i^{\ast\top}
\Bigl[
\mathcal F_i^\chi
+\mathbf G_{ii}^\chi(x_i)\hat u_{ci}
\Bigr]
+
\nabla V_i^{\ast\top}\mathbf B_i^\chi(x_i)\varpi_i
\\[-0.2mm]
&+
\sum_{j\in\mathcal N_i}
\nabla V_i^{\ast\top}
\mathbf B_{ij}^\chi(x_j)\varpi_j
\\
={}&
\nabla V_i^{\ast\top}
\Bigl[
\mathcal F_i^\chi
+\mathbf G_{ii}^\chi(x_i)u_{ci}^\ast
\Bigr]
+
\nabla V_i^{\ast\top}\mathbf B_i^\chi(x_i)\varpi_i
\\[-0.2mm]
&+
\sum_{j\in\mathcal N_i}
\nabla V_i^{\ast\top}
\mathbf B_{ij}^\chi(x_j)\varpi_j
+
\nabla V_i^{\ast\top}
\mathbf G_{ii}^\chi(x_i)\tilde u_i .
\end{aligned}
\label{eq:proof3_vdot_split}
\end{equation}
\endgroup


Here, \(\tilde u_i=\hat u_{ci}-u_{ci}^\ast\). From the local HJI saddle inequality,
\(
\nabla V_i^{\ast\top}
\Bigl[
\mathcal F_i^\chi
+\mathbf G_{ii}^\chi(x_i)u_{ci}^\ast
\Bigr]
+
\nabla V_i^{\ast\top}\mathbf B_i^\chi(x_i)\varpi_i
+
\sum_{j\in\mathcal N_i}
\nabla V_i^{\ast\top}
\mathbf B_{ij}^\chi(x_j)\varpi_j
\le
-Q_{ii}(\chi_i)
-\mathcal U_i(u_{ci}^\ast)
-\sum_{j\in\mathcal N_i}\mathcal U_{ij}(u_{cj})
+
\gamma_i^2\varpi_i^\top\mathbf T_{ii}\varpi_i
+
\gamma_i^2
\sum_{j\in\mathcal N_i}
\varpi_j^\top\mathbf T_{ij}\varpi_j.
\)

Since \(\mathcal U_i(\cdot)\ge0\) and
\(\mathcal U_{ij}(\cdot)\ge0\), and since the disturbance--FDI channels
are bounded, there exists \(\bar d_i>0\) such that
\begin{equation}
    \gamma_i^2\varpi_i^\top\mathbf T_{ii}\varpi_i
    +
    \gamma_i^2
    \sum_{j\in\mathcal N_i}
    \varpi_j^\top\mathbf T_{ij}\varpi_j
    \le
    \bar d_i .
    \label{eq:proof3_attack_bound}
\end{equation}

Hence, \eqref{eq:proof3_vdot_split} becomes

\begin{equation}
    \dot V_i^\ast
\le
-Q_{ii}(\chi_i)
+\nabla V_i^{\ast\top}
\mathbf G_{ii}^{\chi}(x_i)\tilde u_i
+\bar d_i
\label{eq:proof3_vdot_pre_uerr}
\end{equation}

The policy mismatch is now estimated. Define
\(z_i^\ast
:=
\frac{1}{2}
\mathbf R_i^{-1}
\bar{\mathbf U}_i^{-1}
\mathbf G_{ii}^{\chi,\top}(x_i)
\nabla V_i^\ast\)
and
\(\hat z_i
:=
\frac{1}{2}
\mathbf R_i^{-1}
\bar{\mathbf U}_i^{-1}
\mathbf G_{ii}^{\chi,\top}(x_i)
\nabla\hat V_i\).
Then
\(u_{ci}^{\ast}
=
-\bar{\mathbf U}_i\tanh(z_i^\ast)\)
and
\(\hat u_{ci}
=
-\bar{\mathbf U}_i\tanh(\hat z_i)\).
Using
\(\|\tanh(a)-\tanh(b)\|\le\|a-b\|\), one obtains
\begin{equation}
\begin{aligned}
\|\tilde u_i\|
&\le
\|\bar{\mathbf U}_i\|\|\hat z_i-z_i^\ast\|\\
&\le
\frac{1}{2}
\|\bar{\mathbf U}_i\|
\|\mathbf R_i^{-1}\bar{\mathbf U}_i^{-1}\|
\|\mathbf G_{ii}^{\chi,\top}(x_i)\|
\|\nabla\hat V_i-\nabla V_i^\ast\|.
\end{aligned}
\label{eq:proof3_utilde_lip}
\end{equation}
Moreover,
\begin{equation}
\nabla\hat V_i-\nabla V_i^\ast
=
\nabla\phi_i^\top\hat W_i
-\nabla\phi_i^\top W_i^\ast
-\nabla\varepsilon_i
=
\nabla\phi_i^\top\tilde W_i-\nabla\varepsilon_i.
\label{eq:proof3_gradient_error}
\end{equation}

Since all signals remain in the compact set \(\Omega_i^{r}\), there
exist constants \(\bar g_i>0\), \(\bar\phi_{di}>0\), and
\(\bar\varepsilon_{di}>0\) such that
\begin{equation}
\|\mathbf G_{ii}^{\chi}(x_i)\|\le\bar g_i,\qquad
\|\nabla\phi_i\|\le\bar\phi_{di},\qquad
\|\nabla\varepsilon_i\|\le\bar\varepsilon_{di}.
\label{eq:proof3_compact_1}
\end{equation}
Moreover, by the gradient bound in
Lemma~\ref{lem:local_value_bounds},
\(\|\nabla V_i^\ast(\chi_i)\|
\le c_{\nabla i}\|\chi_i\|\).
Therefore,
\begin{equation}
\|\mathbf G_{ii}^{\chi,\top}(x_i)\nabla V_i^\ast\|
\le
\|\mathbf G_{ii}^{\chi}(x_i)\|\,\|\nabla V_i^\ast\|
\le
\bar g_i c_{\nabla i}\|\chi_i\|
=:\bar v_i\|\chi_i\|.
\label{eq:proof3_compact_2}
\end{equation}

Combining \eqref{eq:proof3_utilde_lip} and
\eqref{eq:proof3_gradient_error}, and using the triangle inequality,
yields
\(\|\tilde u_i\|
\le
\frac{1}{2}\|\bar{\mathbf U}_i\|
\|\mathbf R_i^{-1}\bar{\mathbf U}_i^{-1}\|
\|\mathbf G_{ii}^{\chi}(x_i)\|
\bigl(
\|\nabla\phi_i\|\,\|\tilde W_i\|
+\|\nabla\varepsilon_i\|
\bigr)\).
Therefore, using the bounds in \eqref{eq:proof3_compact_1}, one obtains
\begin{equation}
\|\tilde u_i\|
\le
m_{ui}\|\tilde W_i\|+\epsilon_{ui},
\label{eq:proof3_utilde_bound}
\end{equation}
where
\(m_{ui}
:=
\frac{1}{2}\|\bar{\mathbf U}_i\|
\|\mathbf R_i^{-1}\bar{\mathbf U}_i^{-1}\|
\bar g_i\bar\phi_{di}\)
and
\(\epsilon_{ui}
:=
\frac{1}{2}\|\bar{\mathbf U}_i\|
\|\mathbf R_i^{-1}\bar{\mathbf U}_i^{-1}\|
\bar g_i\bar\varepsilon_{di}\).
Using \eqref{eq:proof3_compact_2} and
\eqref{eq:proof3_utilde_bound},
\begin{equation}
\nabla V_i^{\ast\top}
\mathbf G_{ii}^{\chi}(x_i)\tilde u_i
\le
\bar v_i\|\chi_i\|
\bigl(
m_{ui}\|\tilde W_i\|
+\epsilon_{ui}
\bigr).
\label{eq:proof3_policy_young_1}
\end{equation}

Let
\(\underline q_i:=\lambda_{\min}(\mathbf Q_i)>0\).
Since
\(\underline q_i\|\chi_i\|^2
\le
\chi_i^\top\mathbf Q_i\chi_i\),
applying Young's inequality
\(ab\le
\frac{\epsilon_{Qi}}{2}a^2
+\frac{1}{2\epsilon_{Qi}}b^2\),
\(0<\epsilon_{Qi}<1\),
with
\(a=\sqrt{\underline q_i}\|\chi_i\|\)
and
\(b=
\frac{\bar v_i m_{ui}}{\sqrt{\underline q_i}}
\|\tilde W_i\|\),

gives
\begin{equation}
\bar v_i m_{ui}\|\chi_i\|\|\tilde W_i\|
\le
\frac{\epsilon_{Qi}}{2}
\chi_i^\top\mathbf Q_i\chi_i
+
\frac{\bar v_i^2m_{ui}^2}
{2\epsilon_{Qi}\underline q_i}
\tilde W_i^\top\tilde W_i.
\label{eq:proof3_policy_young_1}
\end{equation}
Similarly, choosing \(a=\sqrt{\underline q_i}\|\chi_i\|\) and \(b=\frac{\bar v_i\epsilon_{ui}}{\sqrt{\underline q_i}}\) yields \(\bar v_i\epsilon_{ui}\|\chi_i\|\le\frac{\epsilon_{Qi}}{2}\chi_i^\top\mathbf Q_i\chi_i+\frac{\bar v_i^2\epsilon_{ui}^2}{2\epsilon_{Qi}\underline q_i}\). Combining the above inequalities gives
\begin{equation}
\nabla V_i^{\ast\top}
\mathbf G_{ii}^{\chi}(x_i)\tilde u_i
\le
\epsilon_{Qi}\chi_i^\top\mathbf Q_i\chi_i
+
L_{W_1}\tilde W_i^\top\tilde W_i
+
\Delta_{ui},
\label{eq:proof3_policy_young_2}
\end{equation}
where \(L_{W_1}:=\frac{\bar v_i^2m_{ui}^2}{2\epsilon_{Qi}\underline q_i}>0\) and \(\Delta_{ui}:=\frac{\bar v_i^2\epsilon_{ui}^2}{2\epsilon_{Qi}\underline q_i}\ge0\).

The policy-mismatch estimate is now combined with the fixed-time structure embedded in the state penalty. By the cost design, \(Q_{ii}(\chi_i)=\chi_i^\top\mathbf Q_i\chi_i+\kappa_{i1}\|\chi_i\|^{2\alpha_i}+\kappa_{i2}\|\chi_i\|^{2\beta_i}\). Using the upper bound in Lemma~\ref{lem:local_value_bounds}, \(\|\chi_i\|^{2\alpha_i}\ge\overline c_i^{-\alpha_i}(V_i^\ast)^{\alpha_i}\) and \(\|\chi_i\|^{2\beta_i}\ge\overline c_i^{-\beta_i}(V_i^\ast)^{\beta_i}\). Let \(a_i=\kappa_{i1}\overline c_i^{-\alpha_i}\) and \(b_i=\kappa_{i2}\overline c_i^{-\beta_i}\). Then
\begin{equation}
    Q_{ii}(\chi_i)
    \ge
    \chi_i^\top\mathbf Q_i\chi_i
    +
    a_i(V_i^\ast)^{\alpha_i}
    +
    b_i(V_i^\ast)^{\beta_i}.
    \label{eq:proof3_Qii_lower}
\end{equation}

Substituting \eqref{eq:proof3_policy_young_2} and \eqref{eq:proof3_Qii_lower} into \eqref{eq:proof3_vdot_pre_uerr} gives
\begin{align}
\dot V_i^\ast
\le{}&
-(1-\epsilon_{Qi})
\chi_i^\top \mathbf Q_i\chi_i
-a_i\bigl(V_i^\ast\bigr)^{\alpha_i}
-b_i\bigl(V_i^\ast\bigr)^{\beta_i}
\nonumber\\
&+\ell_{Wi}V_{Wi}
+\Delta_{Vi},
\label{eq:proof3_vdot_shaped}
\end{align}
where \(\Delta_{Vi}:=\bar d_i+\Delta_{ui}\) and \(\ell_{W_i}:=2\lambda_{\max}(\Gamma_i)L_{W_i}>0\).

By \eqref{eq:proof3_barV_dot}, one has \(\dot{\bar V}_i(t)=\int_{t-T_i}^{t}\dot V_i^\ast(\tau)\,d\tau\). Hence, integrating \eqref{eq:proof3_vdot_shaped} over \([t-T_i,t]\) yields
\begin{align}
\dot{\bar V}_i
\le{}&
-(1-\epsilon_{Qi})
\int_{t-T_i}^{t}
\chi_i^\top(\tau)\mathbf Q_i\chi_i(\tau)\,d\tau
\nonumber\\
&-a_i
\int_{t-T_i}^{t}
\bigl(V_i^\ast(\tau)\bigr)^{\alpha_i}\,d\tau
-b_i
\int_{t-T_i}^{t}
\bigl(V_i^\ast(\tau)\bigr)^{\beta_i}\,d\tau
\nonumber\\
&+\ell_{Wi}
\int_{t-T_i}^{t}
V_{Wi}(\tau)\,d\tau
+T_i\Delta_{Vi}.
\label{eq:proof3_barV_ineq_with_quadratic}
\end{align}

Since \(0<\epsilon_{Q_i}<1\) and \(\mathbf Q_i>0\), the quadratic term \(-(1-\epsilon_{Qi})\int_{t-T_i}^{t}\chi_i^\top(\tau)\mathbf Q_i\chi_i(\tau)\,d\tau\) on the right-hand side of \eqref{eq:proof3_barV_ineq_with_quadratic} is nonpositive. Therefore, omitting this term preserves a valid upper bound, yielding
\begin{align}
\dot{\bar V}_i
\le{}&
-a_i
\int_{t-T_i}^{t}
\bigl(V_i^\ast(\tau)\bigr)^{\alpha_i}\,d\tau
-b_i
\int_{t-T_i}^{t}
\bigl(V_i^\ast(\tau)\bigr)^{\beta_i}\,d\tau
\nonumber\\
&+\ell_{Wi}
\int_{t-T_i}^{t}
V_{Wi}(\tau)\,d\tau
+T_i\Delta_{Vi}.
\label{eq:proof3_barV_ineq_1}
\end{align}

The two windowed value-power integrals are treated separately because the low-order power is concave whereas the high-order power is convex. Let \(\zeta_i(\tau):=V_i^\ast(\chi_i(\tau))\) and \(\bar V_i(t):=\int_{t-T_i}^{t}\zeta_i(\tau)\,d\tau\). Since \(V_i^\ast\) is continuous on the compact set \(\Omega_i\), there exists a constant \(V_{i,\max}>0\) such that \(0\le\zeta_i(\tau)\le V_{i,\max}\). Moreover, the preceding derivative bound gives
\begin{equation}
    |\dot{\zeta}_i(\tau)|\le L_{V_i}.
    \label{eq:proof3_Vdot_bound}
\end{equation}

\begin{remark}
For \(0<\alpha_i<1\), the mapping
\(s\mapsto s^{\alpha_i}\) is concave on
\(\mathbb R_{\geq0}\); 
Therefore, Jensen's inequality gives
\[
\frac{1}{T_i}
\int_{t-T_i}^{t}
\zeta_i^{\alpha_i}(\tau)\,d\tau
\le
\left(
\frac{1}{T_i}
\int_{t-T_i}^{t}
\zeta_i(\tau)\,d\tau
\right)^{\alpha_i}.
\]
This is an upper bound on the integral of
\(\zeta_i^{\alpha_i}\). However, in the Lyapunov derivative, this
integral appears with the negative coefficient \(-a_i\).
Multiplication by \(-a_i<0\) reverses the inequality and therefore
provides a lower bound, rather than the upper bound required for
estimating the Lyapunov derivative. Hence, Jensen's inequality alone
cannot produce the desired decay estimate for \(0<\alpha_i<1\),
which motivates the derivative-based finite-window argument below.
\end{remark}

Define \(\mathcal M_{V_i}(t):=\max_{\tau\in[t-T_i,t]}\zeta_i(\tau)\). If \(\mathcal M_{V_i}(t)=0\), the desired estimate holds trivially. Otherwise, let \(\tau_i^m\in[t-T_i,t]\) be such that \(\zeta_i(\tau_i^m)=\mathcal M_{V_i}(t)\). At least one of the two intervals adjacent to \(\tau_i^m\) within \([t-T_i,t]\) has length no smaller than \(T_i/2\). Therefore, over an interval of length \(\tau_i^{\mathrm{loc}}(t):=\min\left\{\frac{\mathcal M_{V_i}(t)}{2L_{V_i}},\frac{T_i}{2}\right\}\), where \(\tau_i^{\mathrm{loc}}(t)=|\tau-\tau_i^m|\) denotes the time distance from the maximizer \(\tau_i^m\), the derivative bound implies \(\zeta_i(\tau)\ge\mathcal M_{V_i}(t)-L_{V_i}|\tau-\tau_i^m|\ge\frac{\mathcal M_{V_i}(t)}{2}\). Consequently, \(\bar V_i(t)\ge\frac{\mathcal M_{V_i}(t)}{2}\tau_i^{\mathrm{loc}}(t)\). If \(\mathcal M_{V_i}(t)\le L_{V_i}T_i\), then \(\tau_i^{\mathrm{loc}}(t)=\mathcal M_{V_i}(t)/(2L_{V_i})\), and the preceding bound gives \(\mathcal M_{V_i}(t)\le2\sqrt{L_{V_i}\bar V_i(t)}\). If \(\mathcal M_{V_i}(t)>L_{V_i}T_i\), then \(\tau_i^{\mathrm{loc}}(t)=T_i/2\), and hence \(\mathcal M_{V_i}(t)\le\frac{4}{T_i}\bar V_i(t)\). Since \(\bar V_i(t)=\int_{t-T_i}^{t}\zeta_i(\tau)\,d\tau\le T_iV_{i,\max}\), it follows that \(\frac{4}{T_i}\bar V_i(t)\le4\sqrt{\frac{V_{i,\max}}{T_i}}\,\bar V_i^{1/2}(t)\). Thus, in both cases, \(\mathcal M_{V_i}(t)\le C_{V_i}\bar V_i^{1/2}(t)\), where \(C_{V_i}:=\max\left\{2\sqrt{L_{V_i}},4\sqrt{\frac{V_{i,\max}}{T_i}}\right\}>0\).

Since \(0<\alpha_i<1\) and \(0\le\zeta_i(\tau)\le\mathcal M_{V_i}(t)\), one has \(\zeta_i^{\alpha_i}(\tau)=\zeta_i(\tau)\zeta_i^{\alpha_i-1}(\tau)\ge\mathcal M_{V_i}^{\alpha_i-1}(t)\zeta_i(\tau)\). Integrating both sides over \([t-T_i,t]\) gives \(\int_{t-T_i}^{t}\zeta_i^{\alpha_i}(\tau)\,d\tau\ge\mathcal M_{V_i}^{\alpha_i-1}(t)\bar V_i(t)\ge C_{V_i}^{-(1-\alpha_i)}\bar V_i^{(\alpha_i+1)/2}(t)\), where \(\mathcal M_{V_i}(t)\le C_{V_i}\bar V_i^{1/2}(t)\) and \(\alpha_i-1<0\) have been used. Therefore, \(\int_{t-T_i}^{t}\bigl(V_i^\ast(\chi_i(\tau))\bigr)^{\alpha_i}\,d\tau\ge\omega_{\alpha i}\bar V_i^{\rho_i}(t)\), where \(\rho_i:=\frac{\alpha_i+1}{2}\in(0,1)\) and \(\omega_{\alpha i}:=C_{V_i}^{-(1-\alpha_i)}>0\).

\begin{remark}[Interpretation of the local time bound]
The quantity \(\mathcal M_{V_i}(t)/(2L_{V_i})\) follows directly from the derivative bound. Indeed, since \(|\dot{\zeta}_i(\tau)|\le L_{V_i}\) and \(\zeta_i(\tau_i^m)=\mathcal M_{V_i}(t)\), absolute continuity implies \(\zeta_i(\tau)\ge \mathcal M_{V_i}(t)-L_{V_i}|\tau-\tau_i^m|\). Hence, whenever \(|\tau-\tau_i^m|\le \mathcal M_{V_i}(t)/(2L_{V_i})\), one has \(\zeta_i(\tau)\ge \mathcal M_{V_i}(t)/2\). Thus, \(\mathcal M_{V_i}(t)/(2L_{V_i})\) is the minimum time required, under the rate bound \(L_{V_i}\), for \(\zeta_i\) to decrease from \(\mathcal M_{V_i}(t)\) to \(\mathcal M_{V_i}(t)/2\). The threshold \(\mathcal M_{V_i}(t)/2\) is introduced to convert the pointwise maximum \(\mathcal M_{V_i}(t)\) into a lower bound that holds over a subinterval of positive length. More generally, choosing a threshold \(\vartheta\mathcal M_{V_i}(t)\), with \(0<\vartheta<1\), yields the guaranteed local length \((1-\vartheta)\mathcal M_{V_i}(t)/L_{V_i}\). In the derivative-limited case, the resulting integral lower bound is proportional to \(\vartheta(1-\vartheta)\mathcal M_{V_i}^{2}(t)\), which is maximized at \(\vartheta=1/2\). The half-maximum threshold therefore provides a simple and sharp choice for the subsequent integral estimate. Moreover, since the maximizing instant \(\tau_i^m\) may be located near an endpoint of \([t-T_i,t]\), only one adjacent subinterval of length at least \(T_i/2\) can always be guaranteed. This motivates the choice \(\tau_i^{\mathrm{loc}}(t):=\min\left\{\mathcal M_{V_i}(t)/(2L_{V_i}),T_i/2\right\}\). Consequently, there exists a subinterval of length \(\tau_i^{\mathrm{loc}}(t)\) on which \(\zeta_i(\tau)\ge \mathcal M_{V_i}(t)/2\), and therefore \(\bar V_i(t)=\int_{t-T_i}^{t}\zeta_i(\tau)\,d\tau\ge \bigl(\mathcal M_{V_i}(t)/2\bigr)\tau_i^{\mathrm{loc}}(t)\).
\end{remark}

\begin{lemma}[Jensen inequality \cite{boyd2004convex}]
\label{lem:power-convexity}
Let \(\beta>1\). Then the function \(f(s):=s^\beta\), \(s\in\mathbb R_{\geq0}\), is strictly convex. Consequently, for every \(T>0\) and every nonnegative integrable function \(z:[t-T,t]\to\mathbb R_{\geq0}\), \(\frac{1}{T}\int_{t-T}^{t}z^\beta(\tau)\,d\tau\ge\left(\frac{1}{T}\int_{t-T}^{t}z(\tau)\,d\tau\right)^\beta\). Equivalently, \(\int_{t-T}^{t}z^\beta(\tau)\,d\tau\ge T^{1-\beta}\left(\int_{t-T}^{t}z(\tau)\,d\tau\right)^\beta\).
\end{lemma}

From Lemma~\ref{lem:power-convexity}, since \(s\mapsto s^{\beta_i}\) is convex for \(\beta_i>1\), Jensen's inequality gives \(\frac{1}{T_i}\int_{t-T_i}^{t}\zeta_i^{\beta_i}(\tau)\,d\tau\ge\left(\frac{1}{T_i}\int_{t-T_i}^{t}\zeta_i(\tau)\,d\tau\right)^{\beta_i}=\left(\frac{\bar V_i(t)}{T_i}\right)^{\beta_i}\). Equivalently, \(\int_{t-T_i}^{t}\bigl(V_i^\ast(\chi_i(\tau))\bigr)^{\beta_i}\,d\tau\ge T_i^{1-\beta_i}\bar V_i^{\beta_i}(t)\).

These two estimates convert the integral plant-value contribution into the same mixed-power form used for the critic subsystem. Set \(A_i=a_i\omega_{\alpha i}\) and \(B_i=b_iT_i^{1-\beta_i}\). Then \eqref{eq:proof3_barV_ineq_1} gives


\begin{equation}
    \dot{\bar V}_i
    \le
    -A_i\bar V_i^{\rho_i}
    -
    B_i\bar V_i^{\theta_i}
    +
    \ell_{Wi}
    \int_{t-T_i}^{t}
    V_{Wi}(\tau)d\tau
    +
    T_i\Delta_{Vi}.
    \label{eq:proof3_barV_ineq_2}
\end{equation}

For the pointwise critic component,
Theorem~\ref{thm:critic_fixed_time} and
\eqref{eq:critic_fixed_time_ineq} imply
\begin{equation}
    \dot V_{Wi}
    \le
    -h_{i1}V_{Wi}^{\rho_i}
    -
    h_{i2}V_{Wi}^{\theta_i}
    +
    \Delta_{Wi}.
    \label{eq:proof3_VW_ineq}
\end{equation}
Moreover, since the weight dynamics are bounded on \(\Omega_i\),
there exists \(M_{Wi}>0\) such that
\begin{equation}
    |\dot V_{Wi}(t)|\le M_{Wi}.
    \label{eq:proof3_VWdot_bound}
\end{equation}

Hence, for any $\tau\in[t-T_i,t]$, the fundamental theorem of
calculus gives
\begin{align}
V_{W_i}(\tau)
&=
V_{W_i}(t)
-
\int_{\tau}^{t}\dot V_{W_i}(s)\,ds
\nonumber\\
&\leq
V_{W_i}(t)
+
\int_{\tau}^{t}
\left|\dot V_{W_i}(s)\right|\,ds
\nonumber\\
&\leq
V_{W_i}(t)+M_{W_i}(t-\tau).
\label{eq:critic_window_pointwise_bound}
\end{align}
Integrating \eqref{eq:critic_window_pointwise_bound} over
$\tau\in[t-T_i,t]$ yields
\begin{equation}
\int_{t-T_i}^{t}V_{W_i}(\tau)\,d\tau
\leq
T_iV_{W_i}(t)
+\frac{1}{2}M_{W_i}T_i^2.
\label{eq:critic_window_integral_bound}
\end{equation}

Substituting \eqref{eq:critic_window_integral_bound} into
\eqref{eq:proof3_barV_ineq_2},
\begin{equation}
    \dot{\bar V}_i
    \le
    -A_i\bar V_i^{\rho_i}
    -
    B_i\bar V_i^{\theta_i}
    +
    \ell_{Wi}T_iV_{Wi}
    +
    \Delta_{\bar V_i},
    \label{eq:proof3_barV_ineq_3}
\end{equation}
where
\begin{equation}
    \Delta_{\bar V_i}
    =
    T_i\Delta_{Vi}
    +
    \frac{1}{2}\ell_{Wi}M_{Wi}T_i^2 .
    \label{eq:proof3_delta_barV}
\end{equation}
Combining \eqref{eq:proof3_barV_ineq_3} and
\eqref{eq:proof3_VW_ineq},
\begin{equation}
\begin{aligned}
    \dot{\mathcal J}_i
    &=
    \dot{\bar V}_i
    +
    \lambda_{Wi}\dot V_{Wi}                                  \\
    &\le
    -A_i\bar V_i^{\rho_i}
    -
    B_i\bar V_i^{\theta_i}
    -
    \lambda_{Wi}h_{i1}V_{Wi}^{\rho_i}
    -
    \lambda_{Wi}h_{i2}V_{Wi}^{\theta_i}                      \\
    &\quad
    +
    \ell_{Wi}T_iV_{Wi}
    +
    \Delta_{\bar V_i}
    +
    \lambda_{Wi}\Delta_{Wi}.
\end{aligned}
\label{eq:proof3_J_dot_before_absorb}
\end{equation}

The remaining linear coupling in \(V_{Wi}\) is absorbed by the
following standard mixed-power estimate.

\begin{lemma}[Mixed-power domination]
\label{lem:mixed-power-domination}
Let \(a\geq 0\), \(0<\rho<1<\theta\), and let
\(\varepsilon_1,\varepsilon_2>0\) be arbitrary constants.
Then, there exists a finite constant
\(\bar{\varepsilon}\geq 0\) such that
\begin{equation}
    as
    \leq
    \varepsilon_1 s^{\rho}
    +
    \varepsilon_2 s^{\theta}
    +
    \bar{\varepsilon},
    \qquad \forall s\geq 0.
    \label{eq:mixed-power-domination}
\end{equation}
In particular, an admissible choice is
\begin{equation}
    \bar{\varepsilon}
    =
    \frac{\theta-1}{\theta}\,
    a
    \left(
        \frac{a}{\theta\varepsilon_2}
    \right)^{\frac{1}{\theta-1}}
    \label{eq:mixed-power-remainder}
\end{equation}
for \(a>0\), whereas
\(\bar{\varepsilon}=0\) may be selected when \(a=0\).
\end{lemma}

\begin{proof} If $a=0$, then \eqref{eq:mixed-power-domination} holds trivially with $\bar{\varepsilon}=0$, since $\varepsilon_1s^\rho+\varepsilon_2s^\theta\geq0$ for every $s\geq0$. Consider now $a>0$ and define $f:[0,\infty)\rightarrow\mathbb{R}$ by $f(s):=as-\varepsilon_2s^\theta$. Since $\theta>1$, one has $f(0)=0$, $f(s)\rightarrow-\infty$ as $s\rightarrow\infty$, and $f'(s)=a-\theta\varepsilon_2s^{\theta-1}$. Hence, $f$ admits the unique maximizer $s^\star=\left(a/(\theta\varepsilon_2)\right)^{1/(\theta-1)}$, because $f'(s)>0$ for $0\leq s<s^\star$ and $f'(s)<0$ for $s>s^\star$. Moreover, the stationarity condition $a=\theta\varepsilon_2(s^\star)^{\theta-1}$ implies $\varepsilon_2(s^\star)^\theta=(a/\theta)s^\star$, and therefore $\max_{s\geq0}f(s)=f(s^\star)=as^\star-\varepsilon_2(s^\star)^\theta=\frac{\theta-1}{\theta}as^\star=\frac{\theta-1}{\theta}a\left(\frac{a}{\theta\varepsilon_2}\right)^{\frac{1}{\theta-1}}=:\bar{\varepsilon}$. Consequently, $as-\varepsilon_2s^\theta\leq\bar{\varepsilon}$ for every $s\geq0$, or equivalently, $as\leq\varepsilon_2s^\theta+\bar{\varepsilon}$. Finally, since $\varepsilon_1s^\rho\geq0$ for every $s\geq0$, it follows that $as\leq\varepsilon_1s^\rho+\varepsilon_2s^\theta+\bar{\varepsilon}$ for all $s\geq0$, which proves the result. \end{proof}

By Lemma~\ref{lem:mixed-power-domination}, for any
\(\varrho_{i1}>0\) and \(\varrho_{i2}>0\), there exists a finite
constant \(\bar{\varrho}_i\geq0\) such that
\begin{equation}
    \ell_{Wi}T_iV_{Wi}
    \leq
    \varrho_{i1}V_{Wi}^{\rho_i}
    +
    \varrho_{i2}V_{Wi}^{\theta_i}
    +
    \bar{\varrho}_i.
    \label{eq:coupling-absorption}
\end{equation}
Substituting \eqref{eq:coupling-absorption} into the preceding
estimate of \(\dot{\mathcal J}_i\) gives
\begin{align}
    \dot{\mathcal J}_i
    \leq{}&
    -A_i\bar V_i^{\rho_i}
    -B_i\bar V_i^{\theta_i}
    -\lambda_{Wi}h_{i1}V_{Wi}^{\rho_i}
    -\lambda_{Wi}h_{i2}V_{Wi}^{\theta_i}
    \nonumber\\
    &+
    \varrho_{i1}V_{Wi}^{\rho_i}
    +
    \varrho_{i2}V_{Wi}^{\theta_i}
    +
    \Delta_{\bar V_i}
    +
    \lambda_{Wi}\Delta_{Wi}
    +
    \bar{\varrho}_i
    \nonumber\\
    ={}&
    -A_i\bar V_i^{\rho_i}
    -B_i\bar V_i^{\theta_i}
    \nonumber\\
    &-
    \bigl(\lambda_{Wi}h_{i1}-\varrho_{i1}\bigr)
    V_{Wi}^{\rho_i}
    -
    \bigl(\lambda_{Wi}h_{i2}-\varrho_{i2}\bigr)
    V_{Wi}^{\theta_i}
    \nonumber\\
    &+
    \Delta_{\bar V_i}
    +
    \lambda_{Wi}\Delta_{Wi}
    +
    \bar{\varrho}_i.
    \label{eq:J-before-lambda-selection}
\end{align}

To ensure that the coefficients of
\(V_{Wi}^{\rho_i}\) and \(V_{Wi}^{\theta_i}\) in
\eqref{eq:J-before-lambda-selection} remain strictly positive, choose
\(\lambda_{Wi}>0\) such that
\(\lambda_{Wi}>
\max\{\varrho_{i1}/h_{i1},\varrho_{i2}/h_{i2}\}\), or equivalently,
\(\lambda_{Wi}h_{i1}>\varrho_{i1}\) and
\(\lambda_{Wi}h_{i2}>\varrho_{i2}\).
Such a selection always exists since \(h_{i1}>0\) and \(h_{i2}>0\).
Define
\(C_i:=\lambda_{Wi}h_{i1}-\varrho_{i1}>0\),
\(D_i:=\lambda_{Wi}h_{i2}-\varrho_{i2}>0\), and
\(\Delta_{ci}:=\Delta_{\bar V_i}
+\lambda_{Wi}\Delta_{Wi}+\bar{\varrho}_i\).
Then,
\(\dot{\mathcal J}_i
\leq
-A_i\bar V_i^{\rho_i}
-B_i\bar V_i^{\theta_i}
-C_iV_{Wi}^{\rho_i}
-D_iV_{Wi}^{\theta_i}
+\Delta_{ci}\).

Rewrite the critic powers in terms of
\(\lambda_{Wi}V_{Wi}\) to obtain
\(\dot{\mathcal J}_i
\le
-A_i\bar V_i^{\rho_i}
-C_i\lambda_{Wi}^{-\rho_i}
(\lambda_{Wi}V_{Wi})^{\rho_i}
-B_i\bar V_i^{\theta_i}
-D_i\lambda_{Wi}^{-\theta_i}
(\lambda_{Wi}V_{Wi})^{\theta_i}
+\Delta_{ci}\).
Let
\(\bar A_i
:=\min\{A_i,C_i\lambda_{Wi}^{-\rho_i}\}\) and
\(\bar B_i
:=\min\{B_i,D_i\lambda_{Wi}^{-\theta_i}\}\).
Then,
\(\dot{\mathcal J}_i
\le
-\bar A_i[
\bar V_i^{\rho_i}
+(\lambda_{Wi}V_{Wi})^{\rho_i}]
-\bar B_i[
\bar V_i^{\theta_i}
+(\lambda_{Wi}V_{Wi})^{\theta_i}]
+\Delta_{ci}\).
It remains to combine the plant and critic powers into powers of the
composite functional \(\mathcal J_i\).

\begin{lemma}[Subadditivity of fractional powers~\cite{HardyLittlewoodPolya1952Inequalities}]
\label{lem:fractional-power-subadditivity}
Let \(x,y\geq0\) and \(0<\rho<1\). Then
\((x+y)^\rho\leq x^\rho+y^\rho\).
\end{lemma}

By Lemma~\ref{lem:fractional-power-subadditivity}, since
\(0<\rho_i<1\),
\(\bar V_i^{\rho_i}
+(\lambda_{Wi}V_{Wi})^{\rho_i}
\ge
(\bar V_i+\lambda_{Wi}V_{Wi})^{\rho_i}
=\mathcal J_i^{\rho_i}\).

\begin{lemma}[Convex power-sum inequality]
\label{lem:convex-power-sum}
Let \(x,y\geq0\) and \(\theta>1\). Then
\(x^\theta+y^\theta
\geq2^{1-\theta}(x+y)^\theta\).
Moreover, equality holds if and only if \(x=y\).
\end{lemma}

\begin{proof} Since $\theta>1$, the function $\varphi(z)=z^\theta$ is strictly convex on $[0,\infty)$. Hence, Jensen's inequality with equal weights gives $\varphi\!\left((x+y)/2\right)\leq\frac{1}{2}\varphi(x)+\frac{1}{2}\varphi(y)$, that is, $\left((x+y)/2\right)^\theta\leq\frac{1}{2}(x^\theta+y^\theta)$. Multiplying both sides by $2$ yields $x^\theta+y^\theta\geq2^{1-\theta}(x+y)^\theta$. Since $\varphi$ is strictly convex, equality in Jensen's inequality holds if and only if $x=y$, which completes the proof. \end{proof}

By Lemma~\ref{lem:convex-power-sum}, since \(\theta_i>1\),
\(\bar V_i^{\theta_i}+(\lambda_{Wi}V_{Wi})^{\theta_i}
\ge2^{1-\theta_i}\left(\bar V_i+\lambda_{Wi}V_{Wi}\right)^{\theta_i}
=2^{1-\theta_i}\mathcal J_i^{\theta_i}\).
Therefore,
\begin{equation}
    \dot{\mathcal J}_i
    \le
    -c_{i1}\mathcal J_i^{\rho_i}
    -
    c_{i2}\mathcal J_i^{\theta_i}
    +
    \Delta_{ci},
    \label{eq:proof3_final_J_ineq}
\end{equation}
where \(c_{i1}=\bar A_i\) and \(c_{i2}=2^{1-\theta_i}\bar B_i\).
By Theorem~\ref{thm:value_fixed_time_condition},
\(\mathcal J_i(t)\) reaches \(\Omega_{ci}\) within the fixed-time.

Finally, the composite residual bound is converted into a pointwise
bound on the local coordination error and then into a bound on the
leader-rooted formation error. It remains to relate \(\mathcal J_i\)
to the formation error. Let \(r_{ci}>0\) satisfy
\(c_{i1}r_{ci}^{\rho_i}+c_{i2}r_{ci}^{\theta_i}\ge\Delta_{ci}\).
For \(t\ge T_{ci}\), one has
\(\bar V_i(t)\le\mathcal J_i(t)\le r_{ci}\).
Using \eqref{eq:proof3_Vdot_bound}, for any
\(\tau\in[t-T_i,t]\),
\begin{equation}
    V_i^\ast(\chi_i(\tau))
    \ge
    V_i^\ast(\chi_i(t))
    -
    L_{V_i}(t-\tau).
    \label{eq:proof3_backward_value}
\end{equation}
Integrating \eqref{eq:proof3_backward_value} over
\([t-T_i,t]\) yields
\(\bar V_i(t)\ge T_iV_i^\ast(\chi_i(t))-\frac{1}{2}L_{V_i}T_i^2\).
Hence,
\(V_i^\ast(\chi_i(t))
\le\frac{\bar V_i(t)}{T_i}+\frac{1}{2}L_{V_i}T_i
\le\frac{r_{ci}}{T_i}+\frac{1}{2}L_{V_i}T_i\).
From the lower bound in Lemma~\ref{lem:local_value_bounds},
\(\|\chi_i(t)\|^2
\le\frac{1}{\underline c_i}
\left(\frac{r_{ci}}{T_i}+\frac{1}{2}L_{V_i}T_i\right)\),
\(t\ge T_{ci}\).

With \(\chi=\operatorname{col}(\chi_1,\ldots,\chi_N)\), the global
coordination-error norm satisfies
\(\|\chi(t)\|^2=\sum_{i=1}^{N}\|\chi_i(t)\|^2\).
Taking \(T_{\rm cl}=\max_{1\le i\le N}T_{ci}\), for all
\(t\ge T_{\rm cl}\),
\begin{equation}
    \|\chi(t)\|^2
    \le
    \sum_{i=1}^{N}
    \frac{1}{\underline c_i}
    \left(
    \frac{r_{ci}}{T_i}
    +
    \frac{1}{2}L_{V_i}T_i
    \right).
    \label{eq:proof3_eta_global}
\end{equation}
Finally, \(\chi=(\mathbf H\otimes\mathbf I_n)e\).
Since the leader is globally reachable, \(\mathbf H\) is nonsingular.
Hence,
\begin{equation}
    \|e(t)\|
    \le
    \|(\mathbf H^{-1}\otimes\mathbf I_n)\|
    \|\chi(t)\|.
    \label{eq:proof3_e_bound}
\end{equation}
Equations \eqref{eq:proof3_eta_global} and
\eqref{eq:proof3_e_bound} prove fixed-time practical boundedness of the
leader-rooted formation error. Moreover, since
\(\mathcal J_i(t)\ge\lambda_{Wi}V_{Wi}(t)\), the critic error is
fixed-time practically bounded as well.

For completeness, the uniform practical settling-time estimate can
be written explicitly. Fix any \(\epsilon_{c_i}\in(0,1)\) and define
\(\Omega_{ci}^{\epsilon_{c_i}}
:=
\left\{
\mathcal J_i\ge0:
c_{i1}\mathcal J_i^{\rho_i}
+
c_{i2}\mathcal J_i^{\theta_i}
\le
\frac{\Delta_{ci}}{1-\epsilon_{c_i}}
\right\}\).
Outside this set, the residual term can be absorbed by the two
negative powers. From \eqref{eq:proof3_final_J_ineq}, one has, for
every \(\mathcal J_i\notin\Omega_{ci}^{\epsilon_{c_i}}\),
\begin{equation}
\begin{aligned}
    \dot{\mathcal J}_i
    &\le
    -c_{i1}\mathcal J_i^{\rho_i}
    -
    c_{i2}\mathcal J_i^{\theta_i}
    +
    \Delta_{ci}                                                \\
    &\le
    -\epsilon_{c_i} c_{i1}\mathcal J_i^{\rho_i}
    -
    \epsilon_{c_i} c_{i2}\mathcal J_i^{\theta_i}.
\end{aligned}
\label{eq:Jci_outside_Omega_eps}
\end{equation}

Let \(T_{ci}\) be the first time such that
\(\mathcal J_i(T_{ci})\in\Omega_{ci}^{\epsilon_{c_i}}\).
For all \(t<T_{ci}\), \eqref{eq:Jci_outside_Omega_eps} gives
\(dt\leq-\dfrac{d\mathcal J_i}
{\epsilon_{c_i}\left(
c_{i1}\mathcal J_i^{\rho_i}
+
c_{i2}\mathcal J_i^{\theta_i}
\right)}\).
Therefore,
\(T_{ci}\leq
\dfrac{1}{\epsilon_{c_i}}
\int_{\mathcal J_i(T_{ci})}^{\mathcal J_i(0)}
\dfrac{ds}{c_{i1}s^{\rho_i}+c_{i2}s^{\theta_i}}
\leq
\dfrac{1}{\epsilon_{c_i}}
\int_{0}^{\infty}
\dfrac{ds}{c_{i1}s^{\rho_i}+c_{i2}s^{\theta_i}}\).
Splitting the last integral over \((0,1]\) and \([1,\infty)\), and
using \(0<\rho_i<1\), \(\theta_i>1\), yields
\(\int_{0}^{\infty}
\dfrac{ds}{c_{i1}s^{\rho_i}+c_{i2}s^{\theta_i}}
\leq
\int_{0}^{1}\dfrac{ds}{c_{i1}s^{\rho_i}}
+
\int_{1}^{\infty}\dfrac{ds}{c_{i2}s^{\theta_i}}
=
\dfrac{1}{c_{i1}(1-\rho_i)}
+
\dfrac{1}{c_{i2}(\theta_i-1)}\).

Consequently,
\begin{equation}
    T_{ci}
    \le
    \frac{1}{\epsilon_{c_i} c_{i1}(1-\rho_i)}
    +
    \frac{1}{\epsilon_{c_i} c_{i2}(\theta_i-1)} .
    \label{eq:Jci_settling_time_eps}
\end{equation}
Since the right-hand side of
\eqref{eq:Jci_settling_time_eps} is independent of
\(\mathcal J_i(0)\), \(\mathcal J_i(t)\) reaches
\(\Omega_{ci}^{\epsilon_{c_i}}\) in fixed time.

The proof is complete.
\end{proof}


\begingroup
\setlength{\abovedisplayskip}{2pt plus 1pt minus 1pt}
\setlength{\belowdisplayskip}{2pt plus 1pt minus 1pt}
\setlength{\abovedisplayshortskip}{1pt}
\setlength{\belowdisplayshortskip}{1pt}
\setlength{\jot}{0pt}

\section{Simulation}

\FloatBarrier
\subsection{Simulation Setup}

Consider one leader and four nonlinear followers evolving in the two-dimensional plane. The state of follower \(i\) is \(x_i=\operatorname{col}(p_{ix},p_{iy},v_{ix},v_{iy})\in\mathbb R^4\), and the follower dynamics are \(\dot x_i=f_i(t,x_i)+\mathbf g_i(x_i)u_{ci}+\mathbf g_i(x_i)u_{ai}+\mathbf d_i\omega_i\), \(i\in\{1,\ldots,4\}\), where \(u_{ci}\in\mathbb R^2\) is the secure control input, \(u_{ai}\in\mathbb R^2\) is the actuator-side FDI signal, and \(\omega_i\in\mathbb R^2\) is the external disturbance. The unknown nonlinear drift is \(f_i(t,x_i)=\operatorname{col}(v_{ix},v_{iy},F_{ix},F_{iy})\), where \(F_{ix}=-0.30v_{ix}+0.16\sin(0.45p_{ix})+0.05\sin(p_{iy}v_{ix}/8)+0.035\cos(0.80t+\vartheta_i)+c_i\), \(F_{iy}=-0.28v_{iy}+0.15\sin(0.42p_{iy})+0.05\sin(p_{ix}v_{iy}/8)+0.035\sin(0.70t+\vartheta_i)-c_i\), \(c_i=0.02v_{ix}v_{iy}/[1+0.20(v_{ix}^2+v_{iy}^2)]\), and \(\vartheta_i=0.7(i-1)\). The state-dependent input matrix is
\(\mathbf g_i(x_i)=\operatorname{col}\!\left(\mathbf 0_{2\times2},
\operatorname{diag}(g_{ix},g_{iy})\right)\),, where \(g_{ix}=1+0.08\cos[0.30p_{ix}+0.20(i-1)]\) and \(g_{iy}=0.95+0.08\sin[0.25p_{iy}-0.20(i-1)]\), while
\(\mathbf d_i=\operatorname{col}(\mathbf 0_{2\times2},\mathbf I_2)\). The disturbance is \(\omega_i(t)=\operatorname{col}(\omega_{ix},\omega_{iy})\), where \(\omega_{ix}=0.045\sin([1.30+0.10(i-1)]t)+0.018\cos([2.20+0.07(i-1)]t+0.30(i-1))\) and \(\omega_{iy}=0.040\cos([1.50+0.08(i-1)]t)+0.020\sin([2.45+0.09(i-1)]t+0.20(i-1))\). The FDI attack is activated at \(t_a=8\,\mathrm{s}\), with \(u_{ai}(t)=\mathbf 0_2\) for \(t<t_a\) and \(u_{ai}(t)=\operatorname{col}(u_{ai,x},u_{ai,y})\) for \(t\ge t_a\), where \(u_{ai,x}=a_{xi}^{(1)}\sin(\omega_{xi}^{(1)}t+\vartheta_i)+a_{xi}^{(2)}\cos(\omega_{xi}^{(2)}t+\varphi_i)\), \(u_{ai,y}=a_{yi}^{(1)}\cos(\omega_{yi}^{(1)}t+\varphi_i)+a_{yi}^{(2)}\sin(\omega_{yi}^{(2)}t+\vartheta_i)\), \(\vartheta_i=0.7(i-1)\), and \(\varphi_i=0.45+0.30(i-1)\). The heterogeneous attack amplitudes are \(a_x^{(1)}=\operatorname{col}(0.32,0.46,0.39,0.54)\), \(a_x^{(2)}=\operatorname{col}(0.15,0.11,0.18,0.13)\), \(a_y^{(1)}=\operatorname{col}(0.41,0.34,0.52,0.47)\), and \(a_y^{(2)}=\operatorname{col}(0.12,0.19,0.14,0.17)\), whereas the corresponding frequency vectors are \(\omega_x^{(1)}=\operatorname{col}(0.83,1.07,0.76,1.23)\), \(\omega_x^{(2)}=\operatorname{col}(1.64,1.41,1.89,1.52)\), \(\omega_y^{(1)}=\operatorname{col}(0.91,1.18,0.87,1.31)\), and \(\omega_y^{(2)}=\operatorname{col}(1.72,1.55,1.96,1.63)\).

The leader state is \(x_0=\operatorname{col}(p_{0x},p_{0y},v_{0x},v_{0y})\), with \(\dot p_{0x}=v_{0x}\), \(\dot p_{0y}=v_{0y}\), \(\dot v_{0x}=-0.035p_{0x}-0.160v_{0x}+0.480\cos(0.32t)\), and \(\dot v_{0y}=-0.030p_{0y}-0.180v_{0y}+0.420\sin(0.28t)\). The desired offsets are \(h_1=\operatorname{col}(-10,-6)\), \(h_2=\operatorname{col}(10,-6)\), \(h_3=\operatorname{col}(-10,6)\), and \(h_4=\operatorname{col}(10,6)\) (The model is used solely to generate simulation data and is not required
by the proposed learning controller). The adjacency and leader-pinning matrices are \(\mathbf A=\begin{bmatrix}0&1&0&1\\1&0&1&0\\0&1&0&1\\1&0&1&0\end{bmatrix}\) and \(\mathbf B=\operatorname{diag}(1.4,0,0,1.2)\), respectively, where \(\mathbf L\) denotes the Laplacian matrix and \(\mathbf H=\mathbf L+\mathbf B\). 



\FloatBarrier


For each follower, the critic is implemented by a six-node Gaussian radial-basis-function neural network (RBFNN) as \(\hat V_i(\chi_i)=\hat W_i^\top\phi_i(\chi_i)\), where \(\hat W_i\in\mathbb R^6\) is the critic-weight vector and \(\chi_i=\operatorname{col}(\chi_{pi},\chi_{vi})\in\mathbb R^4\) is the local distributed coordination error. To improve numerical conditioning, define the normalized critic input as \(\bar\chi_i=\operatorname{col}(\chi_{pi}/5,\chi_{vi}/3)\). The RBF vector is \(\phi_i(\chi_i)=\operatorname{col}(\phi_{i1},\ldots,\phi_{i6})\), with \(\phi_{ik}(\chi_i)=\exp\!\left(-\|\bar\chi_i-c_{ik}\|^2/\sigma_{ik}^2\right)\), \(k\in\{1,\ldots,6\}\), where \(c_{ik}\in\mathbb R^4\) and \(\sigma_{ik}>0\) denote the center and width of the \(k\)th Gaussian basis function, respectively. The six centers are distributed over the normalized operating region of \(\bar\chi_i\), and the widths are selected sufficiently large to provide overlapping activation over the admissible state domain. 


The replay stack is collected during the pre-attack interval \(t\in[0,8]\,\mathrm{s}\) and retained thereafter to preserve the initial-excitation information required for critic learning.

The initial conditions are \(x_0(0)=\operatorname{col}(0,0,1.2,1.6)\), \(x_1(0)=\operatorname{col}(-5,-10,2.2,0.8)\), \(x_2(0)=\operatorname{col}(6,-1,0.5,2.5)\), \(x_3(0)=\operatorname{col}(-4,9,2.0,2.3)\), and \(x_4(0)=\operatorname{col}(5,2,0.3,1.0)\). 

For the numerical fixed-time comparison, the initial formation-error perturbations are additionally scaled by \(s\in\{0.5,1.0,1.5,2.0,2.5\}\), where \(s\) denotes the initial-error scaling factor.
The model is used solely to generate data during the simulation process.

The simulation is performed over
\(t\in[0,25]\,\mathrm{s}\) with integration step
\(T_s=0.01\,\mathrm{s}\) using a fourth-order Runge--Kutta method.
The integral Bellman window is
\(T_i=0.25\,\mathrm{s}\), the replay sampling period is
\(0.05\,\mathrm{s}\), and at most \(80\) replay samples are retained.
Replay data are collected over \(t\in[0,8]\,\mathrm{s}\).
The critic dimension is \(N_c=6\), with initial weights
\(\hat W_i(0)=50\mathbf 1_6\).
The critic powers are \(p_i=0.65\) and \(q_i=1.70\), with gains
\(k_{i1}=1.20\) and \(k_{i2}=0.18\).
The fixed-time powers are
\(\alpha_i=0.65\) and \(\beta_i=1.35\).
The input penalty and saturation bound are
\(\mathbf R_i=0.08\mathbf I_2\) and
\(\bar{\mathbf U}_i=40\mathbf I_2\), respectively.
The FDI signal is activated at
\(t_a=8\,\mathrm{s}\).

\FloatBarrier

\subsection{Simulation Results}

Figures~\ref{fig:tracking-states} and \ref{fig:formation-trajectory} show the closed-loop leader--follower responses. Figures~\ref{fig:critic-weights}--\ref{fig:fixed-time-verification} summarize the learning and convergence behavior. The six critic weights of each follower reach bounded constant values in approximately \(5.7\,\mathrm{s}\) without subsequent drift. Meanwhile, the position errors \(e_{pi}=p_i-p_0-h_i\), initially spanning approximately \([-5,6]\) along \(x\) and \([-4,5]\) along \(y\), are driven close to the origin, and \(\|e_{vi}\|\) remains small after the transient. More importantly, increasing the initial-error scale from \(0.5\) to \(2.5\) increases \(\|\chi(0)\|\) from approximately \(30\) to above \(150\), whereas the observed settling times are only \(9.53\), \(10.05\), \(10.33\), \(10.55\), and \(10.76\,\mathrm{s}\), yielding the common numerical bound \(T_{\mathrm{FT}}^{\mathrm{obs}}=10.76\,\mathrm{s}\).



\begin{figure}[H]
\centering
\includegraphics[width=\columnwidth]{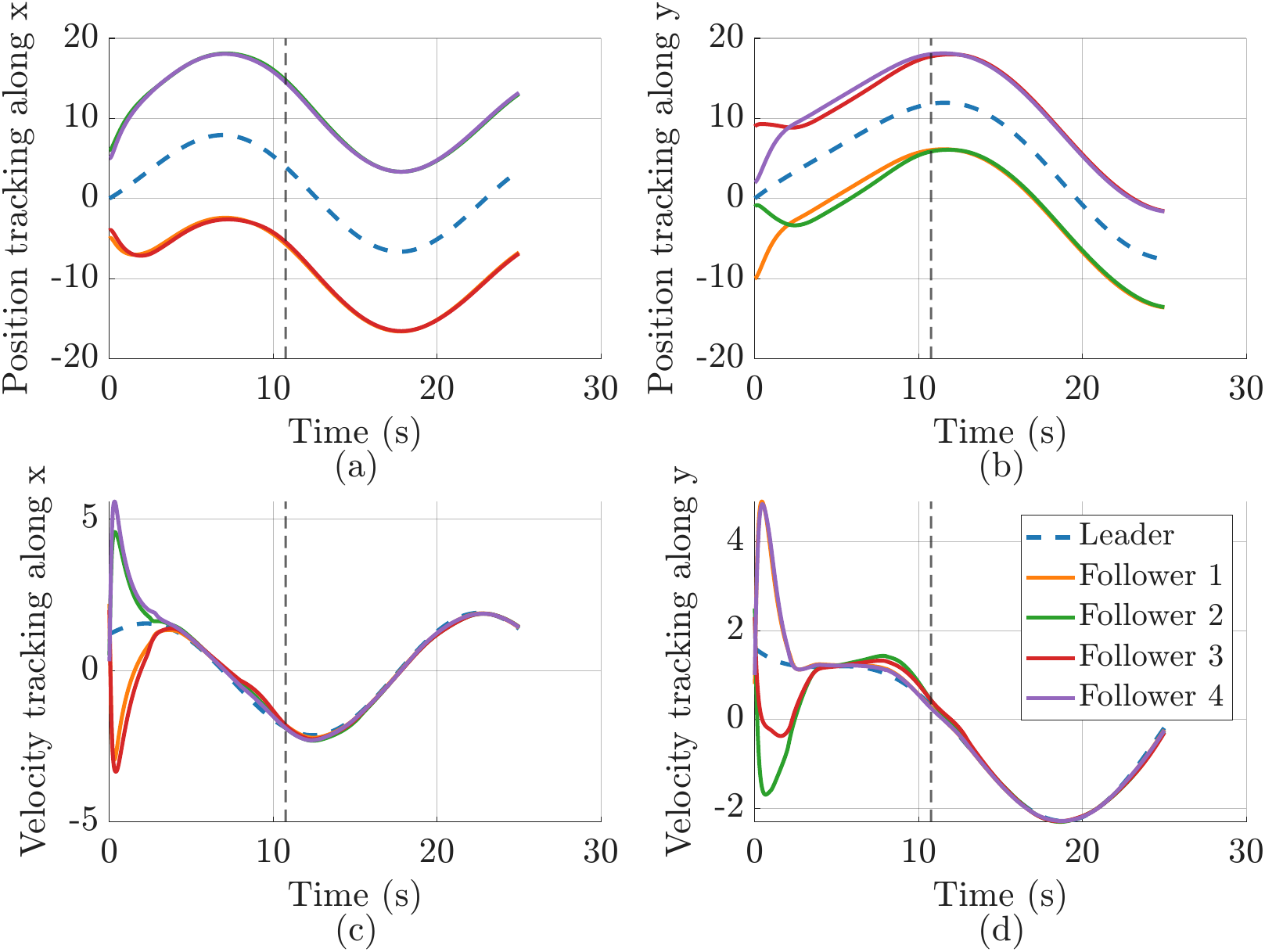}
\caption{Leader--follower position and velocity tracking responses.}
\label{fig:tracking-states}
\end{figure}

\begin{figure}[H]
\centering
\includegraphics[width=\columnwidth]{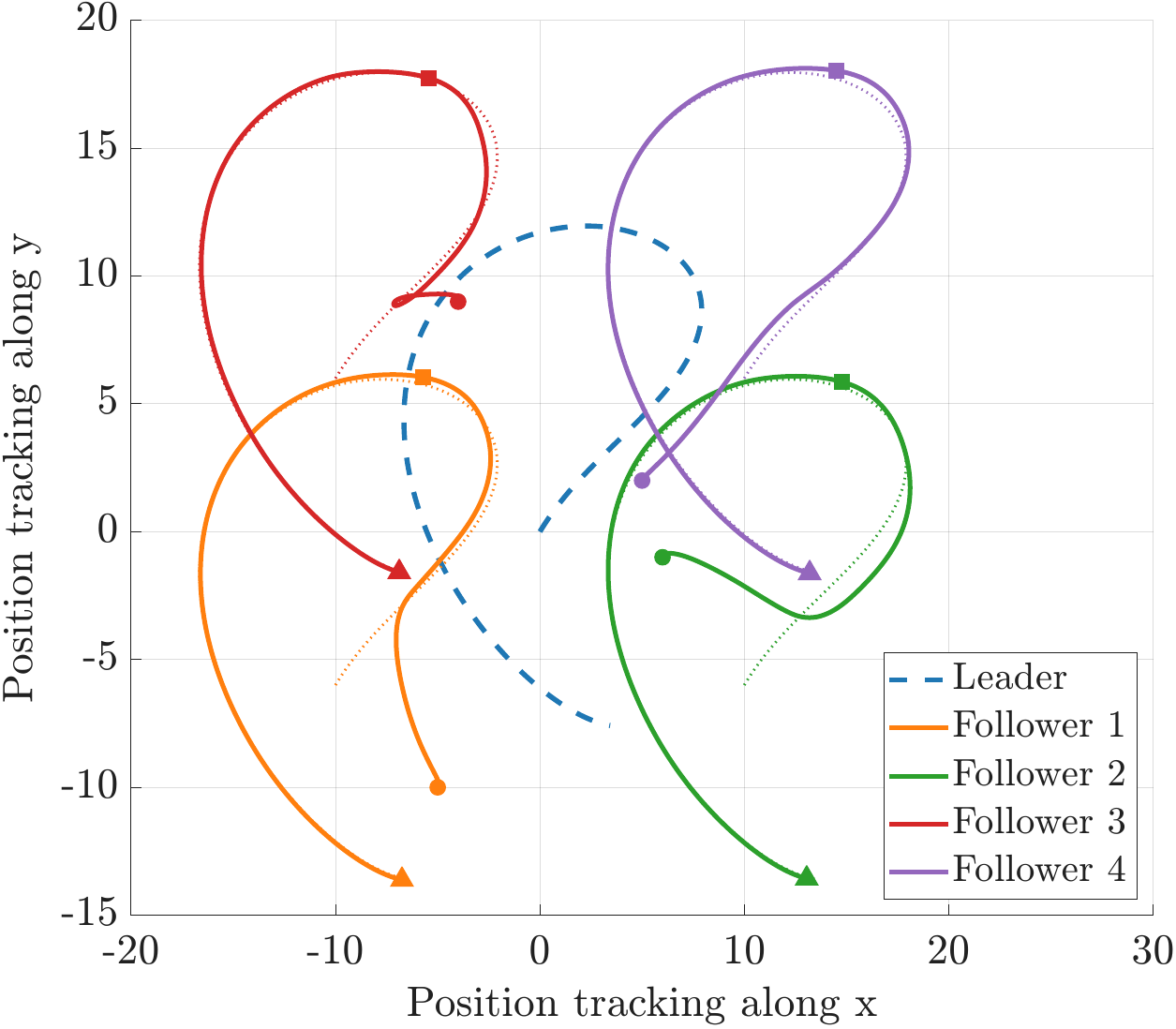}
\caption{Two-dimensional leader--follower formation trajectories.}
\label{fig:formation-trajectory}
\end{figure}

\begin{figure}[H]
\centering
\includegraphics[width=\columnwidth]{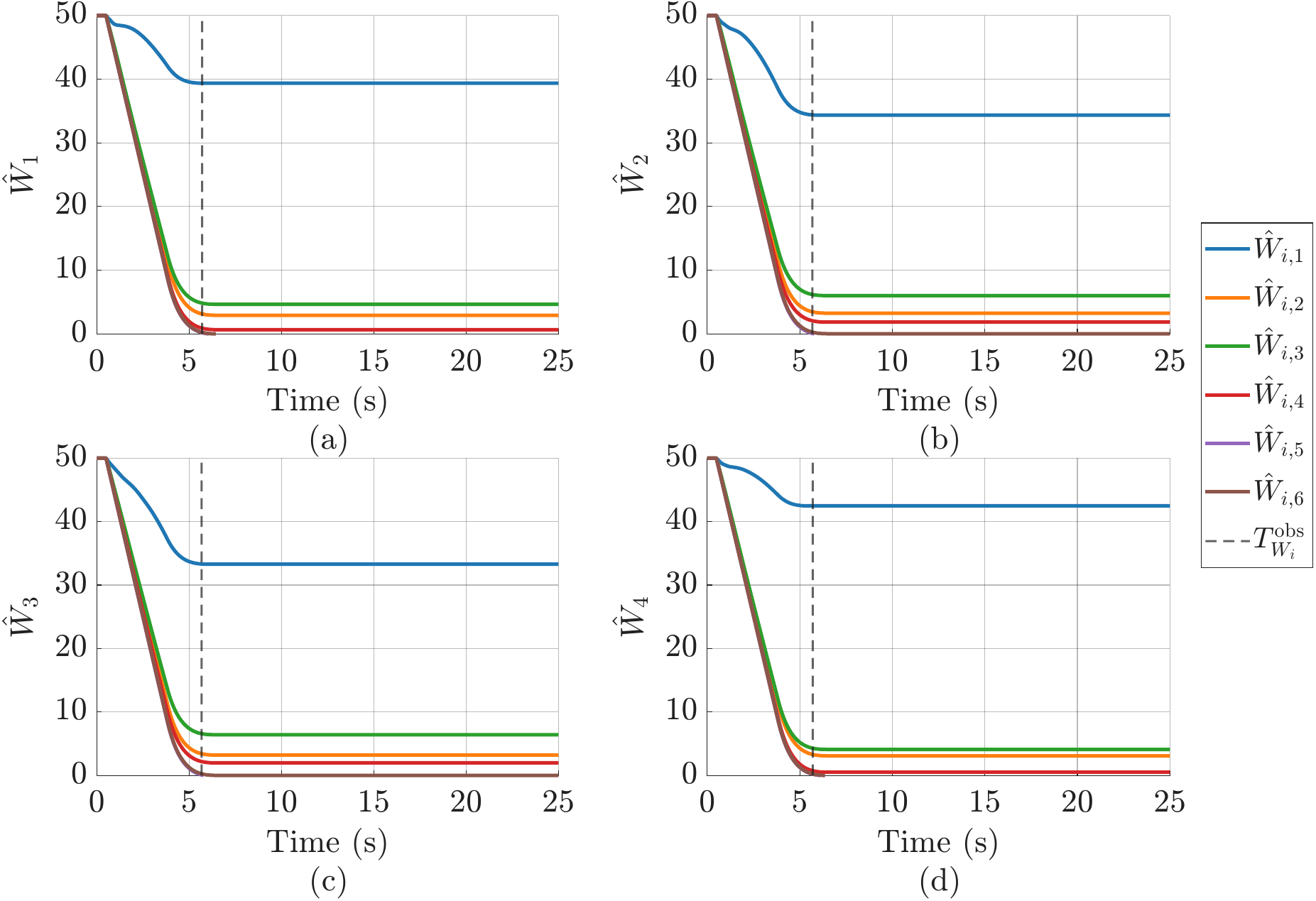}
\caption{Evolution of the critic weights \(\hat{W}_i\) for each follower:
(a) follower~1, (b) follower~2, (c) follower~3, (d) follower~4.}
\label{fig:critic-weights}
\end{figure}


\begin{figure}[H]
\centering
\includegraphics[width=\columnwidth]{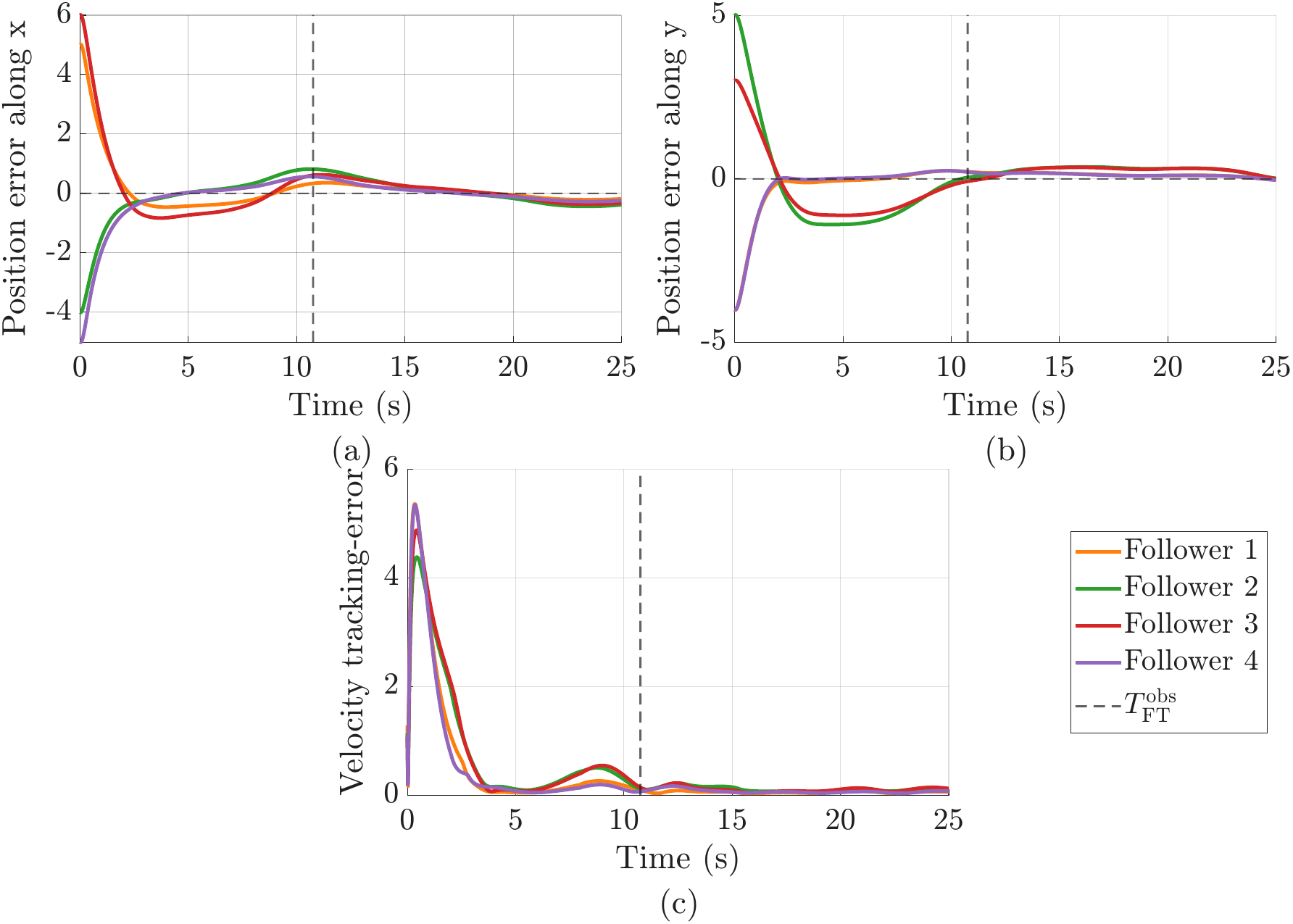}
\caption{Formation position and velocity tracking errors.}
\label{fig:tracking-errors}
\end{figure}

\begin{remark}
The convergence times in Fig.~\ref{fig:critic-weights} are numerical settling times, not the exact theoretical fixed-time bounds. The analysis guarantees an initial-condition-independent fixed-time bound, although its exact value depends on Lyapunov and replay-informativity constants that may not be known a priori. This motivates a predefined-time extension in which the desired convergence time is explicitly assigned by the designer.
\end{remark}

\begin{figure}[H]
\centering
\includegraphics[width=0.85\columnwidth]{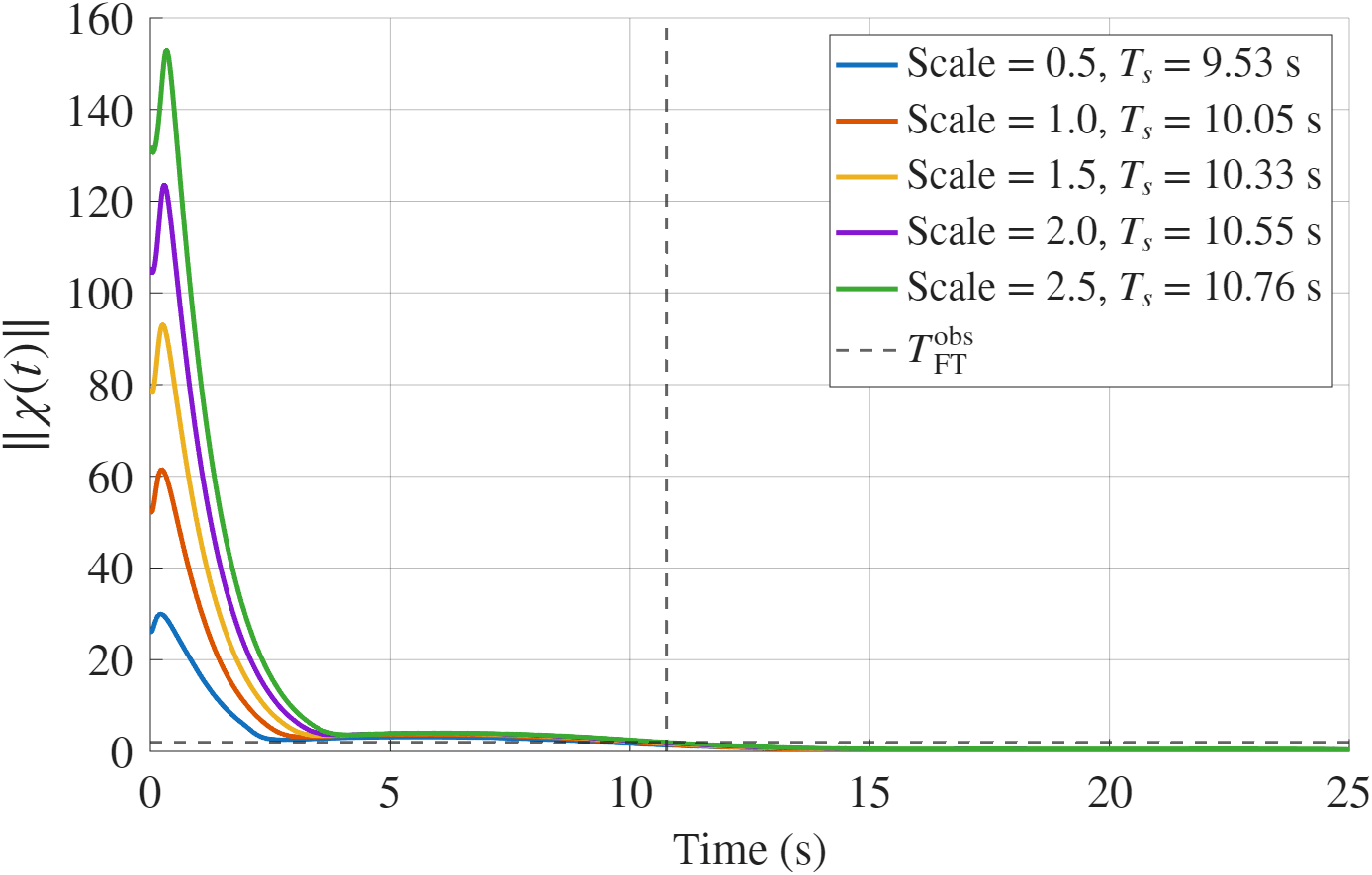}
\caption{Practical fixed-time verification under different initial-error scales.}
\label{fig:fixed-time-verification}
\end{figure}


Figures~\ref{fig:control-inputs}--\ref{fig:disturbances} further confirm the constrained and resilient behavior. The secure inputs \(u_{ci}=\operatorname{col}(u_{ci,x},u_{ci,y})\) remain strictly within the prescribed bound \(|u_{ci,\ell}|<40\), with peak magnitudes of approximately \(39\) and \(32\) along the \(x\)- and \(y\)-channels. At \(t_a=8\,\mathrm{s}\), heterogeneous FDI signals with magnitudes up to approximately \(0.67\) are activated and remain time varying thereafter, while the external disturbances persist within approximately \([-0.063,0.063]\). The bounded tracking errors in Fig.~\ref{fig:tracking-errors} are therefore maintained under simultaneous nonlinear dynamics, persistent disturbances, actuator FDI attacks, and input constraints.

\begin{figure}[H]
\centering
\includegraphics[width=\columnwidth]{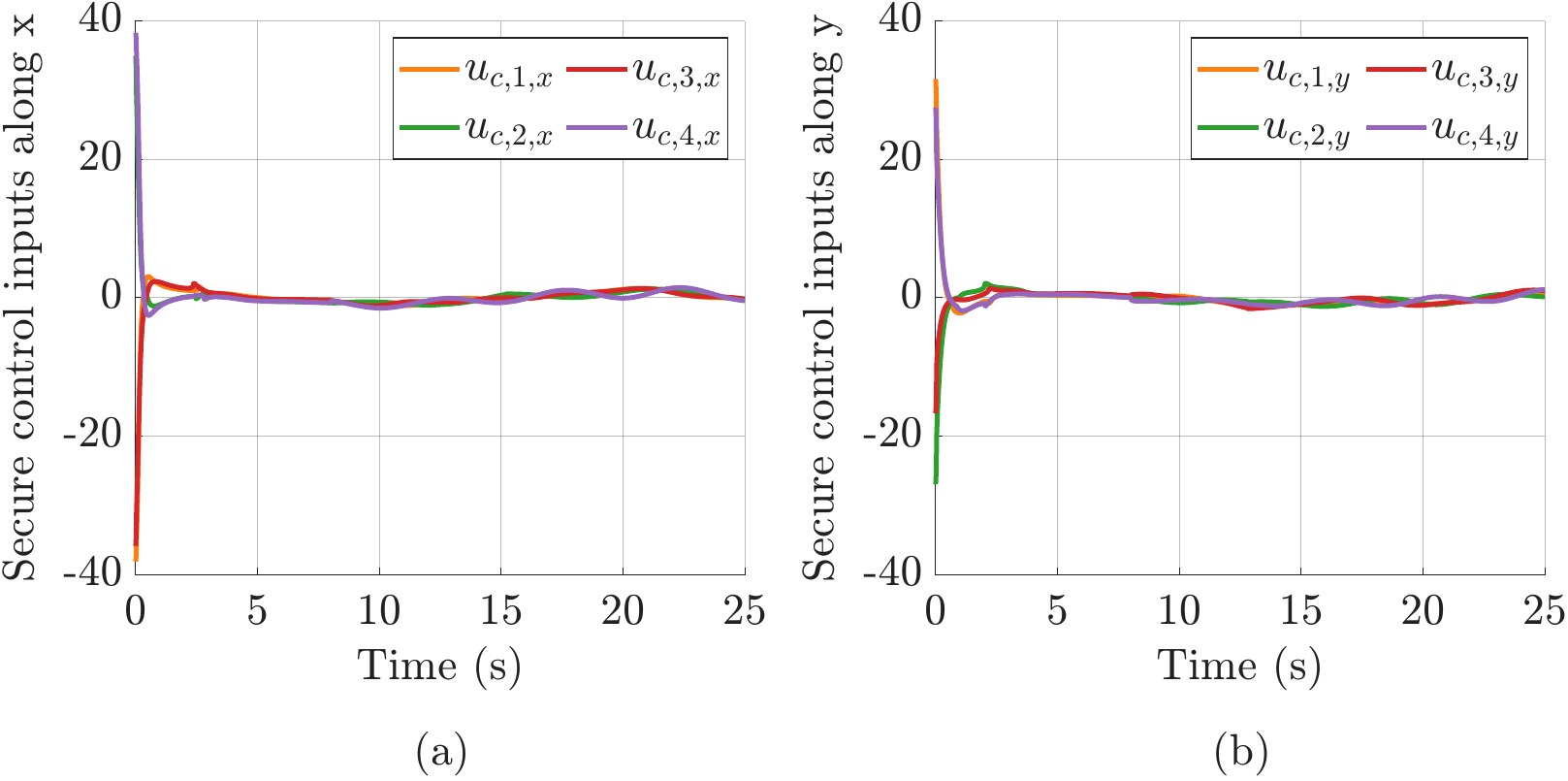}
\caption{Secure control inputs under the prescribed input constraint.}
\label{fig:control-inputs}
\end{figure}

\begin{figure}[H]
\centering
\includegraphics[width=\columnwidth]{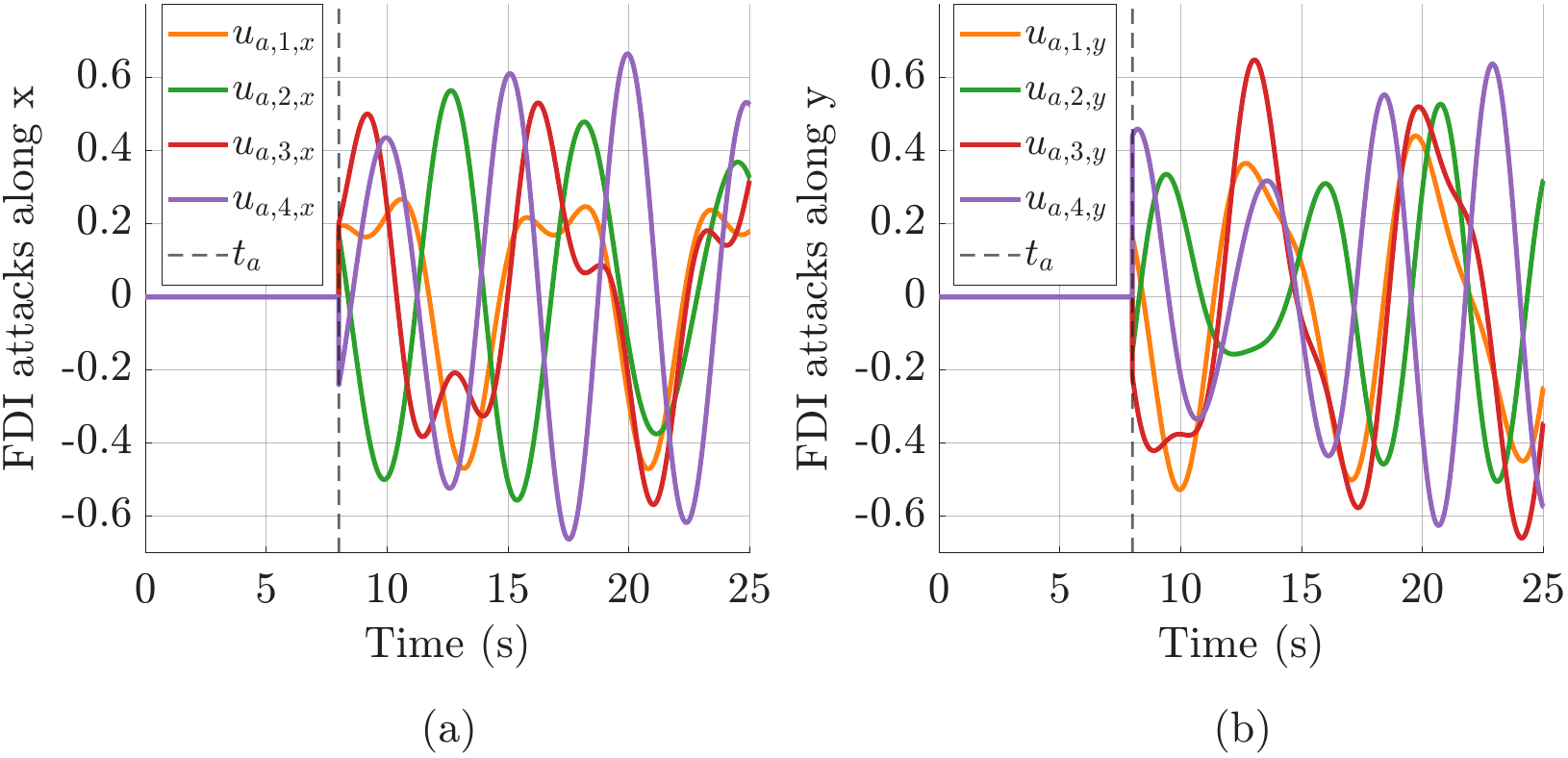}
\caption{Heterogeneous actuator-side FDI attacks.}
\label{fig:fdi-signals}
\end{figure}

\begin{figure}[H]
\centering
\includegraphics[width=\columnwidth]{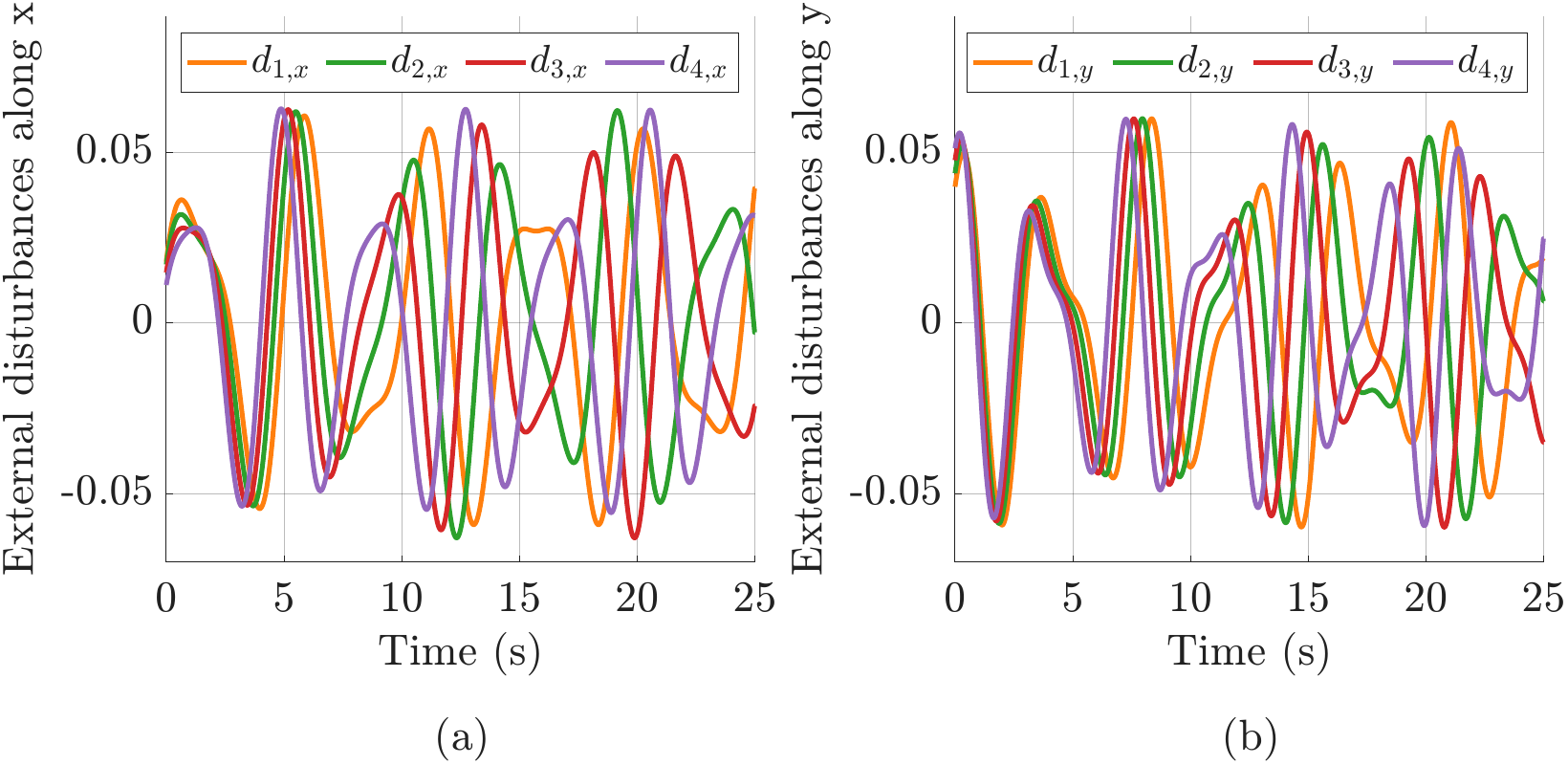}
\caption{External disturbances acting on the followers, where $d_{i,x}$ and $d_{i,y}$ denote the $x$- and $y$-components of the disturbance signal $\boldsymbol{\omega}_i$, respectively.}
\label{fig:disturbances}
\end{figure}

\FloatBarrier

\section{Discussion and Limitations}
\label{sec:discussion}

The proposed controller is constructive, whereas the stability certificate remains partly nonconstructive. In particular, the comparison constants associated with \(V_i^{*}\) and \(\nabla V_i^{*}\), used in Lemma~\ref{lem:local_value_bounds} and Remark~\ref{rem:cost_induced_fixed_time}, are guaranteed to exist on the compact nonterminal region \(\Omega_i^{r}\), but cannot yet be systematically computed from the available system information. Hence, their relation with the cost-shaping coefficients \(\kappa_{i1}\) and \(\kappa_{i2}\) is not fully constructive, and these coefficients are conservatively selected by empirical tuning; since they enter \(Q_{ii}\), they also affect the Bellman--Isaacs residual and critic learning. Moreover, bounded approximation and adversarial residuals yield practical fixed-time convergence to a prescribed neighborhood rather than exact convergence to the origin under persistent perturbations. Deriving computable comparison bounds directly from data, system constraints, and certified envelopes of the unknown HJI value function remains a topic for future investigation.

A second issue concerns admissible initialization, which is fundamental
to undiscounted infinite-horizon HJI formulations. Rather than merely
postulating that \(\Psi_i(\Omega_i)\neq\varnothing\), Appendix~\ref{app:admissible_policy}
outlines how the replay data already collected in
Eqs.~\eqref{eq:history_stack_short}--\eqref{eq:history_integral_short}, and required to satisfy the finite-excitation
condition of Lemma~\ref{lem:ie_condition}, can be reused to construct a Koopman--Riccati
policy \(u_{ci}^{(0)}\in\Psi_i(\Omega_i^0)\). This construction is
motivated by data-driven Koopman control
\cite{KordaMezic2018KoopmanMPC,StrasserBerberichAllgower2023}
and the Koopman--Hamiltonian--Jacobi connection in
\cite{Vaidya2025KoopmanHamiltonJacobi}. Nevertheless, the resulting
certificate is local and does not imply that every subsequent
critic-only policy update is an exact admissibility-preserving policy
iteration step. Establishing data-driven admissibility preservation
throughout online learning, together with constructive evaluation of
the HJI comparison constants, constitutes an important direction for
future work.


\section{Conclusion}


This paper developed a fixed-time integral reinforcement learning framework for nonlinear leader--follower formation under unknown dynamics, disturbances, FDI attacks, and bounded inputs. A nonquadratic utility enforced input constraints, while a two-power cost and replay-based critic update guaranteed fixed-time convergence of the critic and formation errors to compact residual sets, with settling times independent of initial conditions. Avoiding explicit knowledge of the drift, simulations confirmed resilient formation tracking under persistent disturbances and time-varying FDI attacks.

\section*{Acknowledgment}
The authors would like to thank Ho Chi Minh City University of Technology (HCMUT) and Vietnam National University Ho Chi Minh City (VNU-HCM) for supporting this research.

\appendices
\section{Data-Reused Construction of an Initial Admissible Policy}
\label{app:admissible_policy}

The construction below reuses the finite trajectory windows already
stored in \(\mathcal D_i\) in
Eqs.~\eqref{eq:history_stack_short}--\eqref{eq:history_integral_short};
hence, no additional exploration is required.
These data are already used in Lemma~\ref{lem:ie_condition} for critic
learning.
Motivated by data-driven Koopman control
\cite{KordaMezic2018KoopmanMPC,StrasserBerberichAllgower2023}
and the Koopman--Hamiltonian--Jacobi connection in
\cite{Vaidya2025KoopmanHamiltonJacobi}, the same data are reused only
to construct an admissible warm start for the HJI branch.
The critic finite-excitation condition and the Koopman rank condition
below are distinct; therefore, the stored data are checked separately
for both properties.

Let
\(\bm\eta_i=\bm{\mathcal L}_{K_i}(\chi_i)\in\mathbb R^{N_{K_i}}\),
where \(N_{K_i}\) is the lifting dimension, with
\(\bm{\mathcal L}_{K_i}(0)=0\) and
\(\underline c_{\mathcal L_i}\|\chi_i\|
\le\|\bm\eta_i\|
\le\bar c_{\mathcal L_i}\|\chi_i\|\)
on a compact neighborhood
\(\Omega_i^0\subseteq\Omega_i\).
Thus, \(\bm\eta_i=0\) if and only if \(\chi_i=0\) on
\(\Omega_i^0\).

For the same sampling windows used in \(\mathcal D_i\), define
\(\Delta\bm\eta_{ik}
:=\bm\eta_i(t_k)-\bm\eta_i(t_k-T_i)\),
\(\bm\eta_{I,ik}
:=\int_{t_k-T_i}^{t_k}\bm\eta_i(\tau)d\tau\), and
\(\bm u_{I,ik}
:=\int_{t_k-T_i}^{t_k}u_{ci}(\tau)d\tau\).
Let
\(\Delta\bm H_i
:=[\Delta\bm\eta_{i1},\ldots,\Delta\bm\eta_{iM_i}]\),
\(\bm H_{I,i}
:=[\bm\eta_{I,i1},\ldots,\bm\eta_{I,iM_i}]\),
\(\bm U_{I,i}
:=[\bm u_{I,i1},\ldots,\bm u_{I,iM_i}]\), and
\(\bm{\mathcal R}_{K_i}
:=\operatorname{col}\{\bm H_{I,i},\bm U_{I,i}\}\).
Select an informative substack satisfying
\(\operatorname{rank}(\bm{\mathcal R}_{K_i})=N_{K_i}+m\),
which requires \(M_i\ge N_{K_i}+m\), and obtain
\([\bm A_{K_i}\ \bm B_{K_i}]
=\Delta\bm H_i\bm{\mathcal R}_{K_i}^{\dagger}\),
where \((\cdot)^\dagger\) denotes the Moore--Penrose pseudoinverse.
Equivalently,
\(\Delta\bm\eta_{ik}
=\bm A_{K_i}\bm\eta_{I,ik}
+\bm B_{K_i}\bm u_{I,ik}
+\bm E_{K_i}^{(k)}\),
where \(\bm E_{K_i}^{(k)}\) is the finite-data fitting residual.

Thus, on \(\Omega_i^0\), the nominal lifted dynamics are represented as
\(\dot{\bm\eta}_i
=\bm A_{K_i}\bm\eta_i+\bm B_{K_i}u_{ci}+\bm d_{K_i}\),
where
\(\bm d_{K_i}\in\mathbb R^{N_{K_i}}\)
collects the finite-dimensional lifting and identification errors and
is not an additional physical disturbance.

Let
\(\bm Q_{K_i}=\bm Q_{K_i}^{\top}>0\) and
\(\bm R_{K_i}=\bm R_{K_i}^{\top}>0\), and let
\(\bm P_i=\bm P_i^\top>0\) be the stabilizing solution of
\(\bm A_{K_i}^{\top}\bm P_i+\bm P_i\bm A_{K_i}
-\bm P_i\bm B_{K_i}\bm R_{K_i}^{-1}
\bm B_{K_i}^{\top}\bm P_i+\bm Q_{K_i}=0\).
Define
\(\bm K_{0i}:=\bm R_{K_i}^{-1}\bm B_{K_i}^{\top}\bm P_i\),
\(V_{0i}:=\bm\eta_i^\top\bm P_i\bm\eta_i\), and
\(u_{ci}^{(0)}:=-\bm K_{0i}\bm\eta_i\).
The set \(\Omega_i^0\) is selected as a compact sublevel set of
\(V_{0i}\), contained in the validity region of the lifting and residual
bounds, and sufficiently small such that
\(u_{ci}^{(0)}(\chi_i)\in\mathbb U_i\) for all
\(\chi_i\in\Omega_i^0\).

\begin{lemma}[Koopman--Riccati admissible initialization]
\label{lem:koopman_admissible}
Suppose that a validation subset of the stored data together with the
local regularity bound yields
\(\|\bm d_{K_i}\|\le\rho_{K_i}\|\bm\eta_i\|\) on
\(\Omega_i^0\), where \(\rho_{K_i}\ge0\), and define
\(\alpha_{K_i}
:=\lambda_{\min}(\bm Q_{K_i})
-2\|\bm P_i\|\rho_{K_i}>0\).
Then
\(u_{ci}^{(0)}\in\Psi_i(\Omega_i^0)\).
\end{lemma}

\begin{proof}
Along the true nominal lifted dynamics under
\(u_{ci}^{(0)}=-\bm K_{0i}\bm\eta_i\),
\(\dot{\bm\eta}_i
=(\bm A_{K_i}-\bm B_{K_i}\bm K_{0i})\bm\eta_i
+\bm d_{K_i}\).
Using the Riccati equation gives
\(\dot V_{0i}
=-\bm\eta_i^\top\bm Q_{K_i}\bm\eta_i
-u_{ci}^{(0)\top}\bm R_{K_i}u_{ci}^{(0)}
+2\bm\eta_i^\top\bm P_i\bm d_{K_i}\).
Moreover,
\(2\bm\eta_i^\top\bm P_i\bm d_{K_i}
\le
2\|\bm P_i\|\rho_{K_i}\|\bm\eta_i\|^2\),
and hence
\(\dot V_{0i}
\le
-\alpha_{K_i}\|\bm\eta_i\|^2
-u_{ci}^{(0)\top}\bm R_{K_i}u_{ci}^{(0)}
\le
-\alpha_{K_i}\|\bm\eta_i\|^2<0\)
for \(\chi_i\neq0\).

Since \(\Omega_i^0\) is a sublevel set of \(V_{0i}\), it is forward
invariant.
Furthermore,
\(\lambda_{\min}(\bm P_i)\|\bm\eta_i\|^2
\le V_{0i}
\le\lambda_{\max}(\bm P_i)\|\bm\eta_i\|^2\),
so
\(\dot V_{0i}
\le
-\alpha_{K_i}V_{0i}/\lambda_{\max}(\bm P_i)\).
Therefore,
\(\|\chi_i(t)\|
\le M_{0i}e^{-\lambda_{0i}t}\|\chi_i(0)\|\),
where one may take
\(M_{0i}
:=
(\bar c_{\mathcal L_i}/\underline c_{\mathcal L_i})
\sqrt{\lambda_{\max}(\bm P_i)/\lambda_{\min}(\bm P_i)}\)
and
\(\lambda_{0i}
:=
\alpha_{K_i}/[2\lambda_{\max}(\bm P_i)]>0\).

Since
\(0<\alpha_i<1<\beta_i\),
this exponential decay implies
\(\int_0^\infty Q_{ii}(\chi_i(t))dt<\infty\).
Also,
\(u_{ci}^{(0)}=-\bm K_{0i}\bm\eta_i\)
converges exponentially to zero and satisfies
\(u_{ci}^{(0)}\in\mathbb U_i\) on \(\Omega_i^0\).
Because the saturation-compatible input cost is locally smooth and
\(\mathcal U_i(0)=0\), there exists
\(\bar c_{\mathcal U_i}>0\) such that
\(\mathcal U_i(u_{ci}^{(0)})
\le
\bar c_{\mathcal U_i}\|u_{ci}^{(0)}\|^2\)
on the compact control image of \(\Omega_i^0\).
Hence
\(\int_0^\infty\mathcal U_i(u_{ci}^{(0)}(t))dt<\infty\).

Thus,
\(u_{ci}^{(0)}\) is continuous,
\(u_{ci}^{(0)}(0)=0\),
satisfies the input constraint, renders the nominal origin
asymptotically stable while keeping \(\Omega_i^0\) forward invariant,
and yields a finite nominal cost.
Therefore,
\(u_{ci}^{(0)}\in\Psi_i(\Omega_i^0)\).
\end{proof}

Consequently,
\(\Psi_i(\Omega_i^0)\neq\varnothing\).
The exact saturation-compatible policy iteration may therefore be
initialized from
\(u_{ci}^{(0)}\in\Psi_i(\Omega_i^0)\); under the standard admissible
policy-iteration conditions in
\cite{AbuKhalafLewis2005,VamvoudakisLewis2010},
\(u_{ci}^{(k)}\in\Psi_i(\Omega_i^0)\) and
\(V_i^{(k+1)}(\chi_i)\le V_i^{(k)}(\chi_i)\).
This construction only certifies and initializes an admissible HJI
branch; it does not require the finite-dimensional Koopman model to be
exact.
The online critic-only integral Bellman--Isaacs update in the main text
continues to learn directly from the finite-window Bellman residual and
is not interpreted as an exact policy-iteration step.

\bibliographystyle{IEEEtran}
\bibliography{references}

\makeatletter
\def\@IEEEBIOskipN{1.4\baselineskip}
\makeatother

\vspace{0.2em}
\section*{Author Biographies}
\vspace{-0.6em}

\begin{IEEEbiography}
[{\includegraphics[
    width=1in,
    height=1.25in,
    clip,
    keepaspectratio
]{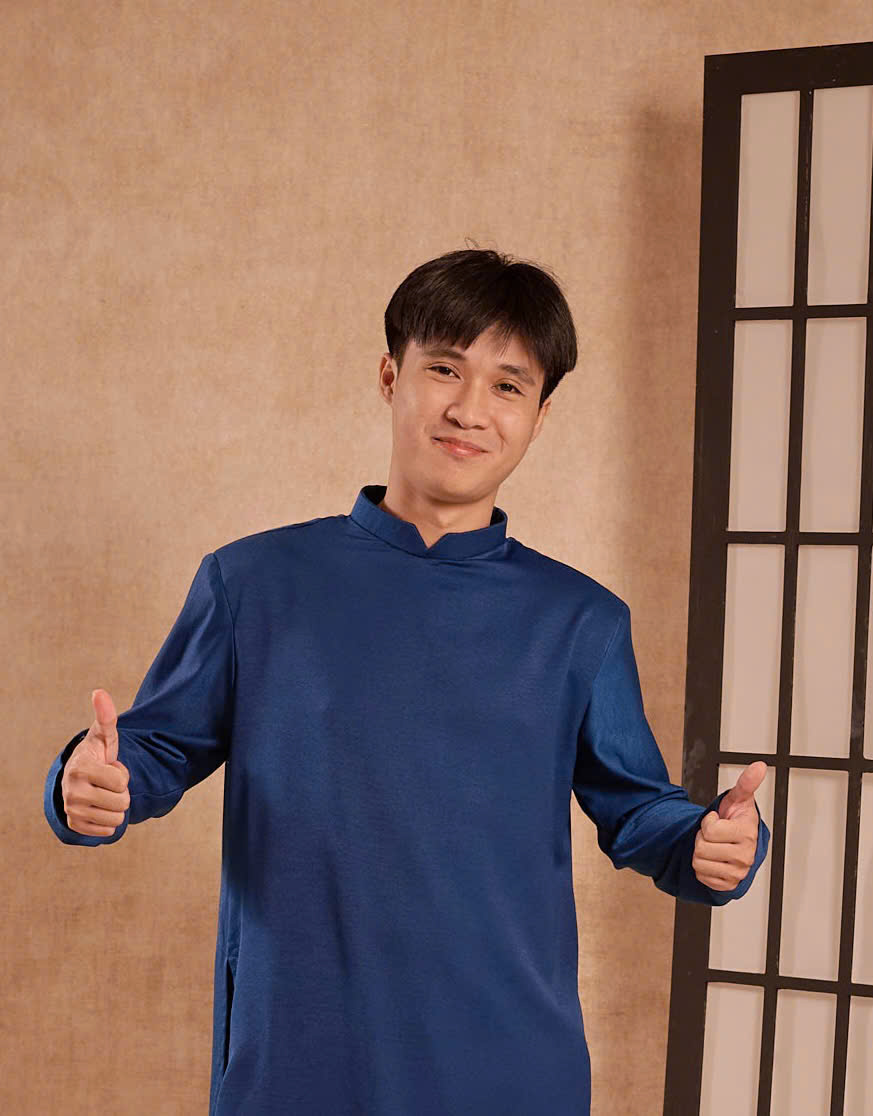}}]
{Tien Dat Vu}
is a senior undergraduate student in the Vietnamese--French Program
in Mechatronics Engineering at Ho Chi Minh City University of
Technology (HCMUT), Vietnam National University Ho Chi Minh City
(VNU-HCM), Ho Chi Minh City, Vietnam. His research interests include
robotics, optimization, formal methods, reinforcement
learning, game theory, model predictive control, fault-tolerant control, data-based systems,
data-driven control, multi-agent systems, stochastic control, and
Riemannian geometric methods for control.
\end{IEEEbiography}

\begin{IEEEbiography}
[{\includegraphics[
    width=1in,
    height=1.25in,
    clip,
    keepaspectratio
]{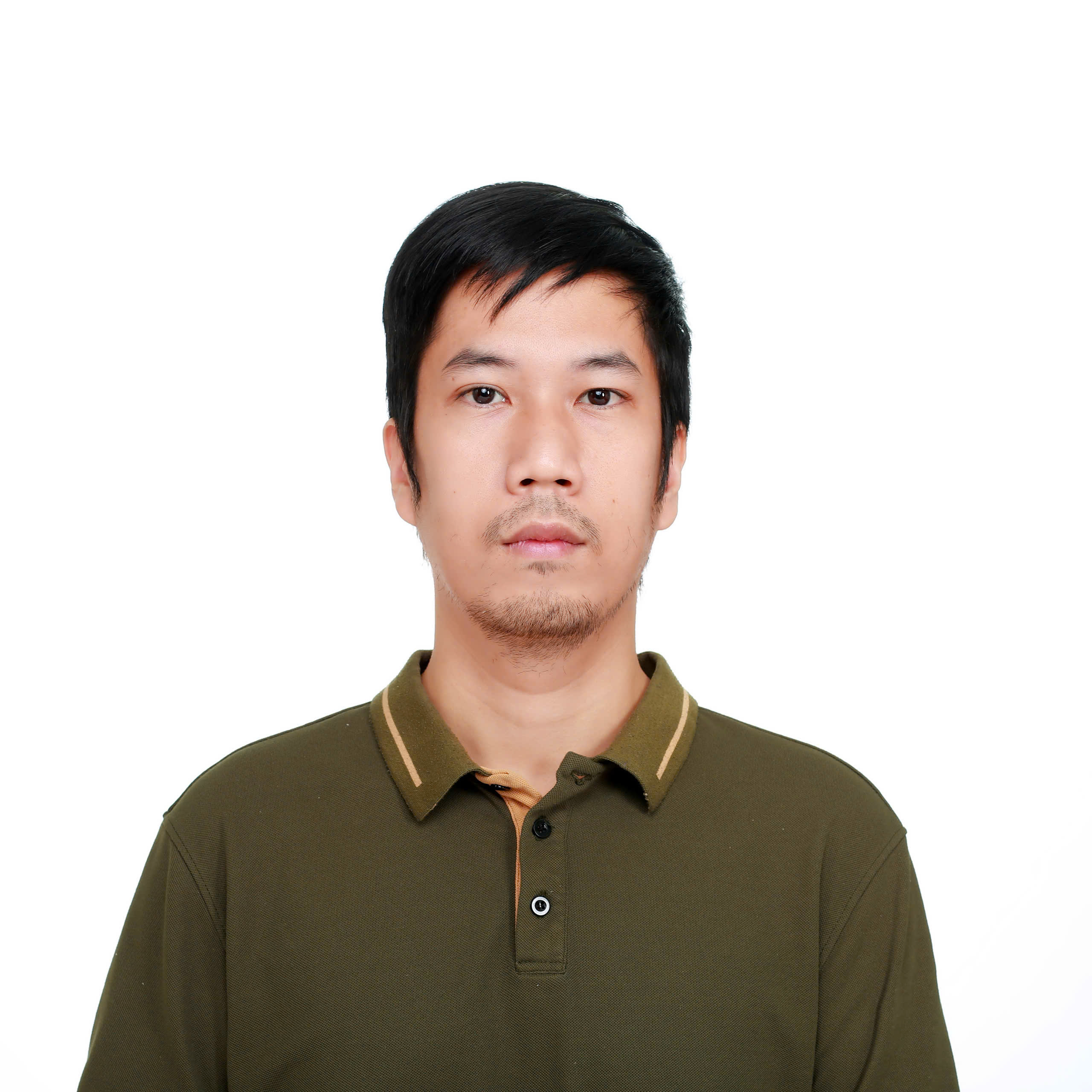}}]
{Minh Doan}
received the B.S. degree in Mechanical Engineering from Bucknell
University, Lewisburg, PA, USA, in 2015, and the Ph.D. degree from
Keio University, Tokyo, Japan, in 2021. He is currently a Lecturer
with the Faculty of Mechanical Engineering, Ho Chi Minh City
University of Technology (HCMUT), Vietnam National University
Ho Chi Minh City (VNU-HCM), Ho Chi Minh City, Vietnam. His research
interests include UAV design, wind turbine systems, multi-agent
systems, and networked control.
\end{IEEEbiography}

\vfill

\end{document}